\documentclass[11pt]{article}
\pdfoutput=1

\usepackage{ifdraft}
\ifdraft{\newcommand{\authnote}[3]{{\color{#3} {\bf  #1:} #2}}}{\newcommand{\authnote}[3]{}}
\newcommand{\anote}[1]{\authnote{ Andras}{#1}{green}}

\newif\ifcount
\ifdraft{\counttrue}

\usepackage[utf8]{inputenc}
\usepackage[T1]{fontenc}
\usepackage{fullpage}
\usepackage[english]{babel}
\usepackage{verbatim}

\usepackage{tikz}
\usepackage{amssymb}
\usepackage{mathtools}
\usepackage{bbm}
\usepackage{mleftright}\mleftright
\usepackage{multirow}
\usepackage{amsthm} 
\theoremstyle{plain}
\usepackage{thmtools} 
\usepackage{thm-restate}

\usepackage{algorithm}
\usepackage{algcompatible}
\usepackage{float}
\usepackage{enumitem}
\makeatletter
\def\namedlabel#1#2{\begingroup
	#2%
	\def\@currentlabel{#2}%
	\phantomsection\label{#1}\endgroup
}

\usepackage{caption}
\usepackage{subcaption}

\usepackage{titling}
\usepackage[final=true,colorlinks = true,allcolors = {blue},]{hyperref}

\newcommand{\eps}{\varepsilon}

\newcommand{\ketbra}[2]{|#1\rangle\! \langle #2|}
\newcommand{\braketbra}[3]{\langle #1|#2| #3 \rangle}
\newcommand{\nrm}[1]{\left\lVert #1 \right\rVert}
\newcommand{\bigO}[1]{\mathcal{O}\left( #1 \right)}

\newcommand{\erf}[1]{\mathrm{erf}\left( #1 \right)}

\newcommand{\erfc}[1]{\mathrm{erfc}\left( #1 \right)}
\newcommand{\sign}[1]{\mathrm{sign}\left( #1 \right)}
\newcommand{\diag}[1]{\mathrm{diag}\left( #1 \right)}
\DeclareMathOperator{\arccosh}{arccosh}

\renewcommand{\H}{\mathcal{H}}
\newcommand{\C}{\mathbb{C}}
\newcommand{\N}{\mathbb{N}}
\newcommand{\R}{\mathbb{R}}
\newcommand{\Z}{\mathbb{Z}}
\newcommand{\PM}{\mathcal{P}}

\newcommand{\cupdot}{\overset{.}{\cup}}
\newcommand{\pvp}{\vec{p}{\kern 0.45mm}'}
\let\oldnabla\nabla
\renewcommand{\nabla}{\oldnabla\!}

\DeclarePairedDelimiter\bra{\langle}{\rvert}
\DeclarePairedDelimiter\ket{\lvert}{\rangle}
\DeclarePairedDelimiterX\braket[2]{\langle}{\rangle}{#1 \delimsize\vert #2}
\newcommand{\underflow}[2]{\underset{\kern-60mm \overbrace{#1} \kern-60mm}{#2}}

\def\Tr{\mathrm{Tr}}

\providecommand{\tr}[1]{\Tr\left[#1\right]}

\providecommand{\rank}[1]{\mathrm{rank}\left(#1\right)}
\providecommand{\poly}[1]{\mathrm{poly}\left(#1\right)}
\providecommand{\img}[1]{\mathrm{img}\left(#1\right)}
\providecommand{\spn}[1]{\mathrm{Span}\left(#1\right)}
\providecommand{\eend}[1]{\mathrm{End}\left(#1\right)}

\long\def\ignore#1{}

\newtheorem{theorem}{Theorem}
\newtheorem{corollary}[theorem]{Corollary}
\newtheorem{lemma}[theorem]{Lemma}

\newtheorem{definition}[theorem]{Definition}

\newtheorem*{claim*}{Claim}

\usepackage{tikz}

\usetikzlibrary{backgrounds,fit,decorations.pathreplacing,calc}

\usepackage{qcircuit}

\title{	
	Quantum singular value transformation and beyond:\\
	exponential improvements for quantum matrix arithmetics
}
\author{
	András Gilyén\thanks{QuSoft, CWI and University of Amsterdam, the Netherlands. Supported by ERC Consolidator Grant 615307-QPROGRESS. \texttt{gilyen@cwi.nl} }
	\and
	Yuan Su\thanks{Department of Computer Science, Institute for Advanced Computer Studies, and Joint Center for Quantum Information and Computer Science, University of Maryland, USA.  \texttt{buptsuyuan@gmail.com}}
	\and
	Guang Hao Low\thanks{Quantum Architectures and Computing group, Microsoft Research, USA.  \texttt{GuangHao.Low@microsoft.com}}
    \and
	Nathan Wiebe\thanks{Quantum Architectures and Computing group, Microsoft Research, USA.  \texttt{nawiebe@microsoft.com}}
}
\date{\today\vspace{-5mm}}

\begin{document}
	
\maketitle
 
  \begin{abstract}
 	Quantum computing is powerful because unitary operators describing the time-evolution of a quantum system have exponential size in terms of the number of qubits present in the system. 
 	We develop a new ``Singular value transformation'' algorithm capable of harnessing this exponential advantage, that can apply polynomial transformations to the singular values of a block of a unitary, generalizing the optimal Hamiltonian simulation results of Low and Chuang~\cite{low2017HamSimUnifAmp}.
 	The proposed quantum circuits have a very simple structure, often give rise to optimal algorithms and have appealing constant factors, while typically only use a constant number of ancilla qubits. 
 	
 	We show that singular value transformation leads to novel algorithms. We give an efficient solution to a ``non-commutative'' measurement problem used for efficient ground-state-preparation of certain local Hamiltonians, and propose a new method for singular value estimation. We also show how to exponentially improve the complexity of implementing fractional queries to unitaries with a gapped spectrum. Finally, as a quantum machine learning application we show how to efficiently implement principal component regression.
 	
 	``Singular value transformation'' is conceptually simple and efficient, and leads to a unified framework of quantum algorithms incorporating a variety of quantum speed-ups. We illustrate this by showing how it generalizes a number of prominent quantum algorithms, and quickly derive the following algorithms: optimal Hamiltonian simulation, implementing the Moore-Penrose pseudoinverse with exponential precision, fixed-point amplitude amplification, robust oblivious amplitude amplification, fast $\mathsf{QMA}$ amplification, fast quantum OR lemma, certain quantum walk results and several quantum machine learning algorithms.

 	In order to exploit the strengths of the presented method it is useful to know its limitations too, therefore we also prove a lower bound on the efficiency of singular value transformation, which often gives optimal bounds.
 	
 \end{abstract}
 
\newpage
\tableofcontents
\newpage

\section{Introduction}

It is often said in quantum computing that there are only a few quantum algorithms that are known to give speed-ups over classical computers.  While this is true, a remarkable number of applications stem from from these primitives.
The first class of quantum speedups is derived from quantum simulation which was originally proposed by Feynman~\cite{feynman1982SimQPhysWithComputers}.  Such algorithms yield exponential speedups over the best known classical methods for simulating quantum dynamics as well as probing electronic structure problems in material science and chemistry.  The two most influential quantum algorithms developed later in the 90's are Shor's algorithm~\cite{shor1994Factoring} (based on quantum Fourier transform) and Grover's search~\cite{grover1996QSearch}.  Other examples have emerged over the years, but arguably quantum walks~\cite{szegedy2004QMarkovChainSearch} and the quantum linear systems algorithm of Harrow, Hassidim and Lloyd~\cite{harrow2009QLinSysSolver} are the most common other primitives that provide speed-ups relative to classical computing.  An important question that remains is whether ``are these primitives truly independent or can they be seen as examples  of a deeper underlying concept?''  The aim of this work is to provide an argument that a wide array of techniques from these disparate fields can all be seen as manifestations of a single quantum idea that we call ``singular value transformation'' generalizing all the above mentioned techniques except for quantum Fourier transform.

Of the aforementioned quantum algorithms, quantum simulation is arguably the most diverse and rapidly developing.  Within the last few years a host of techniques have been developed that have led to ever more powerful methods~\cite{childs2017towardsFirstQSimSpeedup}.  The problem in quantum simulation fundamentally is to take an efficient description of a Hamiltonian $H$, an evolution time $t$, and an error tolerance $\eps$, and find a quantum operation $V$ such that $\nrm{e^{-iHt} -V}\le \eps$ while the implementation of $V$ should use as few resources as possible.  The first methods introduced to solve this problem were Trotter formula decompositions~\cite{lloyd1996UnivQSim,berry2005EffQAlgSimmSparseHam} and subsequently methods based on linear combinations of unitaries were developed~\cite{childs2012HamSimLCU} to provide better asymptotic scaling of the cost of simulation.  

An alternative strategy was also developed concurrent with these methods that used ideas from quantum walks.  Asymptotically, this approach is perhaps the favored method for simulating time-independent Hamiltonians because it is capable of achieving near-optimal scaling with all relevant parameters.  The main tool developed for this approach is a walk operator that has eigenvalues $e^{-i\arcsin(E_k/\alpha)}$ where $E_k$ is the $k^{\rm th}$ eigenvalue of $H$ and $\alpha$ is a normalizing parameter.  While early work adjusted the spectrum recovering the desired eigenvalues $e^{-iE_k}$ by using phase estimation to invert the $\arcsin$, subsequent work achieved better scaling using linear combination of quantum walk steps~\cite{berry2015HamSimNearlyOpt}. Recently another approach, called qubitization~\cite{low2016HamSimQubitization}, was introduced to transform the spectrum in a more efficient manner.

Quantum simulation is not the only field that uses such spectral transformations.  Quantum linear systems algorithms~\cite{harrow2009QLinSysSolver}, as well as algorithms for semi-definite programming~\cite{brandao2016QSDPSpeedup,apeldoorn2017QSDPSolvers},  use these ideas extensively.  Earlier work on linear systems used a strategy similar to the quantum walk simulation method: use phase estimation to estimate the eigenvalues of a matrix $\lambda_j$ and then use quantum rejection sampling to rescale the amplitude of each eigenvector $\ket{\lambda_j}$ via the map $\ket{\lambda_j} \mapsto \lambda_j^{-1} \ket{\lambda_j}$.  This enacts the inverse of a matrix and generalizations to the pseudo-inverse are straightforward.  More recent methods eschew the use of phase estimation in favor of linear-combination of unitary methods~\cite{childs2015QLinSysExpPrec} which typically approximate the inversion using a Fourier-series or Chebyshev series expansion.  Similar ideas can be used to prepare Gibbs states efficiently~\cite{chowdhury2016QGibbsSampling,apeldoorn2017QSDPSolvers}.  

These improvements typically result in exponentially improved scaling in terms of precision in various important subroutines. However, since these techniques work on quantum states, and usually one needs to learn certain properties of these states to a specified precision $\eps$, a polynomial dependence on $\frac{1}{\eps}$ is unavoidable. Therefore these improvements typically ``only'' result in polynomial savings in the complexity. Nevertheless, for complex algorithms this can make a huge difference. These techniques played a crucial role in improving the complexity of quantum semi-definite program solvers~\cite{apeldoorn2018ImprovedQSDPSolving} where the scaling with accuracy was improved from the initial $\mathcal{O}(1/\epsilon^{32})$ to $\mathcal{O}(1/\epsilon^4)$.


We provide a new generalization of qubitization that allows us to view all of the above mentioned applications as a manifestation of a single concept we call \emph{singular value transformation}. The central object for this result is \emph{projected unitary encoding}, which is defined as follows. Suppose that $\widetilde{\Pi}$,$\Pi$ are orthogonal projectors and $U$ is a unitary, then we say that the unitary $U$ and the projectors $\widetilde{\Pi}$,$\Pi$ form a projected unitary encoding of the operator $A:=\widetilde{\Pi} U \Pi$. 
We define singular value transformation by a polynomial $P\in\C[x]$ in the following way: if $P$ is an odd polynomial and $A=W\Sigma V^\dagger$ is a singular value decomposition (SVD), then $P^{(SV)}(A):=W P(\Sigma) V^\dagger$, where the polynomial $P$ is applied to the diagonal entries of $\Sigma$.
Our main result is that for any degree-$d$ odd polynomial $P\in\R[x]$, that is bounded by $1$ in absolute value on $[-1,1]$, we can implement a unitary $U_\Phi$ with a simple circuit using $U$ and its inverse a total number of $d$ times such that 
$$ A=\widetilde{\Pi} U \Pi \Longrightarrow P^{(SV)}(A)=\widetilde{\Pi} U_\Phi \Pi.$$
We prove a similar result for even polynomials as well, but with replacing $\widetilde{\Pi}$ by $\Pi$ in the above equation, and defining $P^{(SV)}(A):=V P(\Sigma) V^\dagger$ for even polynomials. One can view these results as generalizations of the quantum walk techniques introduced by Szegedy~\cite{szegedy2004QMarkovChainSearch}.

In order to illustrate the power of this technique we briefly explain some corollaries of this result. For example suppose that $U$ is a quantum algorithm that on the initial state $\ket{0}^{\otimes n}$ succeeds with probability at least $p$, and indicates success by setting the first qubit to $\ket{1}$. Then we can take $\widetilde{\Pi}:=\ketbra{1}{1}\otimes I_{n-1}$ and $\Pi:= \ketbra{0}{0}^{\otimes n}$. Observe that $A=\widetilde{\Pi} U \Pi$ is a rank-$1$ matrix having a single non-trivial singular value being the square root of the success probability. If $P$ is an odd polynomial bounded by $1$ in absolute value such that $P$ is $\frac{\eps}{2}$-close to $1$ on the interval $[\sqrt{p},1]$, then by applying singular value transformation we get an algorithm $U_\Phi$ that succeeds with probability at least $1-\eps$. Such a polynomial can be constructed with degree $\bigO{\frac{1}{\sqrt{p}}\log\left(\frac{1}{\eps}\right)}$ providing a conceptually simple and efficient implementation of fixed-point amplitude amplification.

It also becomes straightforward to implement the Moore-Penrose pseudoinverse directly. Suppose that $A=W\Sigma V^\dagger$ is an SVD, then the pseudoinverse is simply $A^+=V\Sigma^{-1} W^\dagger$, where we take the inverse of each non-zero diagonal element of $\Sigma$. If we have $A$ represented as a projected unitary encoding, then simply finding an appropriately scaled approximation polynomial of $\frac{1}{x}$ and applying singular value transformation to it implements an approximation of the Moore-Penrose pseudoinverse directly. As an application in quantum machine learning, we design a quantum algorithm for \emph{principal component regression}, and argue that singular value transformation could become a central tool in designing quantum machine learning algorithms.

Based on singular value transformation we develop two main general results: \emph{singular vector transformation}, which maps right singular vectors to left singular vectors, and \emph{singular value threshold projectors}, which project out singular vectors with singular value below a certain threshold. These threshold projectors play a major role in quantum algorithms recently proposed by Kerenidis et al.~\cite{kerenidis2016QRecSys,kerenidis2018SlowFeatureAnalysis}, and our work fills a minor gap that was present in earlier implementation proposals.  Our implementation is also simpler and applies in greater generality than the algorithm of Kerenidis and Prakash~\cite{kerenidis2016QRecSys}.  As a useful application of singular value threshold projectors we develop \emph{singular value discrimination}, which decides whether a given quantum state has singular value below or above a certain threshold. As another application we show that using singular vector transformation one can efficiently implement a form of ``non-commutative measurement'' which is used for preparing ground states of local Hamiltonians. Also we propose a new method for \emph{quantum singular value estimation} introduced by~\cite{kerenidis2016QRecSys}. 

Other algorithms can also be cast in the singular value transformation framework, including optimal Hamiltonian simulation, robust oblivious amplitude amplification, fast $\mathsf{QMA}$ amplification, fast quantum OR lemma and certain quantum walk results. 
Based on these techniques we also show how to exponentially improve the complexity of implementing fractional queries to unitaries with a gapped spectrum. We summarize in Table~\ref{tab:speedUpTypes} the various types of quantum speed-ups that are inherently incorporated in our singular value transformation framework. 

\begin{table}[H]
	\centering
	\begin{tabular}{l|ll}
		Speed-up & Source of speed-up & Examples of algorithms 
		\\ \hline
		\multirow{2}{*}{Exponential} & Dimensionality of the Hilbert space &  Hamiltonian simulation~\cite{lloyd1996UnivQSim} \\
		&  Precise polynomial approximations  &  Improved HHL algorithm~\cite{childs2015QLinSysExpPrec} \\ \hline
		\multirow{3}{*}{Quadratic}    &  Singular value $=$ square root of probability  &  Grover search~\cite{grover1996QSearch} \\
		&  Distinguishability of singular values & Amplitude estimation~\cite{brassard2002AmpAndEst}  \\
		&  Singular values close to $1$ are more useful & Quantum walks~\cite{szegedy2004QMarkovChainSearch}
	\end{tabular}
	\caption{This table gives an intuitive summary of the different types of speed-ups that our singular value transformation framework inherently incorporates. The explanations, examples and the cited papers are far from being complete or representative, the table only intends to give some intuition and illustrate the different sources of speed-ups.}
	\label{tab:speedUpTypes}	
\end{table}

In order to harness the power of singular value transformation one needs to construct projected unitary encodings. A special case of projected unitary encoding is called \emph{block-encoding}, when $\widetilde{\Pi}=\Pi=\ketbra{0}{0}^{\otimes a}\otimes I$. In this case $A$ is literally the top-left block of the unitary $U$. In the paper we provide a versatile toolbox for efficiently constructing block-encodings, summarizing recent developments in the field. In particular we demonstrate how to construct block-encodings of unitary matrices, density operators, POVM operators, sparse-access matrices and matrices stored in a QROM\footnote{By QROM we mean quantum read-only memory, which stores classical data that can be accessed in superposition.}. Furthermore, we show how to form linear combinations and products of block-encodings. 

\subsection{Structure of the paper}

In Section~\ref{sec:QSing} we derive a new formalization of qubitization that allows us to view all of the aforementioned applications as a manifestation of a single concept we call ``singular value transformation''. 
In Subsection~\ref{subsec:QSignProc} we develop a slightly improved version of the quantum signal processing result of Low et al.~\cite{low2016CompositeQuantGates}.  In Subsection~\ref{subsec:SingularvalTrans} we develop our singular value transformation result based on qubitization ideas of Low and Chuang~\cite{low2016HamSimQubitization}.  In Subsection~\ref{subsec:robustness} we prove bounds about the robustness of singular value transformation. We then introduce singular vector transformation and singular value amplification in Subsection~\ref{subsec:SingularvecTrans}, from which we provide elementary derivations of fixed-point amplitude amplification and robust oblivious amplitude amplification.  We then extend these ideas in Subsection~\ref{subsec:singWalk} to solve the problem of singular value threshold projection and singular value discrimination which as we show allow us to detect and find marked elements in a reversible Markov chain.  These ideas then allow us to provide an easy derivation of the quantum linear-systems algorithm, and more generally the quantum least-squares fitting algorithm, in Subsection~\ref{subsec:NonCommMeasAndSVE}.  In Subsection~\ref{subsec:DirectPseudoInverse} we show how to implement a form of ``non-commutative measurement'' which is used for preparing ground states of local Hamiltonians, and propose a new method for quantum singular value estimation. Finally, in Subsection~\ref{subsec:QML}, we design a quantum algorithm for principal component regression, and show how various other machine learning problems can be solved within our framework.

Section \ref{sec:matArith} shows how to efficiently construct block-encodings and contains a discussion of how these techniques can be employed to perform matrix arithmetic on a quantum computer.  In particular we show how to perform basic linear algebra operations on Hamiltonians using block-encodings; we discuss matrix addition and multiplication in Subsections~\ref{subsec:add} and~\ref{subsec:mult}.  We follow this 
up with a discussion of how arbitrary smooth functions of Hermitian matrices can be performed.  We then give an elementary proof of the complexity of block-Hamiltonian simulation in Subsection~\ref{subsec:opt} and discuss approximating piecewise smooth functions of Hamiltonians in Subsection~\ref{subsec:polyApprox} and present the special cases of Gibbs state preparation and fractional queries in Subsection~\ref{subsec:applications}.  We then conclude by proving lower bounds for implementing functions of Hermitian matrices in Section~\ref{sec:lowerBd}, which in turn implies lower bounds on singular value transformation.

\section{Preliminaries and notation}\label{sec:prelims}

It is well known that for every $A\in\C^{m\times n}$ matrix there exists a pair of unitaries $W\in\C^{m\times m}$, $V\in\C^{n\times n}$ and $\Sigma\in\R^{m\times n}$ such that $\Sigma$ is a diagonal matrix with non-negative non-increasing entries on the diagonal, and $A=W\Sigma V^\dagger$. Such a decomposition is called \emph{singular value decomposition}. Let $k:=\min[m,n]$, then we use $\varsigma_i:=\Sigma_{ii}$ for $i\in[k]$ to denote the \emph{singular values} of $A$, which are the diagonal elements of $\Sigma$. The columns of $V$ are called right singular vectors, and the columns of $W$ are called the left singular vectors. In this paper we often define the matrix $A$ as the product of two orthogonal projectors $\widetilde{\Pi},\Pi$ and unitary $U$ such that $A=\widetilde{\Pi}U\Pi$. In such a case we will assume without loss of generality that the first $\text{rank}(\widetilde{\Pi})$ left singular vectors span $\text{img}(\widetilde{\Pi})$ and the first $\rank{\Pi}$ right singular vectors span $\img{\Pi}$. 

The singular value decomposition is not unique if there are multiple singular values with the same value. However, the singular value projectors are uniquely determined, see, e.g., Gilyén and Sattath~\cite{gilyen2016PrepGapHamEffQLLL}.
\begin{definition}[Singular value projectors]\label{def:sinValProj}
	Let $A=W\Sigma V^\dagger$ be a singular value decomposition. 
	Let $\Sigma_{\varsigma}$ be the matrix that we get from $\Sigma$ by replacing all singular values that have value $= \varsigma$ by $1$ and replacing all $\neq \varsigma$ singular values by $0$. Then we define the right singular value projector to singular value $\varsigma$ as $V \Sigma_{\varsigma} V^\dagger$, and define the left singular value projector to singular value $\varsigma$ as $W \Sigma_{\varsigma} W^\dagger$ projecting orthogonally to the subspace spanned by the corresponding singular vectors. For a set $S\subset \R$ we similarly define the right and left singular value projectors $V \Sigma_{S} V^\dagger$, $W \Sigma_{S} W^\dagger$ projecting orthogonally to the subspace spanned by the singular vectors having singular value in $S$.
\end{definition}

In this paper we will work with polynomial approximations, and therefore we introduce some related notation. For a function $f:I\rightarrow \C$ and a subset $I'\subseteq I$ we use the notation $\nrm{f}_{I'}:=\sup_{x\in I'}|f(x)|$ to denote the sup-norm of the function $f$ on the domain $I'$. We say that a function $f\colon \R\rightarrow \C$ \emph{is even} if for all $x\in\R$ we have $f(-x)=f(x)$, and that it \emph{is odd} if for all $x\in\R$ we have $f(-x)=-f(x)$.

Let $P\in\C[x]$ be a complex polynomial $P(x)=\sum_{j=0}^{k}a_jx^j$, then we denote by $P^*(x):=\sum_{j=0}^{k}a^*_jx^j$ the polynomial with conjugated coefficients, and let $\Re[P](x):=\sum_{j=0}^{k}\Re[a_j]x^j$ denote the real polynomial we get by taking the real part of the coefficients. We say that $P$ \emph{is even} if all coefficients corresponding to odd powers of $x$ are $0$, and similarly we say that $P$ \emph{is odd} if all coefficients corresponding to even powers of $x$ are $0$. For an integer number $z\in \Z$ we say that $P$ has parity $z$ if $z$ is even and $P$ is even or $z$ is odd and $P$ is odd. We will denote by $T_d\in\R[x]$ the $d$-th Chebyshev polynomial of the first kind, defined by $T_d(x):=\cos(d \arccos(x))$.

Whenever we present a matrix and put a $.$ in some place we mean a matrix with arbitrary values of the elements in the unspecified block. For example $[.]$ just denotes a matrix with completely arbitrary elements, similarly 
$$U=\left[\begin{array}{cc} A & .\\
. & .\end{array}\right]
$$
denotes an arbitrary matrix whose top-left block is $A$.

For an orthogonal projector $\Pi$ we will frequently use the $\Pi$-controlled NOT gate, denoted by C$_\Pi$NOT, which implements a coherent measurement operator by flipping the value of a qubit based on whether the state of a register is in the image of $\Pi$ or not. For example if $\Pi=\ketbra{1}{1}$, then we just get back the usual CNOT gate controlled by the second qubit.
\begin{definition}[C$_\Pi$NOT gate]\label{def:PiNOT}
	For an orthogonal projector $\Pi$ let us define the $\Pi$-controlled NOT gate as the unitary operator
	$$\mathrm{C}_\Pi\mathrm{NOT}:=X\otimes\Pi+I\otimes(I-\Pi) .$$
\end{definition}

\section{Qubitization and Singular value transformations}
\label{sec:QSing}

The methods in this section are based on the so called ``Quantum Signal Processing'' techniques introduced by Low, Yoder and Chuang~\cite{low2016CompositeQuantGates}.
In Section~\ref{subsec:QSignProc} we present a self-contained treatment of these techniques, significantly streamline the formalism, and develop slightly improved versions of the results presented in~\cite{low2016CompositeQuantGates}. As a corollary of the results we also develop Corollary~\ref{cor:refAchievableP}-\ref{cor:realP}, which will be the only results that we need in the rest of the paper. We suggest the first-time reader to skip the proofs in Section~\ref{subsec:QSignProc}, as they are not necessary in order to understand the main ideas of Sections~\ref{subsec:SingularvalTrans}-\ref{subsec:NonCommMeasAndSVE}.

In Sections~\ref{subsec:SingularvalTrans} we show how to leverage the results of Section~\ref{subsec:QSignProc} to perform \emph{singular value transformation} of projected unitary matrices, with ideas coming from ``qubitization''~\cite{low2016HamSimQubitization}.
Singular value transformation is a common generalization of the techniques developed around qubitization, based on which we can quickly derive a host of well-optimized applications in Sections~\ref{subsec:SingularvecTrans}-\ref{subsec:QML}.
\subsection{Parametrized SU(2) unitaries induced by Pauli rotations} 
\label{subsec:QSignProc} 

In this section we review the results of Low, Yoder and Chuang~\cite{low2016CompositeQuantGates}, who show how to build $2\times2$ unitary matrices whose entries are trigonometric polynomials by taking products of various rotation and phase gates. They consider essentially the following problem, which they call ``Quantum Signal Processing'': suppose one can apply a gate sequence 
\begin{equation}\label{eq:rotAndPhase}
e^{i\phi_0 \sigma_z}e^{i\theta \sigma_x}
e^{i\phi_1 \sigma_z}e^{i\theta \sigma_x}
e^{i\phi_2 \sigma_z}\cdot\ldots\cdot
e^{i\theta \sigma_x}e^{i\phi_k \sigma_z},
\end{equation}
where $\theta$ is unknown (they call $e^{i\theta \sigma_x}$ the signal unitary) but one has control over the angles $\varphi_0,\varphi_1, \ldots, \varphi_k$; which unitary operators can we build this way? They give a characterization of the unitary operators that can be constructed this way, and find that the set of achievable unitary operators is quite rich.

We find it more useful to work with the above matrices using a slightly modified parametrization. For $x\in[-1,1]$ let us define
\begin{equation*}
W(x):=\left[\begin{array}{cc} x & i\sqrt{1-x^2} \\ i\sqrt{1-x^2} & x \end{array}\right]=e^{i\arccos(x) \sigma_x}.
\end{equation*}

It is easy to see that if $\theta\in[0,\pi]$, then by setting $x:=\cos(\theta)$ Eq.~\eqref{eq:rotAndPhase} can be rewritten as 
\begin{equation}\label{eq:sandwitchedForm}
e^{i\phi_0 \sigma_z}
W(x)
e^{i\phi_1 \sigma_z}
W(x)
e^{i\phi_2 \sigma_z}
\cdot\ldots\cdot
W(x)
e^{i\phi_k \sigma_z}.
\end{equation}

Now we present the characterization of Low et al.~\cite{low2016CompositeQuantGates} using the above formalism. Our formulation makes the statement simpler and reduces the number of cases. We also present a succinct simplified proof which can be conveniently described using our formalism.

\begin{theorem}\label{thm:basicCharacterisation}
	Let $k\in\N$; there exists $\Phi=\{\phi_0,\phi_1,\ldots,\phi_k\}\in\R^{k+1}$ such that for all $x\in[-1,1]\colon$
	\begin{equation}\label{eq:sandwitchedThm}		
	e^{i\phi_0 	\sigma_z}\prod_{j=1}^{k}\left(W(x)e^{i\phi_j \sigma_z}\right)
	=		\left[\begin{array}{cc} P(x) & iQ(x)\sqrt{1-x^2} \\ iQ^*(x)\sqrt{1-x^2} & P^*(x) \end{array}\right]
	\end{equation}
	if and only if $P,Q\in\C[x]$ such\footnote{Note that the value of $P(x)$ is only determined for $x\in[-1,1]$ and $Q(x)$ for $x \in(-1,1)$; thus more precisely we should talk about the polynomial functions induced by $\left.P(x)\right|_{[-1,1]}\in \C[x]$ and $\left.Q(x)\right|_{(-1,1)}\in \C[x]$.} that 
	\begin{itemize}
		\item[\namedlabel{prop:i}{(i)}] $\deg(P)\leq k$ and $\deg(Q)\leq k-1$
		\item[\namedlabel{prop:ii}{(ii)}] $P$ has parity-$(k \mod 2)$ and $Q$ has parity-$(k-1 \mod 2)$	
		\item[\namedlabel{prop:iii}{(iii)}] $\forall x\in[-1,1]\colon |P(x)|^2+(1-x^2)|Q(x)|^2=1$.			
	\end{itemize} 
\end{theorem}
\begin{proof}
	``$\Longrightarrow$'': For the $k=0$ case the unitary on the left hand side of \eqref{eq:sandwitchedThm} is $e^{i\phi_0 \sigma_z}$, so that $P\equiv e^{i\phi_0}$ and $Q\equiv 0$ satisfy  the properties \ref{prop:i}-\ref{prop:iii}. Now we prove \ref{prop:i}-\ref{prop:ii} by induction. The induction step can be shown as follows: suppose we have proved for $k-1$ that
	\begin{equation*}	
	e^{i\phi_0 	\sigma_z}\prod_{j=1}^{k-1}\left(W(x)e^{i\phi_j \sigma_z}\right)=
	\left[\begin{array}{cc} \tilde{P}(x) & i\tilde{Q}(x)\sqrt{1-x^2} \\ i\tilde{Q}^*(x)\sqrt{1-x^2} & \tilde{P}^*(x) \end{array}\right],
	\end{equation*}	
	where $\tilde{P},\tilde{Q}\in\C[x]$ satisfy \ref{prop:i}-\ref{prop:ii}. Then
	\begin{align}	
	e^{i\phi_0 	\sigma_z}\prod_{j=1}^{k}\left(W(x)e^{i\phi_j \sigma_z}\right)
	\!\!&=\!
	\left[\begin{array}{cc} \tilde{P}(x) &\!\! i\tilde{Q}(x)\sqrt{1-x^2} \!\\\! i\tilde{Q}^*(x)\sqrt{1-x^2} &\!\! \tilde{P}^*(x) \end{array}\right]
	\left[\begin{array}{cc} e^{i\phi_{k}}x &\!\! ie^{-i\phi_{k}}\sqrt{1-x^2} \!\\\! ie^{i\phi_{k}}\sqrt{1-x^2} &\!\! e^{-i\phi_{k}}x \end{array}\right]\nonumber\\
	&=
	\Bigg[\begin{array}{cc} 
	\overset{P(x):=}{\overbrace{e^{i\phi_{k}}\left(x\tilde{P}(x) + (x^2-1)\tilde{Q}(x)\right)}}
	& ie^{-i\phi_{k}}\left(x\tilde{Q}(x)+\tilde{P}(x)\right)\sqrt{1-x^2} \\ 
	i\underset{Q^*(x):=}{\underbrace{e^{i\phi_{k}}\left(x\tilde{Q}^*(x)+\tilde{P}^*(x)\right)}}\sqrt{1-x^2} 
	& e^{-i\phi_{k}}\left(x\tilde{P}^*(x) + (x^2-1)\tilde{Q}^*(x)\right)
	\end{array}\Bigg],
	\end{align}			
	and it is easy to see that $P,Q$ satisfy \ref{prop:i}-\ref{prop:ii}.
	Finally note that the left hand side of \eqref{eq:sandwitchedThm} is a product of unitaries, therefore the right hand side is unitary too, which implies \ref{prop:iii}.
	
	``$\Longleftarrow$'': Suppose $P,Q$ satisfy \ref{prop:i}-\ref{prop:iii}. First we handle a trivial case: suppose that $\deg(P)=0$, then due to \ref{prop:iii} we must have that $|P(1)|=1$ and thus $P\equiv e^{i\phi_0}$ for some $\phi_0\in\R$. This again using \ref{prop:iii} implies that $Q\equiv 0$. Due to \ref{prop:ii} we must have that $k$ is even, and thus $\Phi=(\phi_0,\frac\pi2,-\frac\pi2,\ldots,\frac\pi2,-\frac\pi2)\in\R^{k+1}$ is a solution, since 
	\begin{align*}
	e^{i\phi_0 	\sigma_z}\prod_{j=1}^{k/2}\left(W(x)e^{i\frac\pi2 \sigma_z}W(x)e^{-i\frac\pi2 \sigma_z}\right)
	=e^{i\phi_0 	\sigma_z}
	=\left[\begin{array}{cc} e^{i\phi_0} & 0 \\ 0 & e^{-i\phi_0} \end{array}\right].
	\end{align*}
	This special case also covers the $k=0$ case, providing the base of our induction. 
	
	Now we show the induction step, assuming that we proved the claim for $k-1$. Note that \ref{prop:iii} can be rewritten as 
	\begin{equation}\label{eq:cancel}
	\forall x\in[-1,1]\colon P(x)P^*(x)+(1-x^2)Q(x)Q^*(x)=1.
	\end{equation}
	Since this equation holds for infinitely many points, the polynomial on the right hand side of \eqref{eq:cancel} must be the constant $\equiv1$ polynomial. Assume without loss of generality that $1\leq\deg(P)=\ell\leq k$, then we must have that $\deg(Q)=\ell-1$, and $|p_{\ell}|=|q_{\ell-1}|$, since the highest order terms cancel each other in \eqref{eq:cancel}. Let $\phi_k\in\R$ be such that $e^{2i\phi_k}=\frac{p_{\ell}}{q_{\ell-1}}$, and let us define
	\begin{align}	
	\!\!\left[\begin{array}{cc} \tilde{P}(x) &\kern-5mm i\tilde{Q}(x)\sqrt{1-x^2}\kern-1mm \\
	\kern-2mm i\tilde{Q}^*(x)\sqrt{1-x^2} &\kern-5mm \tilde{P}^*(x) \end{array}\right]
	\!\!:=&
	\left[\begin{array}{cc} P(x) &\kern-5mm iQ(x)\sqrt{1-x^2}\kern-1mm \\ \kern-2mm iQ^*(x)\sqrt{1-x^2} &\kern-5mm P^*(x) \end{array}\right]e^{-i\phi_k\sigma_z}W^\dagger(x)\!\nonumber\\		
	\!\!=&
	\left[\begin{array}{cc} P(x) &\kern-5mm iQ(x)\sqrt{1-x^2}\kern-1mm \\ \kern-2mm iQ^*(x)\sqrt{1-x^2} &\kern-5mm P^*(x) \end{array}\right]
	\!\left[\begin{array}{cc} e^{-i\phi_k}x &\kern-5mm -ie^{-i\phi_k}\sqrt{1-x^2} \kern-1mm \\ \kern-2mm -ie^{i\phi_k}\sqrt{1-x^2} &\kern-5mm e^{i\phi_k}x \end{array}\right]\!\nonumber\\
	\!\!=&
	\Bigg[\begin{array}{cc} 
	\overset{\tilde{P}(x):=}{\overbrace{ e^{-i\phi_k}xP(x)+e^{i\phi_k}(1-x^2)Q(x)}}
	& i\tilde{Q}(x)\sqrt{1-x^2} \\ 
	i\underset{\tilde{Q}^*(x):=}{\underbrace{\left(e^{-i\phi_k}xQ^*(x)-e^{i\phi_k}P^*(x)\right)}}\sqrt{1-x^2} 
	& \tilde{P}^*(x)
	\end{array}\Bigg]\label{eq:inducted}
	\end{align}	
	where 
	\begin{equation}\label{eq:Ptilde}
	\tilde{P}(x)= e^{-i\phi_k}xP(x)+e^{i\phi_k}(1-x^2)Q(x)=e^{-i\phi_k}\left(xP(x)+\frac{p_{\ell}}{q_{\ell-1}}(1-x^2)Q(x)\right)
	\end{equation} and
	\begin{equation}\label{eq:Qtilde}
	\tilde{Q}(x)=e^{i\phi_k}xQ(x)-e^{-i\phi_k}P(x)=e^{-i\phi_k}\left(\frac{p_{\ell}}{q_{\ell-1}}xQ(x)-P(x)\right).
	\end{equation} 
	It is easy to see that the highest order terms in \eqref{eq:Ptilde}-\eqref{eq:Qtilde} cancel out, and therefore $\deg(\tilde{P})\leq \ell-1 \leq k-1$, $\deg(\tilde{Q})\leq \ell-2\leq k-2$. Using \eqref{eq:Ptilde}-\eqref{eq:Qtilde} we can also verify that $\tilde{P},\tilde{Q}$ satisfy \ref{prop:i}-\ref{prop:ii} regarding $k-1$, moreover property \ref{prop:iii} is preserved due to unitarity. So by the induction hypothesis we get that \eqref{eq:inducted} equals $e^{i\tilde{\phi}_0 	\sigma_z}\left(\prod_{j=1}^{k-1}W(x)e^{i\tilde{\phi}_j \sigma_z}\right)$ for some $\tilde{\Phi}\in\R^{k}$, therefore $\Phi:=(\tilde{\phi}_0,\tilde{\phi}_1,\tilde{\phi}_2,\ldots,\tilde{\phi}_{k-1},\phi_k)\in\R^{k+1}$ is a solution.
\end{proof}

\newpage
Note that the above proof also gives an algorithm that finds $\Phi$ using $\bigO{k^2}$ arithmetic operations. The following two characterizations and their proofs also follow a constructive approach which can be translated to a polynomial time algorithm. However, they have the drawback that they rely on finding roots of high-degree polynomials,\footnote{For a good bound on the complexity of approximate root finding see, e.g., the work of Neff and Reif~\cite{neff1996AlgComplexRootFinding}.} which makes it harder in practice to execute the resulting protocols.

\begin{theorem}\label{thm:achievablePorQ}
	Let $k\in\N$ be fixed. Let $P\in\C[x]$, there exists some $Q\in\C[x]$ such that $P,Q$ satisfy properties \ref{prop:i}-\ref{prop:iii} of Theorem~\ref{thm:basicCharacterisation} if and only if $P$ satisfies properties \ref{prop:i}-\ref{prop:ii} of Theorem~\ref{thm:basicCharacterisation} and 
	\begin{itemize}
		\item[\namedlabel{prop:iva}{(iv.a)}] $\forall x\in[-1,1]\colon |P(x)|\leq 1$
		\item[\namedlabel{prop:ivb}{(iv.b)}] $\forall x\in(-\infty,-1]\cup[1,\infty)\colon |P(x)|\geq 1$		
		\item[\namedlabel{prop:ivc}{(iv.c)}]if $k$ is even, then $\forall x\in\R\colon P(ix)P^*(ix)\geq 1$.
	\end{itemize} 
	
	Similarly, let $Q\in\C[x]$, there exists some $P\in\C[x]$ such that $P,Q$ satisfy properties \ref{prop:i}-\ref{prop:iii} of Theorem~\ref{thm:basicCharacterisation} if and only if $Q$ satisfies properties \ref{prop:i}-\ref{prop:ii} of Theorem~\ref{thm:basicCharacterisation} and 
	\begin{itemize}
		\item[\namedlabel{prop:va}{(v.a)}] $\forall x\in[-1,1]\colon \sqrt{1-x^2}|Q(x)|\leq 1$
		\item[\namedlabel{prop:vb}{(v.b)}] if $k$ is odd, then $\forall x\in\R\colon (1+x^2)Q(ix)Q^*(ix)\geq 1$.
	\end{itemize} 
\end{theorem}
\begin{proof}``$\Longrightarrow$'': Trivially follows from \ref{prop:iii}: $$\forall x\in\C\colon P(x)P^*(x)+(1-x^2)Q(x)Q^*(x)=1.$$
	``$\Longleftarrow$'':
	First consider the case when $k$ is odd, and consider the polynomial $A(x):=1-P(x)P^*(x)$. Note that $A\in\R[x]$ and $A$ is even, therefore $A$ is in fact a polynomial in $x^2$. Let $y=x^2$ and consider the real polynomial $\tilde{A}(y):=A(\sqrt{y})$. Observe that $\forall y\geq 1\colon\tilde{A}(y)\leq 0$ due to \ref{prop:ivb}, $\forall y \in [0,1]\colon\tilde{A}(y)\geq 0$ due to \ref{prop:iva} and $\forall y \leq 0\colon\tilde{A}(y)\geq 1$ since
	\begin{align*}
	\tilde{A}(y)&=A(i\sqrt{-y})\tag{$y \leq 0$ }\\
	&=1-P(i\sqrt{-y})P^*(i\sqrt{-y})
	=1+P(i\sqrt{-y})P^*(-i\sqrt{-y})\tag{$P$ is odd }\\
	&=1+P(i\sqrt{-y})(P(i\sqrt{-y}))^*
	=1+|P(i\sqrt{-y})|^2
	\geq 1.
	\end{align*} 
	Therefore all real roots have even multiplicity except for $1$, moreover all complex roots come in pairs. Thus $\tilde{A}(y)=(1-y)K^2\prod_{s\in S}(y-s)(y-s^*)$ for some $K\in\R$ and $S\subseteq\C$ multiset of roots. Let $W(y):=K\prod_{s\in S}(y-s)\in\C[y]$, then $\tilde{A}(y)=(1-y)W(y)W^*(y)$, and thus $A(x)=(1-x^2)W(x^2)W^*(x^2)$, i.e., $1=P(x)P^*(x)+(1-x^2)W(x^2)W^*(x^2)$. Setting $Q(x):=W(x^2)$ concludes this case.
	
	The other cases can be proven similarly, by examining the polynomial $1-P(x)P^*(x)$ or $1-(1-x^2)Q(x)Q^*(x)$ respectively.
\end{proof}

The original proof of the next theorem in \cite{low2016CompositeQuantGates} used the Weierstrass substitution, which made it difficult to understand, and made it hard to analyze the numerical stability of the induced algorithm. Also the theorem was stated in a slightly less general form requiring $\Re[\tilde{P}](1)=1$. We roughly follow the approach of \cite{low2016CompositeQuantGates}, but improve all of the mentioned aspects of the theorem and its proof, while making the statement and the proof conceptually simpler.

\begin{theorem}\label{thm:cleverSOS}
	Let $k\in\N$ be fixed. Let $\tilde{P},\tilde{Q}\in\R[x]$, there exists some $P,Q\in\C[x]$ satisfying properties \ref{prop:i}-\ref{prop:iii} of Theorem~\ref{thm:basicCharacterisation} such that $\tilde{P}=\Re[P]$, $\tilde{Q}=\Re[Q]$, if and only if $\tilde{P},\tilde{Q}$ satisfy properties \ref{prop:i}-\ref{prop:ii} of Theorem~\ref{thm:basicCharacterisation} and 
	\begin{itemize}
		\item[\namedlabel{prop:vi}{(vi)}] $\forall x\in[-1,1]\colon \tilde{P}(x)^2+(1-x^2)\tilde{Q}^2\leq 1$.
	\end{itemize} 
	(Note that the same holds if we replace $\Re[P]$ by $\Im[P]$ and/or $\Re[Q]$ by $\Im[Q]$ in the statement. Moreover we may set $\tilde{Q}\equiv 0$ or $\tilde{P}\equiv 0$ if we are only interested in $\tilde{P}$ or $\tilde{Q}$.)
\end{theorem}
\begin{proof}``$\Longrightarrow$'': Trivial.\\
	``$\Longleftarrow$'': Apply Lemma~\ref{lemma:cleverSOS} to the polynomial $1-\tilde{P}(x)^2-(1-x^2)\tilde{Q}(x)^2$, and set $P:=\tilde{P}+iB$, $Q:=\tilde{Q}+iC$.
\end{proof}
\begin{lemma}\label{lemma:cleverSOS}
	Suppose that $A\in\R[x]$ is an even polynomial such that $\deg(A)\leq 2k$ and for all $x\in [-1,1]$ we have $A(x)\geq 0$. 
	Then there exist polynomials $B,C\in\R[x]$ such that  $A(x)=B^2(x)+(1-x^2)C^2(x)$, moreover $\deg(B)\leq k$, $\deg(C)\leq k-1$, $B$ has parity-$(k \mod 2)$ and $C$ has parity-$(k-1 \mod 2)$.
\end{lemma}
\begin{proof}
	If $A=0$ the statement is trivial, so we assume in the rest that $A\neq 0$.
	Let $S$ be the multiset of roots, containing the roots of $A$ with their algebraic multiplicity. Note that if $s\in S$ then also $-s\in S$ and $s^*\in S$ since $A$ is an even real polynomial. (This statement holds considering multiplicities.) Let us introduce the following subsets of $S$ (these are again multisets):
	\begin{align*}
	S_{0}&:=\{s\in S\colon s=0\}\\
	S_{(0,1)}&:=\{s\in S\colon s\in (0,1)\}\\
	S_{[1,\infty)}&:=\{s\in S\colon s\in [1,\infty)\}\\
	S_{I}&:=\{s\in S\colon \mathrm{Re}(s)=0\,\&\,\mathrm{Im}(s)> 0\}\\
	S_{C}&:=\{s\in S\colon \mathrm{Re}(s)> 0\,\&\,\mathrm{Im}(s)> 0\}.	
	\end{align*}
	Using the roots in $S$ and some scaling factor $K\in \R_+$ we can write
	\begin{equation}\label{eq:Aprod}
	A(x)=K^2x^{|S_0|}
	\!\prod_{s\in S_{(0,1)}}\!\!(x^2-s^2)
	\!\prod_{s\in S_{[1,\infty)}}\!\!\!(s^2-x^2)
	\prod_{s\in S_{I}}(x^2+|s|^2)		
	\!\!\prod_{(a+b i)\in S_{C}}\!\!\!\!\left(x^4+2x^2(b^2-a^2)+(a^2+b^2)^2\right).
	\end{equation}
	Consider the following rearrangement of the above terms corresponding to the roots in $S_{[1,\infty)},S_{I},S_{C}$:
	\begin{align}
	s^2-x^2=(s^2-1)x^2+s^2(1-x^2)&=\underset{R_{(s)}(x):=}{\underbrace{\left(\sqrt{(s^2-1)}x+is\sqrt{1-x^2}\right)}}R_{(s)}^*(x)\\
	x^2+|s|^2=(|s|^2+1)x^2 + |s|^2(1-x^2)&=\underset{P_{(s)}(x):=}{\underbrace{\left(\sqrt{(|s|^2+1)}x+i|s|\sqrt{1-x^2}\right)}}P_{(s)}^*(x)\\		
	x^4+2x^2(b^2-a^2)+(a^2+b^2)^2
	&=\underset{Q_{(a,b)}(x):=}{\underbrace{\left(\left(cx^2-(a^2+b^2)\right)+i\sqrt{c^2-1}x\sqrt{1-x^2}\right)}} Q_{(a,b)}^*(x), \label{eq:Qprod}\\
	\text{ where\footnotemark }\,\, c&=a^2+b^2+\sqrt{2 \left(a^2+1\right)
		b^2+\left(a^2-1\right)^2+b^4}\nonumber.
	\end{align}
	\footnotetext{Observe that $c\geq 1$ for all $a,b\geq 0$ and thus $\sqrt{c^2-1}\in\R$.}
	Let us define 
	$$W(x):=Kx^{|S_0|/2}
	\prod_{s\in S_{(0,1)}}\sqrt{(x^2-s^2)}
	\prod_{s\in S_{[1,\infty)}}R_{s}(x)
	\prod_{s\in S_{I}}P_{s}(x)	
	\prod_{(a+b i)\in S_{C}}Q_{(a,b)}(x).$$
	Note that the factor $x^{|S_0|/2}\prod_{s\in S_{(0,1)}}\sqrt{(x^2-s^2)}$ is a polynomial, since every root in $S_0$ and $S_{(0,1)}$ has even multiplicity as $A(x)\geq 0$ for all $x\in (-1,1)$.
	Also note that $W(x)$ is a product of expressions of the form
	$B'(x)+i\sqrt{1-x^2}C'(x)$ where $B',C'\in \R[x]$ are polynomials having opposite parities (n.b. the zero polynomial is both even and odd, thus it has opposite parity to any even/odd polynomial). Since the product of expressions of such form can again be written in such a form, we have that $W(x)=B(x)+i\sqrt{1-x^2} C(x)$ for some $B,C\in \R[x]$ having opposite parities. Also note that $\deg(B)\leq |S|/2$ and $\deg(C)\leq |S|/2-1$.
	
	Finally observe that by \eqref{eq:Aprod}-\eqref{eq:Qprod} we have that $A(x)=W(x)\cdot W^*(x)$, thus $A(x)=B(x)^2+ (1-x^2)C(x)^2$. Since $\deg(B)\leq |S|/2\leq k$ and $\deg(C)\leq |S|/2-1\leq k-1$, in case $\deg(A)=2k$, we must have that $B$ has parity-$(k \mod 2)$ and $C$ has parity-$(k-1 \mod 2)$.
	If $\deg(A)\leq 2k-2$ and $B$ has parity-$(k-1 \mod 2)$, then consider $\tilde{W}(x):=W(x)\cdot\left(x+i\sqrt{1-x^2}\right)$. Since $\left(x+i\sqrt{1-x^2}\right)\left(x+i\sqrt{1-x^2}\right)^{\!\!*}=1$ we still have that $\tilde{W}(x)\tilde{W}^*(x)=A(x)$. Now let us denote $\tilde{W}(x)=\tilde{B}(x)+i\sqrt{1-x^2}\tilde{C}(x)$, then we get that $A(x)=\tilde{B}(x)^2+ (1-x^2)\tilde{C}(x)^2$, moreover $\deg(\tilde{B})\leq k$, $\deg(\tilde{C})\leq k-1$, $\tilde{B}$ has parity-$(k \mod 2)$ and $\tilde{C}$ has parity-$(k-1 \mod 2)$.   
\end{proof}

Note that the proofs of Theorems~\ref{thm:basicCharacterisation}-\ref{thm:cleverSOS} are constructive, therefore they also give algorithms to find $P,Q$ and $\Phi$. The most difficult step in the proofs is to find the roots of a given degree-$d$ univariate complex polynomial. This problem is fortunately well studied, and can be solved up to $\eps$ precision on a classical computer in time $\bigO{\poly{d,\log(1/\eps)}}$.

Now we prove a corollary of the above result where we replace the $W(x)$ rotation operators with the following $R(x)$ reflection gates, which fit the block-encoding formalism nicer.

\begin{definition}[Parametrized family of single qubit reflections]
	We define a parametrized \linebreak family of single qubit reflection operators for all $x\in[-1,1]$ such that
	\begin{equation}\label{eq:2DReflection}
	R(x):=\left[\begin{array}{cc} x & \sqrt{1-x^2} \\ \sqrt{1-x^2} & -x \end{array}\right].
	\end{equation}
\end{definition}

\begin{corollary}[Quantum signal processing using reflections]\label{cor:refAchievableP}
	Let $P\in\C[x]$ be a degree-$d$ polynomial, such that
	\begin{itemize}
		\item $P$ has parity-$(d \mod 2)$,
		\item $\forall x\in[-1,1]\colon |P(x)|\leq 1$,
		\item $\forall x\in(-\infty,-1]\cup[1,\infty)\colon |P(x)|\geq 1$,	
		\item if $d$ is even, then $\forall x\in\R\colon P(ix)P^*(ix)\geq 1$.
	\end{itemize} 
	Then there exists $\Phi\in\R^d$ such that\footnote{Note that the $e^{i\phi_1\sigma_z}$ gate can in fact be replaced by a simple phase gate $e^{i\phi_1}$.}
	\begin{equation}\label{eq:IterRef}
	\prod_{j=1}^{d}\left(e^{i\phi_j\sigma_z}R(x)\right)
	=\left[\begin{array}{cc} P(x) & . \\ . & . \end{array}\right].
	\end{equation}	
	Moreover for $x\in\{-1,1\}$ we have that $P(x)=x^d\prod_{j=1}^{d}e^{i\phi_j}$, and for $d$ even $P(0)=e^{-i\sum_{j=1}^{d}(-1)^j\phi_j}$.	
\end{corollary}
\begin{proof}
	By Theorem~\ref{thm:achievablePorQ} we have that there exists $\Phi'\in\R^{d+1}$ for some $d\geq 1$ such that
	\begin{equation}
	e^{i\phi'_0\sigma_z}\left(\prod_{j=1}^{d}W(x)e^{i\phi'_j\sigma_z}\right)
	=\left[\begin{array}{cc} P(x) & . \\ . & . \end{array}\right].
	\end{equation}	
	Observe that 
	\begin{equation}\label{eq:IterRefRaw}
	W(x)=i e^{-i\frac\pi4\sigma_z}R(x)e^{i\frac\pi4\sigma_z},
	\end{equation}		
	thus the left-hand-side of \eqref{eq:IterRef} equals
	\begin{align*}
	e^{i\phi'_0\sigma_z}\left(\prod_{j=2}^{d}e^{i\phi_j\sigma_z}i e^{-i\frac\pi4\sigma_z}R(x)e^{i\frac\pi4\sigma_z}\right)
	&=i^d e^{i(\phi'_0-\frac\pi4)\sigma_z}R(x)\left(\prod_{j=2}^{d}e^{i(\phi'_{j-1}-\frac\pi2)\sigma_z} R(x)\right)e^{i(\phi'_d-\frac\pi4)}.
	\end{align*}	
	Therefore
	\begin{equation*}
	e^{i(\phi'_0+\phi'_d+(d-1)\frac\pi2)}R(x)\left(\prod_{j=2}^{d}e^{i(\phi'_{j-1}-\frac\pi2)\sigma_z} R(x)\right)
	=\left[\begin{array}{cc} P(x) & . \\ . & . \end{array}\right].
	\end{equation*}		
	So choosing $\phi_1:=\phi'_0+\phi'_d+(d-1)\frac\pi2$ and for all $j\in\{2,3,\ldots,d\}$ setting $\phi_j:=\phi'_{j-1}-\frac\pi2$, results in a $\Phi\in\R^d$ that clearly satisfies \eqref{eq:IterRef}. The additional result for $x\in\{-1,1\}$ follows from the fact that for $x\in\{-1,1\}$ every matrix in \eqref{eq:IterRef} becomes diagonal.
	
	The claim about $P(0)$ follows from the observation that 
	$$ e^{i\phi_1\sigma_z}R(0)e^{i\phi_2\sigma_z}R(0) = e^{i(\phi_1-\phi_2)\sigma_z}.$$
\end{proof}

The above requirements on $P$ are not very intuitive, but fortunately we have a good understanding of the polynomials that can emerge by taking the real part of the above complex polynomials. Before stating the corresponding corollary, we note that Chebyshev polynomials satisfy the above requirements. One can prove it directly, but instead of doing so we just explicitly describe\footnote{By Theorem~\ref{thm:achievablePorQ} it actually proves that the conditions of Corollary~\ref{cor:refAchievableP} hold for Chebyshev polynomials.} the corresponding $\Phi$. 
\begin{lemma}[Constructing Chebyshev polynomials via quantum signal processing]\label{lemma:chebyshev}
	Let $T_d\in\R[x]$ be the $d$-th Chebyshev polynomial of the first kind. Let $\Phi\in\R^d$ be such that $\phi_1=(1-d)\frac\pi2$, and for all $i\in[d]\setminus\{1\}$ let $\phi_{i}:=\frac\pi2$. Using this $\Phi$ in equation \eqref{eq:IterRef} we get that $P=T_d$.
\end{lemma}
\begin{proof}
	One can prove this, e.g., by induction using the substitution $x:=\cos(\theta)$.
\end{proof}

\begin{restatable}{corollary}{realQuantumSignal}\emph{(Real quantum signal processing)}\label{cor:realP}
	Let $ P_{\Re}(x)\in\R[x]$ be a degree-$d$ polynomial for some $d\geq 1$, such that 
	\begin{itemize}
		\item $ P_{\Re}$ has parity-($d\mod 2$), and
		\item for all $x\in[-1,1]\colon | P_{\Re}(x)|\leq 1$.
	\end{itemize}
	Then there exists $P\in\C[x]$ that satisfies the requirements of Corollary~\ref{cor:refAchievableP}.

	Moreover, given $ P_{\Re}(x)$ and $\delta\geq 0$ we can find a $P$ and a corresponding $\Phi$, such that $|\Re[P]- P_{\Re}|\leq \delta$ for all $x\in[-1,1]$, using a classical computer in time $\bigO{\mathrm{poly}(d,\log(1/\delta))}$.
\end{restatable}
\begin{proof}
	The existence of such $P$ follows directly from Theorem~\ref{thm:basicCharacterisation}-\ref{thm:cleverSOS}.
	
	The complexity statement follows from the fact that we can find $P$ and $\Phi'$ using the procedures of Theorems~\ref{thm:basicCharacterisation}-\ref{thm:cleverSOS} on a classical computer in time $\bigO{\mathrm{poly}(d,\log(1/\eps))}$ as noted above. Computing $\Phi$ from $\Phi'$ as in the proof of Corollary~\ref{cor:refAchievableP} only yields a small overhead.
\end{proof}

\subsection{Singular value transformation by qubitization}\label{subsec:SingularvalTrans}

Qubitization is a technique introduced by Low and Chuang~\cite{low2016HamSimQubitization} in order to apply polynomial transformations to the spectrum of a Hermitian (or normal) operator, which is represented as the top-left block of a unitary matrix. They also showed how to use their techniques in order to develop advanced amplitude amplification techniques. In this section we generalize their results, and develop the technique of singular value transformation, which applies to any operator as opposed to only normal operators. 

It turns out that by applying a unitary $U$ back and forth interleaved with some simple phase operators one can induce polynomial transformations to the singular values of a particular (not necessarily rectangular) block-matrix of the unitary $U$. The main idea behind the qubitization approach is to lift the quantum signal processing results presented in the previous section. One can do so by defining some two-dimensional invariant subspaces within which the results of quantum signal processing apply, thereby ``qubitizing''\footnote{Another justification for the term ``qubitization'' is that the involved higher-dimensional phase operations reduce to carefully choosing a single qubit phase gate, see Figure~\ref{fig:Piphi}.} the problem. Then by understanding how the two-dimensional subspaces behave, one can infer the higher-dimensional behavior. 

The original qubitization approach can be understood along the lines of C. Jordan's Lemma~\cite{jordan1875ProducsOfReflections} about the common invariant subspaces of two reflections.\footnote{By reflection we mean a Hermitian operator having only $\pm 1$ eigenvalues, possibly having multiple $-1$ eigenvalues.} Jordan's result is most often presented stating that the product of two reflections decomposes to one- and two-dimensional invariant subspaces, such that the operator has eigenvalue $\pm 1$ on the one-dimensional subspaces, and the operator acts as a rotation on the two-dimensional subspaces. This higher dimensional insight lies at the heart of Szegedy's quantum walk results~\cite{szegedy2004QMarkovChainSearch} as well as Marriott and Watrous' QMA amplification scheme~\cite{marriott2005QAMGames}.

Motivated by a series of prior work on quantum search algorithms \cite{grover2005FixedPointSearch,hoyer2000ArbitraryPhasesAmpAmp,yoder2014FixedPointSearch} the original qubitization approach of Low and Chuang~\cite{low2016HamSimQubitization} replaced one of the reflections in Jordan's Lemma by a phase-gate, such as in Figure~\ref{fig:Piphi}. They examined the operators arising by iterative application of the reflection- and phase-operator with applying possibly different phases in each step. In this paper we go one step further and replace the other reflection\footnote{One could also merge $U$ into one of the projectors, leading to a product of reflections as in Jordan's Lemma~\cite{jordan1875ProducsOfReflections}.} by an arbitrary unitary operator $U$, and analyze the procedure with carefully chosen one- and two-dimensional subspaces coming from singular value decomposition.

\begin{definition}[Singular value decomposition of a projected unitary]\label{def:singDec}
	$\,\,$Let $\H_U$ be a finite-dimensional Hilbert space and let $U,\Pi, \widetilde{\Pi}\in\eend{\H_U}$ be linear operators on $\H_U$ such that $U$ is a unitary, and $\Pi, \widetilde{\Pi}$ are orthogonal projectors, and let
	$$A=\widetilde{\Pi} U \Pi.$$
	Let $d:=\rank{\Pi}$, $\tilde{d}:=\rank{\widetilde{\Pi}}$, $d_{\min}:=\min(d,\tilde{d})$. By singular value decomposition we know that there exist orthonormal bases $\left(\ket{\psi_i}\colon i\in [d]\right)$, $\left(\ket{\tilde{\psi}_i}\colon i\in [\tilde{d}]\right)$ of the subspaces $\img{\Pi}$ and $\img{\widetilde{\Pi}}$ respectively, 
	such that
	\begin{equation}\label{eq:SVD}
		A=\sum_{i=1}^{d_{\min}}\varsigma_i\ketbra{\tilde{\psi}_i}{\psi_i},
	\end{equation}
	and\footnote{In singular value decomposition one usually requires that the diagonal elements of $\Sigma$ are non-negative. Here we could also allow negative reals, all the proofs of this section would go through with minor modifications, mostly defining the ordering of the singular values with decreasing absolute value. This would enable one to treat spectral decompositions of Hermitian matrices also as singular value decompositions.} for all $i\in[d_{\min}]\colon\varsigma_i\in\R^+_0$.
	Moreover, $\varsigma_i\geq \varsigma_j$ for all $i\leq j\in[d_{\min}]$.
\end{definition}
	
\begin{definition}[Invariant subspaces associated to a singular value decomposition]\label{def:SVDInvSubspaces} $\phantom{+}$\linebreak
	Let $U,\Pi, \widetilde{\Pi}, A$ be as in Definition~\ref{def:singDec}, and let us use its notation.
	Let $k\in[d_{\min}]$ be the largest index for which $\varsigma_k=1$, and
	let $r=\rank{A}$. For
	 \begin{align*}
		 &i\in[k] \text{ let} & 	
		 &\H_i:=\spn{\ket{\psi_i}} \text{ and}& &\tilde{\H}_i:=\spn{\ket{\tilde{\psi}_i}},\\
		 &i\in[r]\setminus[k]\text{ let} &
		 &\H_i:=\spn{\ket{\psi_i},\ket{\psi^\perp_i}} \text{ where} & 
		 &\ket{\psi^\perp_i}:=\frac{(I-\Pi)U^\dagger\ket{\tilde{\psi}_i}}{\nrm{(I-\Pi)U^\dagger\ket{\tilde{\psi}_i}}}=\frac{(I-\Pi)U^\dagger\ket{\tilde{\psi}_i}}{\sqrt{1-\varsigma^2_i}},\\
		 &i\in[r]\setminus[k]\text{ let} &	  	
		 &\tilde{\H}_i:=\spn{\ket{\tilde{\psi}_i},\ket{\tilde{\psi}^\perp_i}} \text{ where} & 
		 &\ket{\tilde{\psi}^\perp_i}:=\frac{(I-\widetilde{\Pi})U\ket{\psi_i}}{\nrm{(I-\widetilde{\Pi})U\ket{\psi_i}}}=\frac{(I-\widetilde{\Pi})U\ket{\psi_i}}{\sqrt{1-\varsigma^2_i}},\\
		 &i\in[d]\setminus[r]\text{ let} &
		 &\H^R_i:=\spn{\ket{\psi_i}} \text{ and}&
		 &\tilde{\H}^R_i:=\spn{U\ket{\psi_i}},\\	
		 &i\in[\tilde{d}]\setminus[r]\text{ let} &
		 &\H^L_i:=\spn{U^\dagger\ket{\tilde{\psi}_i}} \text{ and}&
		 &\tilde{\H}^L_i:=\spn{\ket{\tilde{\psi}_i}}.			 	 
	 \end{align*}
	 Finally let 
	 \begin{align*}
	 &\H_\perp:=\left(\bigoplus_{i\in[r]}\H_i\,\,\oplus\!
	 \bigoplus_{i\in[d]\setminus[r]}\H_i^R\,\,\oplus\!
	 \bigoplus_{i\in[\tilde{d}\,]\setminus[r]}\H_i^L\right)^{\!\!\!\perp}\,\,\text{ and} &
	 &\tilde{\H}_\perp:=\left(\bigoplus_{i\in[r]}\tilde{\H}_i\,\,\oplus\!
	 \bigoplus_{i\in[d]\setminus[r]}\tilde{\H}_i^R\,\,\oplus\!
	 \bigoplus_{i\in[\tilde{d}\,]\setminus[r]}\tilde{\H}_i^L\right)^{\!\!\!\perp}\!\!.
	 \end{align*}
\end{definition}

Now we show that the subspaces $\H_i\colon i\in[k]$, $\H_i\colon i\in[r]\setminus[k]$, $\H^R_i\colon i\in[d]\setminus[r]$ and $\H^L_i\colon i\in[\tilde{d}]\setminus[r]$ are indeed pairwise orthogonal, by proving that their spanning bases described in Definition~\ref{def:SVDInvSubspaces} form an orthonormal system of vectors. (By symmetry it also implies that the spanning bases of the $\tilde{\H}$ subspaces form also an orthonormal system of vectors.) The proof is summarized in Table~\ref{tab:orthoProof}, relying on the following observations:
\begin{align}
&\forall i,j\in [d] & &\braket{\psi_i}{\psi_{j}}=\delta_{ij} \label{eq:orth:ONB} \\ 
&\forall i\in [d], j\in[r]\setminus[k]  & &\braket{\psi_i}{\psi^\perp_{j}}\propto \braketbra{\psi_i}{(I-\Pi)U^\dagger}{\tilde{\psi}_j}\propto\bra{\psi_i}(I-\Pi)=0 \label{eq:orth:notInPi}\\ 
&\forall i\in [d], j\in[\tilde{d}]\setminus[r]  & &\braketbra{\psi_i}{U^\dagger}{\tilde{\psi}_j}=\braketbra{\psi_i}{\Pi U^\dagger\widetilde{\Pi}}{\tilde{\psi}_j}=\braketbra{\psi_i}{A^\dagger}{\tilde{\psi}_j}\propto A^\dagger\ket{\tilde{\psi}_j}=0  \label{eq:orth:zeroSV}\\
&\forall i,j\in[r]\setminus[k]  & &\braket{\psi^\perp_i}{\psi^\perp_j}=\frac{\braketbra{\tilde{\psi}_i}{U(I-\Pi)U^\dagger}{\tilde{\psi}_j}}{\sqrt{(1-\varsigma_i^2)(1-\varsigma_j^2)}}=\frac{\delta_{ij}-\braketbra{\tilde{\psi}_i}{AA^\dagger}{\tilde{\psi}_j}}{\sqrt{(1-\varsigma_i^2)(1-\varsigma_j^2)}}=\delta_{ij}  \label{eq:orth:ONBcomp} \\
&\forall i\in[r]\setminus[k],j \in[\tilde{d}]\setminus[r] & &\braketbra{\psi^\perp_i}{U^\dagger}{\tilde{\psi}_j}=\frac{\braketbra{\tilde{\psi}_i}{U(I-\Pi)U^\dagger}{\tilde{\psi}_j}}{\sqrt{(1-\varsigma_i^2)}}=\frac{\delta_{ij}-\braketbra{\tilde{\psi}_i}{AA^\dagger}{\tilde{\psi}_j}}{\sqrt{(1-\varsigma_i^2)}}=0 \label{eq:orth:ONBcomp2} \\
&\forall i,j\in [\tilde{d}] & &\braket{\tilde{\psi}_i}{\tilde{\psi}_{j}}=\delta_{ij} \label{eq:orth:ONBtilde}
\end{align}

\begin{table}[H]
	\centering
	\renewcommand{\arraystretch}{1.5}
	\begin{tabular}{|c|ccccc|}\hline
		\multirow{2}{*}[0mm]{$\H_i \perp \H_{j}$ 	}								
		&	$\ket{\psi_j}\in \H_{j}$&	$\ket{\psi_j}\in \H_{j}$ &	$\ket{\psi^\perp_j}\in \H_{j}$  &	$\ket{\psi_j}\in \H^R_{j}$  & $U^\dagger\ket{\tilde{\psi_j}}\in \H^L_{j}$  \\	
		&	$j\in[k]$ & $j\in[r]\setminus[k]$  &	$j\in[r]\setminus[k]$  &	$j\in[d]\setminus[r]$  & $j\in[\tilde{d}]\setminus[r]$  \\ \hline
		$\ket{\psi_i}\in \H_{i}$ & by \eqref{eq:orth:ONB}	& by \eqref{eq:orth:ONB} & by \eqref{eq:orth:notInPi} & by \eqref{eq:orth:ONB} & by \eqref{eq:orth:zeroSV}  \\
		$i\in[k]$	& $\braket{\psi_i}{\psi_{j}}=\delta_{ij}$	& $\braket{\psi_i}{\psi_{j}}=0$  & $\braket{\psi_i}{\psi^\perp_j}=0$ &  $\braket{\psi_i}{\psi_{j}}=0$ &  $\braketbra{\psi_i}{U^\dagger}{\tilde{\psi}_j}=0$ \\ \hline
		$\ket{\psi_i}\in \H_{i}$ & 	& by \eqref{eq:orth:ONB} & by \eqref{eq:orth:notInPi} & by \eqref{eq:orth:ONB} & by \eqref{eq:orth:zeroSV}  \\
		$i\in[r]\setminus[k]$	& & $\braket{\psi_i}{\psi_{j}}=\delta_{ij}$  & $\braket{\psi_i}{\psi^\perp_j}=0$ &  $\braket{\psi_i}{\psi_{j}}=0$ &  $\braketbra{\psi_i}{U^\dagger}{\tilde{\psi}_j}=0$ \\ \hline
		$\ket{\psi^\perp_i}\in \H_{i}$ & &  & by \eqref{eq:orth:ONBcomp} & by \eqref{eq:orth:notInPi} & by \eqref{eq:orth:ONBcomp2}  \\
		$i\in[r]\setminus[k]$	& & & $\braket{\psi^\perp_i}{\psi^\perp_j}=\delta_{ij}$ &  $\braket{\psi^\perp_i}{\psi_{j}}=0$ &  $\braketbra{\psi^\perp_i}{U^\dagger}{\tilde{\psi}_j}=0$ \\ \hline		
		$\ket{\psi_i}\in \H^R_{i}$ & &  & & by \eqref{eq:orth:ONB} & by \eqref{eq:orth:zeroSV} \\
		$i\in[d]\setminus[r]$	& & & &  $\braket{\psi_i}{\psi_{j}}=\delta_{ij}$ &  $\braketbra{\psi_i}{U^\dagger}{\tilde{\psi}_j}=0$ \\ \hline	
		$U^\dagger \ket{\tilde{\psi}_i}\in \H^L_{i}$ & &  & & & by \eqref{eq:orth:ONBtilde} \\
		$i\in[\tilde{d}]\setminus[r]$	& & & & &  $\braketbra{\tilde{\psi}_i}{U U^\dagger}{\tilde{\psi}_j}=\delta_{ij}$ \\ \hline	
	\end{tabular}
	\caption{Proof of the orthonormality of the spanning bases described in Definition~\ref{def:SVDInvSubspaces}. }
	\label{tab:orthoProof}
\end{table}

Now we introduce some notation for matrices that represent linear maps acting between different subspaces. This will enable us to conveniently express matrices in a block-diagonal form. We will use the subspaces of Definition~\ref{def:SVDInvSubspaces}, because they enable us to block-diagonalize the unitaries used for implementing singular value transformation.

\begin{definition}[Notation for matrices of linear maps between different vector spaces]
	$\phantom{..}$ For two  vector (sub)spaces $\H,\H'$ let us denote by $[\,\cdot\,]^{\H}_{\H'}$ the matrix of a linear map that maps $\H\mapsto \H'$. Moreover, if the subspaces are as in Definition~\ref{def:SVDInvSubspaces} and we explicitly write down matrix elements, they are meant to be interpreted in the spanning bases we used for defining $\H,\H'$ in Definition~\ref{def:SVDInvSubspaces}.
\end{definition}

\begin{lemma}[Invariant subspace decomposition of a projected unitary]\label{lemma:singInvDec}
	Let $\H_U$ be a finite-dimensional Hilbert-space and $U,\Pi, \widetilde{\Pi} \in\eend{\H_U}$ be as in Definition~\ref{def:singDec}. Then using the singular value decomposition of Definition~\ref{def:SVDInvSubspaces} we have that
	\begin{align}\label{eq:UDec}
	U=
	\bigoplus_{i\in[k]}\left[\varsigma_i\right]_{\tilde{\H}_i}^{\H_i}
	\oplus 
	\bigoplus_{i\in[r]\setminus[k]}\left[\begin{array}{cc} \varsigma_i & \sqrt{1-\varsigma_i^2}\\
	\sqrt{1-\varsigma_i^2} & -\varsigma_i\end{array}\right]_{\tilde{\H}_i}^{\H_i}
	\oplus 
	\bigoplus_{i\in[d]\setminus[r]}\left[1\right]_{\tilde{\H}^R_i}^{\H^R_i}
	\oplus 
	\bigoplus_{i\in[\tilde{d}]\setminus[r]}\left[1\right]_{\tilde{\H}^L_i}^{\H^L_i}
	\oplus 	
	\left[\,\cdot\,\right]_{\tilde{\H}_\perp}^{\H_\perp}.
	\end{align}
	Moreover,
	\begin{align}
	2\Pi-I&=	\bigoplus_{i\in[k]}\left[1\right]_{\H_i}^{\H_i}
	\oplus 
	\bigoplus_{i\in[r]\setminus[k]}\left[\begin{array}{cc} 1 & 0\\
	0& -1\end{array}\right]_{\H_i}^{\H_i}
	\oplus 
	\bigoplus_{i\in[d]\setminus[r]}\left[1\right]_{\H^R_i}^{\H^R_i}
	\oplus 	
	\bigoplus_{i\in[d]\setminus[r]}\left[-1\right]_{\H^L_i}^{\H^L_i}
	\oplus 		
	\left[\,\cdot\,\right]_{\H_\perp}^{\H_\perp},\label{eq:PiUDec}\\
	e^{i\phi(2\Pi-I)}&=	\bigoplus_{i\in[k]}\left[e^{i\phi}\right]_{\H_i}^{\H_i}
	\oplus 
	\bigoplus_{i\in[r]\setminus[k]}\left[\begin{array}{cc} e^{i\phi} & 0\\
	0& e^{-i\phi}\end{array}\right]_{\H_i}^{\H_i}
	\oplus 
	\bigoplus_{i\in[d]\setminus[r]}\left[e^{i\phi}\right]_{\H^R_i}^{\H^R_i}
	\oplus 	
	\bigoplus_{i\in[d]\setminus[r]}\left[e^{-i\phi}\right]_{\H^L_i}^{\H^L_i}
	\oplus 		
	\left[\,\cdot\,\right]_{\H_\perp}^{\H_\perp},\label{eq:PiPhaseUDec}	
	\end{align}
	and
	\begin{align}
	2\widetilde{\Pi}-I&=	\bigoplus_{i\in[k]}\left[1\right]_{\tilde{\H}_i}^{\tilde{\H}_i}
	\oplus 
	\bigoplus_{i\in[r]\setminus[k]}\left[\begin{array}{cc} 1 & 0\\
	0& -1\end{array}\right]_{\tilde{\H}_i}^{\tilde{\H}_i}
	\oplus 
	\bigoplus_{i\in[d]\setminus[r]}\left[-1\right]_{\tilde{\H}^R_i}^{\tilde{\H}^R_i}
	\oplus 	
	\bigoplus_{i\in[d]\setminus[r]}\left[1\right]_{\tilde{\H}^L_i}^{\tilde{\H}^L_i}
	\oplus 		
	\left[\,\cdot\,\right]_{\tilde{\H}_\perp}^{\tilde{\H}_\perp},\label{eq:PitUDec}\\
	e^{i\phi(2\widetilde{\Pi}-I)}&=	\bigoplus_{i\in[k]}\left[e^{i\phi}\right]_{\tilde{\H}_i}^{\tilde{\H}_i}
	\oplus 
	\bigoplus_{i\in[r]\setminus[k]}\left[\begin{array}{cc} e^{i\phi} & 0\\
	0& e^{-i\phi}\end{array}\right]_{\tilde{\H}_i}^{\tilde{\H}_i}
	\oplus 
	\bigoplus_{i\in[d]\setminus[r]}\left[e^{-i\phi}\right]_{\tilde{\H}^R_i}^{\tilde{\H}^R_i}
	\oplus 	
	\bigoplus_{i\in[d]\setminus[r]}\left[e^{i\phi}\right]_{\tilde{\H}^L_i}^{\tilde{\H}^L_i}
	\oplus 		
	\left[\,\cdot\,\right]_{\tilde{\H}_\perp}^{\tilde{\H}_\perp}.\label{eq:PitPhaseUDec}	
	\end{align}
\end{lemma}
\begin{proof}
	For all $i\in[r]\setminus[k]$ we can verify that
	\begin{equation}\label{eq:Upsi}
	U\ket{\psi_i}
	=\widetilde{\Pi} U\ket{\psi_i} + (I-\widetilde{\Pi})U\ket{\psi_i}
	=\underset{A}{\underbrace{\widetilde{\Pi} U \Pi}}\ket{\psi_i} + (I-\widetilde{\Pi})U\ket{\psi_i}
	=\varsigma_i \ket{\tilde{\psi}_i} + \sqrt{1-\varsigma^2_i}\ket{\tilde{\psi}^\perp_i},
	\end{equation}
	and 
	\begin{align}
	\sqrt{1-\varsigma^2_i}U\ket{\psi^\perp_i}
	&=U(I-\Pi)U^\dagger\ket{\tilde{\psi}_i}
	=\ket{\tilde{\psi}_i}-U\Pi U^\dagger\ket{\tilde{\psi}_i}
	=\ket{\tilde{\psi}_i}-U\underset{A^\dagger}{\underbrace{\Pi U^\dagger\widetilde{\Pi}}}\ket{\tilde{\psi}_i}
	=\ket{\tilde{\psi}_i}-U\varsigma_i \ket{\psi_i}	\nonumber\\
	&=(1-\varsigma^2_i) \ket{\tilde{\psi}_i} - \varsigma_i \sqrt{1-\varsigma^2_i}\ket{\tilde{\psi}^\perp_i},\label{eq:Upsiperp}
 	\end{align}
	where in the last equality we used \eqref{eq:Upsi}. Since $U$ is unitary, it preserves the inner product and therefore maps $\H_\perp$ onto $\tilde{\H}_\perp$. Now equation \eqref{eq:UDec} directly follows from \eqref{eq:Upsi}-\eqref{eq:Upsiperp}.
	The other statements trivially follow from Definition~\ref{def:singDec}.
\end{proof}

\begin{definition}[Alternating phase modulation sequence]\label{def:phaseSeq}
	Let $\H_U$ be a finite-dimensional \linebreak Hilbert space and let $U,\Pi, \widetilde{\Pi}\in\eend{\H_U}$ be linear operators on $\H_U$ such that $U$ is a unitary, and $\Pi, \widetilde{\Pi}$ are orthogonal projectors.
	Let $\Phi\in\R^{n}$, then we define the \emph{phased alternating sequence} $U_\Phi$ as follows
	\begin{equation}\label{eq:phaseSeq}
		U_\Phi:=\left\{\begin{array}{rcl} \underset{\phantom{\sum}}{e^{i\phi_{1} (2\widetilde{\Pi}-I)}U}
		\prod_{j=1}^{(n-1)/2}\left(e^{i\phi_{2j} (2\Pi-I)}U^\dagger e^{i\phi_{2j+1} (2\widetilde{\Pi}-I)}U\right) & &\text{if }n\text{ is odd, and}\\
		\prod_{j=1}^{n/2}\left(e^{i\phi_{2j-1} (2\Pi-I)}U^\dagger e^{i\phi_{2j} (2\widetilde{\Pi}-I)}U\right)& & \text{if }n\text{ is even.}\end{array}\right.
	\end{equation}
\end{definition}
\begin{definition}[Singular value transformation by even/odd functions]\label{def:PolySVTrans}
		Let $f:\mathbb{R}\rightarrow\mathbb{C}$ be an even or odd function.
		Let $A\in\C^{\tilde{d}\times d}$, let $d_{\min}:=\min(d,\tilde{d})$ and let 
		\begin{equation*}
		A=\sum_{i=1}^{d_{\min}}\varsigma_i\ketbra{\tilde{\psi}_i}{\psi_i}
		\end{equation*}
		be a singular value decomposition of $A$. 
		
		We define the \emph{polynomial singular value transformation} of $A$, for odd function $f$ as
		\begin{equation*}
		f^{(SV)}(A):=\sum_{i=1}^{d_{\min}}f(\varsigma_i)\ketbra{\tilde{\psi}_i}{\psi_i},
		\end{equation*}
		and for even $f$ as
		\begin{equation*}
		f^{(SV)}(A):=\sum_{i=1}^{d}f(\varsigma_i)\ketbra{\psi_i}{\psi_i},
		\end{equation*}	
		where for $i\in[d]\setminus[d_{\min}]$ we define $\varsigma_i:=0$.
	\end{definition}

The following theorem is a generalized and improved version of the ``Flexible quantum signal processing'' result of Low and Chuang~\cite[Theorem 4]{low2017HamSimUnifAmp}. Our result is more general because it works for arbitrary matrices as opposed to only working for Hermitian (or normal) matrices. Another improvement is that we remove the constraint $P_{\Re}(0)=0$ for even $d$, which stems from our improved treatment of Theorem~\ref{thm:cleverSOS} and Corollary~\ref{cor:refAchievableP}. Also note that the following theorem can be viewed as a generalization of the quantum walk techniques introduced by Szegedy~\cite{szegedy2004QMarkovChainSearch}.

\begin{theorem}[Singular value transformation by alternating phase modulation]\label{thm:singValTransformation}
	Let $\H_U$ be a finite-dimensional Hilbert space and let $U,\Pi, \widetilde{\Pi}\in\eend{\H_U}$ be linear operators on $\H_U$ such that $U$ is a unitary, and $\Pi, \widetilde{\Pi}$ are orthogonal projectors.
	Let $P\in\C[x]$ and $\Phi\in\R^{n}$ is as in Corollary~\ref{cor:refAchievableP}, then
	\begin{equation}\label{eq:singValTransform}
		P^{(SV)}(\widetilde{\Pi} U\Pi)=\left\{\begin{array}{rcl} \widetilde{\Pi} U_{\Phi}\Pi & &\text{if }n\text{ is odd, and}\\
		\Pi U_{\Phi}\Pi& & \text{if }n\text{ is even.}\end{array}\right.
	\end{equation}
\end{theorem}
\begin{proof} We first prove the odd case. Observe that $P(1)=\prod_{j=1}^{n}e^{i\phi_j}$, and let $e^{i\phi_0}:=e^{i\sum_{j=1}^{n}(-1)^n\phi_j}$
	\begin{align*}
	U_{\Phi}&=e^{i\phi_{1} (2\widetilde{\Pi}-I)}U
		\prod_{j=1}^{n/2}\left(e^{i\phi_{2j} (2\Pi-I)}U^\dagger e^{i\phi_{2j+1} (2\widetilde{\Pi}-I)}U\right)\\
	&=\bigoplus_{i\in[k]}\left[\varsigma_k^nP(1)\right]_{\tilde{\H}_i}^{\H_i}
	\oplus 
	\bigoplus_{i\in[r]\setminus[k]}\left[\prod_{j=1}^{n}\left(e^{i\phi_j\sigma_z}R(\varsigma_\ell)\right)\right]_{\tilde{\H}_i}^{\H_i}
	\oplus 
	\bigoplus_{i\in[d]\setminus[r]}\left[e^{i\phi_0}\right]_{\tilde{\H}^R_i}^{\H^R_i}
	\oplus 
	\bigoplus_{i\in[\tilde{d}]\setminus[r]}\left[e^{-i\phi_0}\right]_{\tilde{\H}^L_i}^{\H^L_i}
	\oplus 	
	\left[\,\cdot\,\right]_{\tilde{\H}_\perp}^{\H_\perp}
	\tag{by Lemma~\ref{lemma:singInvDec}}\\	
	&=\bigoplus_{i\in[k]}\left[P(\varsigma_i)\right]_{\tilde{\H}_i}^{\H_i}
	\oplus 
	\bigoplus_{i\in[r]\setminus[k]}\left[\begin{array}{cc} P(\varsigma_i) & .\\
	. & . \end{array}\right]_{\tilde{\H}_i}^{\H_i}
	\oplus 
	\bigoplus_{i\in[d]\setminus[r]}\left[e^{i\phi_0}\right]_{\tilde{\H}^R_i}^{\H^R_i}
	\oplus 
	\bigoplus_{i\in[\tilde{d}]\setminus[r]}\left[e^{-i\phi_0}\right]_{\tilde{\H}^L_i}^{\H^L_i}
	\oplus 	
	\left[\,\cdot\,\right]_{\tilde{\H}_\perp}^{\H_\perp}
	.\tag{by Corollary~\ref{cor:refAchievableP}}
	\end{align*}
	Finally equation \eqref{eq:singValTransform} follows from the fact that $\Pi=\sum_{i=1}^{d}\ketbra{\psi_i}{\psi_i}$ and $\widetilde{\Pi}=\sum_{i=1}^{\tilde{d}}\ketbra{\tilde{\psi}_i}{\tilde{\psi}_i}$, therefore 
	\begin{align*}
	\widetilde{\Pi} U_{\Phi}\Pi
	&=\bigoplus_{i\in[k]}\left[P(\varsigma_i)\right]_{\tilde{\H}_i}^{\H_i}
	\oplus 
	\bigoplus_{i\in[r]\setminus[k]}\left[\begin{array}{cc} P(\varsigma_i) & 0\\
	0 & 0 \end{array}\right]_{\tilde{\H}_i}^{\H_i}
	\oplus 
	\bigoplus_{i\in[d]\setminus[r]}\left[0\right]_{\tilde{\H}^R_i}^{\H^R_i}
	\oplus 
	\bigoplus_{i\in[\tilde{d}]\setminus[r]}\left[0\right]_{\tilde{\H}^L_i}^{\H^L_i}
	\oplus 	
	\left[0\right]_{\tilde{\H}_\perp}^{\H_\perp}\\
	&=\sum_{i=1}^{d_{\min}}P(\varsigma_i)\ketbra{\tilde{\psi}_i}{\psi_i}.
	\end{align*}
	The last equality follows from the observation that for $n$ odd $P$ is odd, therefore $P(0)=0$.
	
	For the even case we can similarly derive that 
	\begin{align*}
	U_{\Phi}&=\bigoplus_{i\in[k]}\left[P(\varsigma_i)\right]_{\H_i}^{\H_i}
	\oplus 
	\bigoplus_{i\in[r]\setminus[k]}\left[\begin{array}{cc} P(\varsigma_i) & .\\
	. & . \end{array}\right]_{\H_i}^{\H_i}
	\oplus 
	\bigoplus_{i\in[d]\setminus[r]}\left[e^{-i\phi_0}\right]_{\H^R_i}^{\H^R_i}
	\oplus 
	\bigoplus_{i\in[\tilde{d}]\setminus[r]}\left[e^{i\phi_0}\right]_{\H^L_i}^{\H^L_i}
	\oplus 	
	\left[\,\cdot\,\right]_{\H_\perp}^{\H_\perp}
	.\tag{by Corollary~\ref{cor:refAchievableP}}
	\end{align*}	
	Finally equation \eqref{eq:singValTransform} follows from the fact that $\Pi=\sum_{i=1}^{d}\ketbra{\psi_i}{\psi_i}$, and therefore 
	\begin{align*}
	\Pi U_{\Phi}\Pi
	&=\bigoplus_{i\in[k]}\left[P(\varsigma_i)\right]_{\H_i}^{\H_i}
	\oplus 
	\bigoplus_{i\in[r]\setminus[k]}\left[\begin{array}{cc} P(\varsigma_i) & 0\\
	0 & 0 \end{array}\right]_{\H_i}^{\H_i}
	\oplus 
	\bigoplus_{i\in[d]\setminus[r]}\left[e^{-i\phi_0}\right]_{\H^R_i}^{\H^R_i}
	\oplus 
	\bigoplus_{i\in[\tilde{d}]\setminus[r]}\left[0\right]_{\H^L_i}^{\H^L_i}
	\oplus 	
	\left[0\right]_{\H_\perp}^{\H_\perp}\\
	&=\sum_{i=1}^{d}P(\varsigma_i)\ketbra{\psi_i}{\psi_i}.
	\end{align*}
	The last equality uses the observation that for $n$ even  $P(0)=e^{-i\phi_0}$, as shown by Corollary~\ref{cor:refAchievableP}.
\end{proof}

\begin{corollary}[Singular value transformation by real polynomials]\label{cor:matchingParity}
	Let $U,\Pi, \widetilde{\Pi}$ be as in \linebreak Theorem~\ref{thm:singValTransformation}.
	Suppose that $ P_{\Re}\in\R[x]$ is a degree-$n$ polynomial satisfying that
	\begin{itemize}
		\item $ P_{\Re}$ has parity-$(n\mod 2)$ and
		\item for all $x\in[-1,1]\colon$ $| P_{\Re}(x)|\leq 1$.
	\end{itemize}
	Then there exist $\Phi\in\R^n$, such that
	\begin{equation}\label{eq:realSingValTransform}
	P_\Re^{(SV)}\!\left(\widetilde{\Pi}U\Pi\right)\!=\!\left\{\begin{array}{rcl} \!\!\!\left(\bra+\otimes\widetilde{\Pi}\right)\Big(\ketbra00\!\otimes\! U_{\Phi}+\ketbra11\!\otimes\! U_{\!-\Phi}\Big) \left(\ket+\otimes\overset{\phantom{.}}{\Pi}\right)\!\!\!\!\!& &\text{if }n\text{ is odd, and}\\[\medskipamount]
	\!\!\!\left(\bra+\otimes\underset{\phantom{.}}{\Pi}\right)\Big(\ketbra00\!\otimes\! U_{\Phi}+\ketbra11\!\otimes\! U_{\!-\Phi}\Big) \left(\ket+\otimes\underset{\phantom{.}}{\Pi}\right)\!\!\!\!\!& & \text{if }n\text{ is even.}\end{array}\right.
	\end{equation}
\end{corollary}
\begin{proof}
	By Corollary~\ref{cor:realP} we can find a $\Phi\in\R^n$ such that $\Re[P]=P_\Re$. Observe that $-\Phi$ gives rise to $P^*$ in Corollary~\ref{cor:refAchievableP} as can be seen from equation~\eqref{eq:IterRef}.
	Let $\Pi'=\widetilde{\Pi}$ for $n$ odd and let $\Pi'=\Pi$ for $n$ even. Then by Theorem~\ref{thm:singValTransformation} we get that $P^{(SV)}\left(\widetilde{\Pi}U\Pi\right)=\Pi'U_\Phi\Pi$, and $P^{*(SV)}\left(\widetilde{\Pi}U\Pi\right)=\Pi'U_{-\Phi}\Pi$. Therefore 
	\begin{align*}
		\left(\bra+\otimes\Pi'\right)\left(\ketbra00\otimes U_{\Phi}\right) \left(\ket+\otimes\Pi\right)
		&=P^{(SV)}\left(\widetilde{\Pi}U\Pi\right)/2\\
		\left(\bra+\otimes\Pi'\right)\left(\ketbra11\otimes U_{\!-\Phi}\right) \left(\ket+\otimes\Pi\right)
		&=P^{*(SV)}\left(\widetilde{\Pi}U\Pi\right)/2.
	\end{align*}
	We can conclude by observing that $P_\Re=(P+P^*)/2$, and therefore 
	\begin{equation*}
		P_\Re^{(SV)}\left(\widetilde{\Pi}U\Pi\right)=\left(P^{(SV)}\!\left(\widetilde{\Pi}U\Pi\right)+P^{*(SV)}\!\left(\widetilde{\Pi}U\Pi\right)\right)/2.
	\end{equation*}
\end{proof}

Note that the above result is essentially optimal in the sense that the requirements are necessary. It is obvious that the polynomial needs to be bounded within $[-1,1]$ since the matrix must have norm at most $1$ as it is a projected unitary. Also one cannot implement a degree $d$ Chebyshev polynomial with $d-1$ uses of the unitary $U$, since $T_d(x)$ takes value $1$ at $1$ with derivative $d^2$. Substituting $y:=1$ and $x:=1-\delta$ for some small $\delta$ to equation \eqref{eq:preciseQueryBound} in Theorem~\ref{thm:lowerBoundEVT} shows that exactly implementing $T_d(x)$ requires at least $d$ uses of $U$.
Finally, about the parity constraint, note that every result in this subsection would stay valid if we would extend the concept of singular values by allowing negative values as well. But then by changing a singular vector/singular value term $\varsigma\ketbra{\phi}{\psi}$ to $-\varsigma\ketbra{-\phi}{\psi}$ would be again a valid decomposition, where singular value transformation could be applied with a polynomial $P$.
It would require that $P(\varsigma)\ketbra{\psi}{\psi}=P(-\varsigma)\ketbra{\psi}{\psi}$, and $P(\varsigma)\ketbra{\phi}{\psi}=-P(-\varsigma)\ketbra{-\phi}{\psi}$ for consistency, showing the necessity of the even/odd constraint. Equations \eqref{eq:HermitianEven}-\eqref{eq:HermitianOdd} in the proof of Corollary~\ref{cor:PolyNormDiffHerm} also show that the even/odd case separation for singular value transformation is quite natural.

What remains is to discuss how to efficiently implement an alternating phase modulation sequences. Observe that with  a single ancilla qubit, two uses of C$_\Pi$NOT, and a single-qubit phase gate $e^{-i\phi\sigma_z}$ we can implement the operator $e^{i\phi(2\Pi-I)}$=C$_\Pi$NOT$\left(I\otimes e^{-i\phi\sigma_z}\right)$C$_\Pi$NOT, which leads to an efficient implementation of $U_{\Phi}$, see Figure~\ref{fig:Piphi}.
\begin{lemma}[Efficient implementation of alternating phase modulation sequences]\label{lemma:implementingPhasedSeq}
	$\phantom{+}$\linebreak Let $\Phi\in\R^n$, then the alternating phased sequence $U_\Phi$ of Definition~\ref{def:phaseSeq} can implemented using a single ancilla qubit with $n$ uses of $U$ and $U^\dagger$, $n$ uses of C$_\Pi$NOT and $n$ uses of C$_{\widetilde{\Pi}}$NOT gates and $n$ single qubit gates. A controlled version of $U_\Phi$ can be built similarly just replacing the $n$ single qubit gates by controlled gates, and in case $n$ is odd replacing one $U$ gate with a controlled $U$ gate. For a set of vectors $\{\Phi^{(k)}\in\R^n\colon k\in\{0,1\}^m \}$ a multi-controlled alternating phased sequence $\sum_{k\in\{0,1\}^m}\ketbra{k}{k}\otimes U_{\Phi^{(k)}}$ can be implemented similarly by replacing the single qubit gates with multiply controlled single qubit gates of the form $\sum_{k\in\{0,1\}^m}\ketbra{k}{k}\otimes e^{i\phi^{(k)}}$.
\end{lemma}
\begin{proof}
	See the constructions of Figure~\ref{fig:qubitization}.
\end{proof}

	Note that Figure~\ref{fig:qubitization} also explains the term ``qubitization'': the fine-tuned driving of the circuit giving rise to the required polynomial transformation is achieved by cleverly chosen Pauli-$z$ rotations of a single ancilla qubit. The rotations of the single ancilla qubit induce rotations on the common two-dimensional invariant subspaces of $U,\Pi,\widetilde{\Pi}$ cf. Lemma~\ref{lemma:singInvDec}.

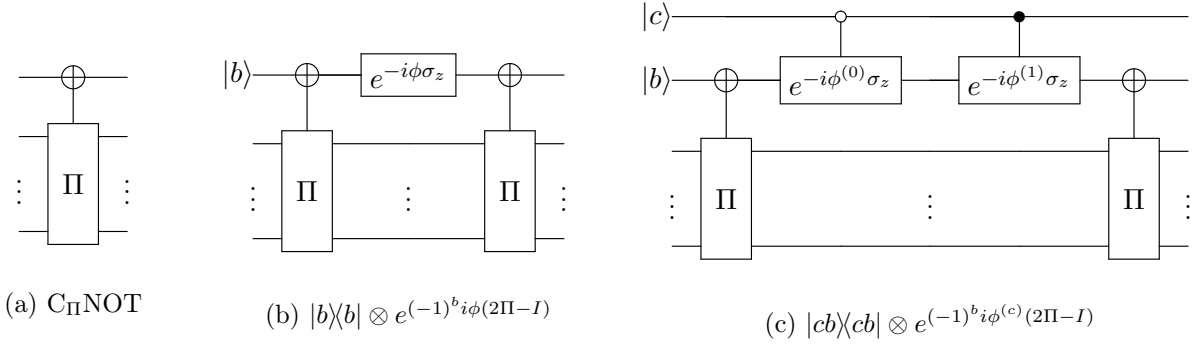
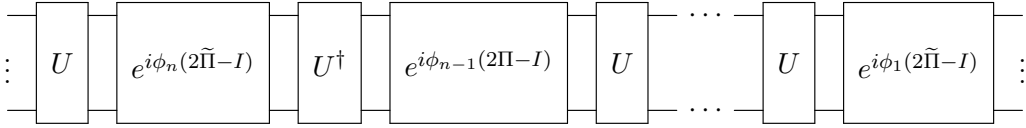
\begin{figure}[H]
	\providecommand{\ctrlA}{\push{\rule{1.5mm}{1.5mm}}}\vskip-10mm
	\begin{subfigure}{.15\textwidth}
		\centering
		\begin{displaymath}
		\Qcircuit @C=1.0em @R=1.2em {
			& & \\
			& & \\				
			&\targ{1}				& \qw\\
			&\multigate{2}{\Pi}\qwx	& \qw\\
			\vdots& 	&\vdots\\
			&\ghost{\Pi} 			& \qw	
		}		
		\end{displaymath}		
		\caption{C$_\Pi$NOT}
		\label{fig:CPiNot}
	\end{subfigure}%
	\hskip0.02\textwidth	
	\begin{subfigure}{.35\textwidth}
		\centering
		\begin{displaymath}
		\Qcircuit @C=1.0em @R=1.2em {
			& & & &\\
			& & & &\\				
			\lstick{\ket{b}\kern-2mm}
			&\targ{1}				& \gate{e^{-i\phi\sigma_z}} \qw	& \targ{1}				&\qw\\
			&\multigate{2}{\Pi}\qwx	& \qw			& \multigate{2}{\Pi}\qwx&\qw \\
			\vdots& 	&\vdots  & 	& \vdots\\
			&\ghost{\Pi} 			& \qw			& \ghost{\Pi}			&\qw
		}		
		\end{displaymath}		
		\caption{$\ketbra{b}{b}\otimes e^{(-1)^bi\phi(2\Pi-I)}$}
		\label{fig:Piphi}
	\end{subfigure}%
	\hskip0.02\textwidth
	\begin{subfigure}{.45\textwidth}
		\centering
		\begin{displaymath}
		\Qcircuit @C=1.0em @R=1.2em {
			& & & &\\			
			\lstick{\ket{c}\kern-2mm}
			&\qw					& \ctrlo{1} \qw	& \qw &\ctrl{1} \qw	
			& \qw					&\qw\\
			\lstick{\ket{b}\kern-2mm}
			&\targ{1}				& \gate{e^{-i\phi^{(0)}\sigma_z}} \qw & \qw	& \gate{e^{-i\phi^{(1)}\sigma_z}} \qw
			&\targ{1}				&\qw\\
			&\multigate{2}{\Pi}\qwx	& \qw	 & \qw  & \qw		 
			&\multigate{2}{\Pi}\qwx&\qw \\
			\vdots& 	& &\vdots & & 	& \vdots\\
			&\ghost{\Pi} 			& \qw & \qw	 & \qw			& \ghost{\Pi}			&\qw
		}		
		\end{displaymath}		
		\caption{$\ketbra{cb}{cb}\otimes e^{(-1)^bi\phi^{(c)}(2\Pi-I)}$}
		\label{fig:cPiphi}
	\end{subfigure}%
	\\\vskip-5mm\centering	
	\begin{subfigure}{0.9\textwidth}
		\centering
		\begin{displaymath}
		\Qcircuit @C=1.0em @R=1.2em {
			& &\\
			& &\\
			& \multigate{2}{U} 
			&\multigate{2}{e^{i\phi_{n}(2\widetilde{\Pi}-I)}}	
			& \multigate{2}{U^\dagger} 
			&\multigate{2}{e^{i\phi_{n-1}(2\Pi-I)}}	
			& \multigate{2}{U} 						
			& \qw & \cdots &
			& \multigate{2}{U}	&\multigate{2}{e^{i\phi_1(2\widetilde{\Pi}-I)}} & \qw  \\
			\vdots & & & & & & & & & & & \vdots \\
			&\ghost{U}		
			&\ghost{e^{i\phi_{n}(2\widetilde{\Pi}-I)}} 	
			&\ghost{U^\dagger}		
			&\ghost{e^{i\phi_{n-1}(2\Pi-I)}} 
			&\ghost{U}											
			& \qw & \cdots &
			& \ghost{U}	&\ghost{e^{i\phi_{1}(2\widetilde{\Pi}-I)}} &\qw
		}		
		\end{displaymath}	
		\caption{$U_{\Phi}=e^{i\phi_{1} (2\widetilde{\Pi}-I)}U
			\prod_{j=1}^{(n-1)/2}\left(e^{i\phi_{2j} (2\Pi-I)}U^\dagger e^{i\phi_{2j+1} (2\widetilde{\Pi}-I)}U\right)$ (for odd $n$)}
		\label{fig:UPhi}
	\end{subfigure}%
	\caption{Gates and gate sequences used for singular value transformation in Theorem~\ref{thm:singValTransformation}.
		Figure~\ref{fig:CPiNot} shows how to implement a C$_\Pi$NOT gate, and Figure~\ref{fig:Piphi} shows how to implement $e^{i\phi(2\Pi-I)}$ using a single ancilla qubit, two C$_\Pi$NOT gates and an $e^{-i\phi\sigma_z}$ gate. Figure~\ref{fig:cPiphi} demonstrates how to implement a controlled version of the gate $e^{i\phi^{(c)}(2\Pi-I)}$, by only controlling the single qubit gate $e^{-i\phi^{(c)}\sigma_z}$. Finally, Figure~\ref{fig:UPhi} summarizes the complete circuit used in Theorem~\ref{thm:singValTransformation}.\\
}
	\label{fig:qubitization}
\end{figure}

\subsection{Robustness of singular value transformation}\label{subsec:robustness}
	In this subsection we will prove results about the robustness of singular value transformation. More precisely we prove bounds on how big can be difference $\nrm{P^{(SV)}(A)-P^{(SV)}(\tilde{A})}$ in terms of the magnitude of ``perturbation'' $\nrm{A-\tilde{A}}$.
	
	First consider the generalization of ordinary $\R\rightarrow \C$ functions to Hermitian matrices. One is tempted to think that if such a function is Lipschitz-continuous, then the induced operator function is also Lipschitz-continuous, however it turns out to be false. For a recent survey on the topic see the work of Aleksandrov and Peller~\cite{aleksandrov2016OperatorLipschitzFun}. 
	
	Although the Lipschitz property cannot be saved directly, one may not lose more than some logarithmic factors. We invoke a nice result form the theory of operator functions, quantifying this claim. The following theorem is due to Farforovskaya and Nikolskaya~\cite[Theorem 10]{farforovskaya2008ModulusContOperatorFun}.
	\begin{theorem}[Robustness of eigenvalue transformation]\label{thm:ModulusDiff}
		Suppose that $f\colon [-1,1]\rightarrow \C$ is a function such that $\omega \colon [0,2]\rightarrow [0,\infty]$ is a modulus of continuity, i.e., for all $x,x'\in[-1,1]$ 
		$$ |f(x)-f(x')|\leq \omega(|x-x'|).$$
		Then for all $A,B$ Hermitian matrices such that $\nrm{A},\nrm{B}\leq 1$, we have that $$\nrm{f(A)-f(B)}\leq  4 \left[\ln\left(\frac{2}{\nrm{A-B}}+1\right)+1\right]^2\omega(\nrm{A-B}).$$
	\end{theorem}
	Now we show how this general theorem implies a general robustness result for singular value transformation.
	
	\begin{corollary}[Robustness of singular value transformation 1]\label{cor:PolyNormDiffHerm}
		If $f\colon [-1,1]\rightarrow \C$ is an even or odd function such that $\omega \colon [0,2]\rightarrow [0,\infty]$ is a modulus of continuity, and $A,\tilde{A}\in\C^{\tilde{d}\times d}$ are matrices of operator norm at most $1$, then we have that 
		$$\nrm{f^{(SV)}(A)-f^{(SV)}(\tilde{A})}\leq4 \left[\ln\left(\frac{2}{\nrm{A-\tilde{A}}}+1\right)+1\right]^2 \omega\bigg(\nrm{A-\tilde{A}}\bigg).$$
	\end{corollary}
	\begin{proof}
		Let us assume that $f$ is an even function and that $\tilde{d}\leq d$. Then, using singular value decomposition, we can rewrite $A$ as
		\begin{equation*}
			A=W
			\begin{bmatrix}
				\Sigma & 0
			\end{bmatrix}
			V^\dagger,
		\end{equation*}
		where $W\in\mathbb{C}^{\tilde{d}\times\tilde{d}}$, $V\in\mathbb{C}^{d\times d}$ are unitaries and $\Sigma\in\mathbb{R}^{\tilde{d}\times\tilde{d}}$ is a diagonal matrix with nonnegative diagonal entries. Let $\overline{A}:=\begin{bmatrix}
		0 & A\\
		A^\dagger & 0
		\end{bmatrix}\in\C^{(\tilde{d}+d)\times(\tilde{d}+d)}$ be the Hermitian matrix obtained from $A$. We claim that
		\begin{equation}\label{eq:HermitianEven}
		f(\overline{A})=
		\begin{bmatrix}
		f^{(SV)}(A^\dagger) & 0\\
		0 & f^{(SV)}(A)
		\end{bmatrix}.
		\end{equation}
		
		To prove this claim, first note that
		\begin{equation*}
			\overline{A}=
			\begin{bmatrix}
			0 & A\\
			A^\dagger & 0
			\end{bmatrix}
			=\begin{bmatrix}
			0 & W\begin{bmatrix}	\Sigma & 0	\end{bmatrix}	V^\dagger\\
			V\begin{bmatrix}	\Sigma \\ 0	\end{bmatrix}	W^\dagger & \begin{bmatrix}		0 & 0\\		0 & 0	\end{bmatrix}
			\end{bmatrix}
			=\begin{bmatrix}
			W & 0\\
			0 & V
			\end{bmatrix}
			\begin{bmatrix}
			0 & \Sigma & 0\\
			\Sigma & 0 & 0\\
			0 & 0 & 0
			\end{bmatrix}
			\begin{bmatrix}
			W^\dagger & 0\\
			0 & V^\dagger
			\end{bmatrix}
		\end{equation*}
		and that
		\begin{equation*}
			\begin{bmatrix}
				0 & \Sigma\\
				\Sigma & 0
			\end{bmatrix}
			=\frac{1}{\sqrt{2}}
			\begin{bmatrix}
			I & I\\
			I & -I
			\end{bmatrix}
			\begin{bmatrix}
			\Sigma & 0\\
			0 & -\Sigma
			\end{bmatrix}
			\frac{1}{\sqrt{2}}
			\begin{bmatrix}
			I & I\\
			I & -I
			\end{bmatrix}.
		\end{equation*}
		Therefore, if we denote
		\begin{equation*}
			U=\begin{bmatrix}
			W & 0\\
			0 & V
			\end{bmatrix}
			\begin{bmatrix}
			\frac{1}{\sqrt{2}}	\begin{bmatrix}	I & I\\	I & -I	\end{bmatrix} & \begin{array}{c} 0 \\ 0	\end{array}\\
			\begin{array}{c c}	0 & 0	\end{array} & I
			\end{bmatrix},
		\end{equation*}
		we get
		\begin{equation*}
			\overline{A}=U
			\begin{bmatrix}
			\Sigma & 0 & 0\\
			0 & -\Sigma & 0\\
			0 & 0 & 0
			\end{bmatrix}
			U^\dagger,
		\end{equation*}
		which implies that
		\begin{equation*}
		\begin{aligned}
			f(\overline{A})=&U
			\begin{bmatrix}
			f(\Sigma) & 0 & 0\\
			0 & f(-\Sigma) & 0\\
			0 & 0 & f(0)I
			\end{bmatrix}
			U^\dagger
			=U
			\begin{bmatrix}
			f(\Sigma) & 0 & 0\\
			0 & f(\Sigma) & 0\\
			0 & 0 & f(0)I
			\end{bmatrix}
			U^\dagger\\
			=&\begin{bmatrix}
			Wf(\Sigma)W^\dagger & \begin{array}{c c}	0 & 0	\end{array}\\
			\begin{array}{c} 0 \\ 0	\end{array} & V \begin{bmatrix}		f(\Sigma) & 0\\		0 & f(0)I		\end{bmatrix} V^\dagger
			\end{bmatrix}
			=\begin{bmatrix}
			f^{(SV)}(A^\dagger) & 0\\
			0 & f^{(SV)}(A)
			\end{bmatrix}.
		\end{aligned}
		\end{equation*}
        
		Thus, using Theorem~\ref{thm:ModulusDiff} we get that
		\begin{equation*}
		\begin{aligned}
		\nrm{f^{(SV)}(A)-f^{(SV)}(\tilde{A})}&\leq\nrm{f(\overline{A})-f(\overline{\tilde{A}})}\\
		&\leq 4 \left[\ln\left(\frac{2}{\nrm{\overline{A}-\overline{\tilde{A}}}}+1\right)+1\right]^2 \omega\bigg(\nrm{\overline{A}-\overline{\tilde{A}}}\bigg)\\
		&=4 \left[\ln\left(\frac{2}{\nrm{A-\tilde{A}}}+1\right)+1\right]^2 \omega\bigg(\nrm{A-\tilde{A}}\bigg),
		\end{aligned}
		\end{equation*}
		which completes the proof for the case where $f$ is an even function and that $\tilde{d}\leq d$. The case $\tilde{d}\ge d$ can be handled by symmetry. Finally, the remaining case where $f$ is odd can be handled similarly by observing that
		\begin{equation}\label{eq:HermitianOdd}
		f(\overline{A})=
		\begin{bmatrix}
		0 & f^{(SV)}(A)\\
		f^{(SV)}(A^\dagger) & 0
		\end{bmatrix}.
		\end{equation}
	\end{proof}
	
	We can also prove robustness results by bootstrapping our exact (non-robust) results, enabling us to remove the $\log$ factor from the above corollary under certain circumstances. We study two cases. First we make no extra assumptions, and establish error bounds that scale with the square root of the initial error. Then we improve the dependence to linear under the assumption that the singular values are bounded away from $1$ in absolute value.


\begin{lemma}[Robustness of singular value transformation 2]\label{lem:PolyNormDiff}
	If $ P\in\C[x]$ is a degree-$n$ polynomial satisfying the requirements of Corollary~\ref{cor:refAchievableP},
	moreover $A,\tilde{A}\in\C^{\tilde{d}\times d}$ are matrices of operator norm at most $1$, then we have\footnote{Let us do a sanity check for $d=\tilde{d}=1$. For large $d$ we have that $T_d(1)-T_d(1-\frac{1}{2d^2})\approx1-\cos(1)\approx 0.46$, whereas our upper bound gives $2\sqrt{2}$, showing that the upper bound is tight up to a constant factor, for arbitrary large $d$ and for arbitrary small $\eps$. (However, the joint dependence on $d$ and $\eps$ might not be optimal.)} that 
	$$\nrm{P^{(SV)}(A)-P^{(SV)}(\tilde{A})}\leq 4n\sqrt{\nrm{A-\tilde{A}}}.$$
\end{lemma}
\begin{proof}
	Let $\eps=\nrm{\tilde{A}-A}$, and let $B,\tilde{B}\in\C^{(d+\tilde{d})\times(d+\tilde{d})}$ be the matrices such that 	
	$$B:=\left[\begin{array}{cc} A & 0\\
	0 & 0\end{array}\right], \,\,\,
	\tilde{B}:=\left[\begin{array}{cc}\frac{\tilde{A}-A}{\eps} & 0\\
	0 & 0\end{array}\right],
	$$
	and let $U\in\C^{4(d+\tilde{d})\times4(d+\tilde{d})}$ be a unitary such that\footnote{We denote by $\,\cdot\,$ arbitrary matrix blocks and elements that are irrelevant for our presentation.}
	$$
		U= \left[\begin{array}{cccc} 
		B & 0 & . & . \\
		0 & \tilde{B} & . & . \\
		. & . & . & . \\
		. & . & . & . 
		\end{array}\right].
	$$
	Such $U$ must exist because $\nrm{B}\leq 1$ and $\nrm{\tilde{B}}\leq 1$.
	Let $\Pi$ be the orthogonal projector projecting to the first $d$ coordinates, and let $\widetilde{\Pi}$ be the orthogonal projector projecting to the first $\tilde{d}$ coordinates. Observe that $\widetilde{\Pi}U\Pi=A$. Let $W\in\C^{4(d+\tilde{d})\times4(d+\tilde{d})}$ be the unitary
	$$W:= \left[\begin{array}{cccc} 
	\sqrt{\frac{1}{1+\eps}}I & -\sqrt{\frac{\eps}{1+\eps}}I & 0 & 0 \\
	\sqrt{\frac{\eps}{1+\eps}}I & \sqrt{\frac{1}{1+\eps}}I & 0 & 0 \\
	0 & 0 & I & 0 \\
	0 & 0 & 0 & I 
	\end{array}\right].$$
	Let $\bar{U}:=W^\dagger U W$, and observe that $\widetilde{\Pi}\bar{U}\Pi=\tilde{A}/(1+\eps)$. 
	Also observe that 
	$$\nrm{W-I}=\sqrt{2-2/\sqrt{1+\eps}}\leq \sqrt{\eps},$$
	therefore $\nrm{U-\bar{U}}\leq 2\sqrt{\eps}$. Let $\Pi'=\widetilde{\Pi}$ if $n$ is odd, and let $\Pi'=\Pi$ for $n$ even. Let $\Phi$ be as in Corollary~\ref{cor:refAchievableP}, then
	Theorem~\ref{thm:singValTransformation} implies that 
	$$\nrm{P^{(SV)}\!(A)-P^{(SV)}\!(\tilde{A}/(1+\eps))}
	=\nrm{\Pi'U_{\Phi}\Pi-\Pi'\bar{U}_{\Phi}\Pi}
	\leq \nrm{U_{\Phi}-\bar{U}_{\Phi}}
	\leq n\nrm{U-\bar{U}}
	\leq
	2n\sqrt{\nrm{A-\tilde{A}}}.\!$$
	Let $B'\in\C^{(d+\tilde{d})\times(d+\tilde{d})}$ be the matrix such that 	
	$$B':=\left[\begin{array}{cc} \tilde{A} & 0\\
	0 & 0\end{array}\right],
	$$
	and let $U'\in\C^{4(d+\tilde{d})\times4(d+\tilde{d})}$ be a unitary such that
	$$
	U'= \left[\begin{array}{cccc} 
	B' & 0 & . & . \\
	0 & 0 & . & . \\
	. & . & . & . \\
	. & . & . & . 
	\end{array}\right].
	$$
	Observe that $\widetilde{\Pi}U'\Pi=\tilde{A}$, and $\bar{U}':=W^\dagger \tilde{V} W$ is such that $\widetilde{\Pi}\bar{U}'\Pi=\tilde{A}/(1+\eps)$. By the same argument as before we get that
	$$\nrm{P^{(SV)}(\tilde{A})-P^{(SV)}(\tilde{A}/(1+\eps))}\leq 2n\sqrt{\nrm{A-\tilde{A}}}.$$
	We can conclude using the triangle inequality.
\end{proof}

Now we establish another lemma which improves on the previous results for example in the case when the singular values are bounded away from $1$ in absolute value.

\begin{lemma}[Robustness of singular value transformation 3]\label{lem:PolyNormDiff2}
	If $ P\in\C[x]$ is a degree-$n$ polynomial satisfying the requirements of Corollary~\ref{cor:refAchievableP},
	moreover $A,\tilde{A}\in\C^{\tilde{d}\times d}$ are matrices of operator norm at most $1$, such that 
	$$\nrm{A-\tilde{A}}+\nrm{\frac{A+\tilde{A}}{2}}^2\leq 1,$$
	then we have that 
	$$\nrm{P^{(SV)}(A)-P^{(SV)}(\tilde{A})}\leq n\sqrt{\frac{2}{1-\nrm{\frac{A+\tilde{A}}{2}}^2}}\nrm{A-\tilde{A}}.$$
\end{lemma}
\begin{proof}
	Let $B,\tilde{B}\in\C^{(d+\tilde{d})\times(d+\tilde{d})}$ be the matrices such that 	
	$$B:=\left[\begin{array}{cc} \frac{A+\tilde{A}}{\nrm{A+\tilde{A}}} & 0\\
	0 & 0\end{array}\right], \,\,\,
	\tilde{B}:=\left[\begin{array}{cc}\frac{A-\tilde{A}}{\nrm{A-\tilde{A}}} & 0\\
	0 & 0\end{array}\right].
	$$
	Let $x> 1$ and let $U\in\C^{4(d+\tilde{d})\times4(d+\tilde{d})}$ be a unitary such that
	$$
	U= \left[\begin{array}{cccc} 
	\sqrt{\frac{x-1}{x}} B & \sqrt{\frac{1}{x}} \tilde{B} & . & . \\
	. & . & . & . \\
	. & . & . & . \\
	. & . & . & . 
	\end{array}\right].
	$$
	Let $C:=\sqrt{\frac{x-1}{x}} B \oplus \sqrt{\frac{1}{x}}\tilde{B}$ be top-left block of $U$. It is easy to see that 
	$$\nrm{C}^2\leq \frac{x-1}{x}\nrm{B}^2 + \frac{1}{x}\nrm{\tilde{B}}^2
	=\frac{x-1}{x} + \frac{1}{x} = 1,$$
	 therefore a unitary $U$ must exist with $C$ being the top-left block.
	Suppose that
	\begin{equation}\label{eq:optimisedTangent}
	\frac{x}{x-1}\frac{\nrm{A+\tilde{A}}^2}{4}+x\frac{\nrm{A-\tilde{A}}^2}{4}=1.
	\end{equation} 
	Let $W_{\pm}\in\C^{4(d+\tilde{d})\times4(d+\tilde{d})}$ be the unitary
	$$W_{\pm}:= \left[\begin{array}{cccc} 
	\sqrt{\frac{x}{x-1}}\frac{\nrm{A+\tilde{A}}}{2}I & \mp\sqrt{x}\frac{\nrm{A-\tilde{A}}}{2}I & 0 & 0 \\
	\pm\sqrt{x}\frac{\nrm{A-\tilde{A}}}{2}I & \sqrt{\frac{x}{x-1}}\frac{\nrm{A+\tilde{A}}}{2}I & 0 & 0 \\
	0 & 0 & I & 0 \\
	0 & 0 & 0 & I 
	\end{array}\right].$$
	
	Let $\Pi$ be the orthogonal projector projecting to the first $d$ coordinates, and let $\widetilde{\Pi}$ be the orthogonal projector projecting to the first $\tilde{d}$ coordinates. Observe that $\widetilde{\Pi}UW_+\Pi=A$ and $\widetilde{\Pi}UW_-\Pi=\tilde{A}$.
	Also observe that $\nrm{W_+-W_-}=\sqrt{x}\nrm{A-\tilde{A}}$, thus $\nrm{UW_+-UW_-}=\sqrt{x}\nrm{A-\tilde{A}}$.

	Let $\eps:=\nrm{A-\tilde{A}}^2$ and let $\delta:=4-\nrm{A+\tilde{A}}^2$.
	We can rewrite \eqref{eq:optimisedTangent} as
	\begin{equation}\label{eq:optimisedTangent2}
	\frac{x}{x-1}\frac{4-\delta}{4}+x\frac{\eps}{4}=1,
	\end{equation} 	 
	which has a solution
	\begin{equation}\label{eq:optimisedTangent3}
	x=\frac{4}{\delta+\eps}\left(1+ \frac{ \left(1-\frac{8 \eps}{(\delta+\eps)^2}\right)-\sqrt{1-\frac{16\eps}{(\delta+\eps)^2}}}{\frac{8\eps}{(\delta+\eps)^2}}\right).
	\end{equation}
	Now let $y:=8\eps/(\delta+\eps)^2$, it is easy to see that for $y\leq\frac12$ we have that $\frac{(1-y)-\sqrt{1-2y}}{y}\leq 1$.
	It is also easy to see that if $\eps\leq \delta^2/16$, then $y\leq \frac12$. Thus for $\eps\leq \delta^2/16$ we get that $x\leq 8/(\delta+\eps)$, and therefore $\nrm{UW_+-UW_-}=\sqrt{8/(\delta+\eps)}\nrm{A-\tilde{A}}\leq\sqrt{8/\delta}\nrm{A-\tilde{A}}$.
	
	Now we proceed similarly to the proof of Lemma~\ref{lem:PolyNormDiff}. Let $\Pi'=\widetilde{\Pi}$ if $n$ is odd, and let $\Pi'=\Pi$ for $n$ even. Let $\Phi$ be as in Corollary~\ref{cor:refAchievableP} and let $U^{(\pm)}:=UW_\pm$, then
	Theorem~\ref{thm:singValTransformation} implies that 
	$$\nrm{P^{(SV)}(A)-P^{(SV)}(\tilde{A})}
	=\nrm{\Pi'U^+_{\Phi}\Pi-\Pi'U^-_{\Phi}\Pi}
	\leq \nrm{U^+_{\Phi}-U^-_{\Phi}}
	\leq n\nrm{U^+-U^-}
	=
	n\sqrt{\frac{8}{\delta}}\nrm{A-\tilde{A}}.$$
	Finally note that $\eps\leq \delta^2/16$ is equivalent to $4\sqrt{\eps}\leq\delta$, which by definition is equivalent to $$\nrm{A-\tilde{A}}+\nrm{\frac{A+\tilde{A}}{2}}^2\leq 1.$$
\end{proof}

\subsection{Singular vector transformation and singular value amplification}\label{subsec:SingularvecTrans} 

In this subsection we derive some corollaries of singular value transformation. We call the first corollary projected singular vector transformation, because it implements a unitary that transforms the right singular vectors to the left singular vectors above some singular value threshold. Then we show how to quickly derive advanced amplitude amplification results using this general technique. Finally, we develop a corollary called singular value amplification, which shows how to uniformly amplify the singular values of a matrix represented as a projected unitary.

First we define singular value threshold projectors which are slight modifications of the singular value projectors of Definition~\ref{def:sinValProj}.
\begin{definition}[Singular value threshold projectors]\label{def:sinValThrProj}
	Let $A= \widetilde{\Pi}U\Pi=W \Sigma V^\dagger$ be a singular value decomposition of a projected unitary. 
	For $S\subseteq \R$ we define $\Pi_{S}:=\Pi V \Sigma_{S}V^\dagger \Pi$, and similarly $\widetilde{\Pi}_{S}:=\widetilde{\Pi} W \Sigma_{S}W^\dagger \widetilde{\Pi}$. 
	For $\delta\in \R$ we define $\Pi_{\geq\delta}:=\Pi_{[\delta,\infty)} $, also we define $\Pi_{>\delta},\Pi_{\leq \delta}, \Pi_{<\delta}, \Pi_{=\delta}$  and $\widetilde{\Pi}_{>\delta},\widetilde{\Pi}_{\leq \delta}, \widetilde{\Pi}_{<\delta}, \widetilde{\Pi}_{=\delta}$ analogously.
\end{definition}

Then we invoke a result of Low and Chuang~\cite[Corollary 6]{low2017HamSimUnifAmp} about constructive polynomial approximations of the sign function -- the error of the optimal approximation, studied by Eremenko and Yuditskii~\cite{eremenko2006uniformApxSgn},  achieves similar scaling but is non-constructive. 

\begin{lemma}[Polynomial approximations of the sign function]\label{lemma:signApx}
	For all $\delta>0$ , $\eps\in(0,1/2)$ there exists an efficiently computable odd polynomial $P\in\R[x]$ of degree $n=\bigO{\frac{\log(1/\eps)}{\delta}}$, such that 
	\begin{itemize}
		\item for all $x\in[-2,2]\colon |P(x)|\leq 1$, and
		\item for all $x\in[-2,2]\setminus(-\delta,\delta)\colon |P(x)-\sign{x}|\leq \eps$.
	\end{itemize}
\end{lemma}

Now we are ready to prove our result about singular value transformation. Our singular vector transformation implements a unitary which maps a right singular vector having singular value at least $\delta$ to the corresponding left singular vector.

\begin{theorem}[Singular vector transformation]\label{thm:singularVecTrans}
	Let $U,\Pi, \widetilde{\Pi}$ be as in Theorem~\ref{thm:singValTransformation} and let $\delta>0$. Suppose that $\widetilde{\Pi}U\Pi=W \Sigma V^\dagger$ is a singular value decomposition. Then there is an $m= \bigO{\frac{\log(1/\eps)}{\delta}}$ and a $\Phi\in\R^m$ such that 
	$\nrm{\widetilde{\Pi}_{\geq\delta}U_\Phi \Pi_{\geq\delta}-\widetilde{\Pi}_{\geq\delta}(WV^\dagger) \Pi_{\geq\delta}}\leq \eps.$  Moreover, $U_\Phi$ can be implemented using a single ancilla qubit with $m$ uses of $U$ and $U^\dagger$, $m$ uses of C$_\Pi$NOT and $m$ uses of C$_{\widetilde{\Pi}}$NOT gates and $m$ single qubit gates.
\end{theorem}
\begin{proof}
	By Lemma~\ref{lemma:signApx} we can construct an odd polynomial $P_\Re\in\R[x]$ of degree $m=\bigO{\frac{\log(1/\eps^2)}{\delta}}$ that approximates the sign function with $\eps^2/2$ precision on the domain $[-1,1]\setminus(-\delta,\delta)$. By Corollary~\ref{cor:realP} we know that there exists a polynomial $P$ of the same degree as $P_\Re$ such that $\Re[P]=P_\Re$, moreover $P$ satisfies the conditions of Corollary~\ref{cor:refAchievableP}. Use singular value transformation Theorem~\ref{thm:singValTransformation} to construct a $\Phi\in\R^m$ such that $\widetilde{\Pi}U_\Phi\Pi=P^{(SV)}\left(\widetilde{\Pi}U\Pi\right)$ up to precision $\eps$ and observe that $\nrm{\widetilde{\Pi}_{\geq\delta}P^{(SV)}\left(\widetilde{\Pi}U\Pi\right) \Pi_{\geq\delta}-\widetilde{\Pi}_{\geq\delta}(WV^\dagger) \Pi_{\geq\delta}}\leq \eps.$
	Conclude the gate complexity using Lemma~\ref{lemma:implementingPhasedSeq}.
\end{proof}

	\noindent As an easy corollary we recover and improve upon fixed-point amplitude amplification results~\cite{hoyer2000ArbitraryPhasesAmpAmp,grover2005FixedPointSearch,aaronson2012QMoneyHiddSubgrp,yoder2014FixedPointSearch} by combining the advantages of prior art. On one hand, the query complexity of $\mathcal{O}(\frac{1}{\delta}\operatorname{poly}(1/\eps))$ by~\cite{aaronson2012QMoneyHiddSubgrp} is optimal with respect to target state overlap $\delta$, but converges slowly with respect to error $\eps$. On the other hand, the query complexity of  $\mathcal{O}(\frac{1}{\delta}\log{(1/\eps)})$ by ~\cite{yoder2014FixedPointSearch} is optimal and exhibits exponentially fast convergence with respect to the error, but it introduces an unknown phase on the amplified state. Our presented approach has the same optimal asymptotic scaling and also ensures that this phase error is $\epsilon$-close to $0$.

\begin{theorem}[Fixed-point amplitude amplification]\label{thm:fixedPointAmp}
	Let $U$ be a unitary and $\Pi$ be an orthogonal projector such that $a\ket{\psi_G}=\Pi U\ket{\psi_0}$, and $a> \delta>0$.
	There is a unitary circuit $\tilde{U}$ such that $\nrm{\ket{\psi_G}-\tilde{U}\ket{\psi_0}}\leq \eps$, which uses a single ancilla qubit and consists of $\bigO{\frac{\log(1/\eps)}{\delta}}$ $U,U^\dagger$, C$_\Pi$NOT, C$_{\ketbra{\psi_0}{\psi_0}}$NOT and $e^{i\phi\sigma_z}$ gates.
\end{theorem}
\begin{proof}
	Set $\widetilde{\Pi}:=\Pi$ and $\Pi':=\ketbra{\psi_0}{\psi_0}$  and observe that 
	$$ \widetilde{\Pi}U\Pi'=a\ketbra{\psi_G}{\psi_0}.$$
	Now use Theorem~\ref{thm:singularVecTrans} in order to get an algorithm $\tilde{U}$ that satisfies 
	$$\nrm{\ketbra{\psi_G}{\psi_G}\tilde{U}\ketbra{\psi_0}{\psi_0}-\ketbra{\psi_G}{\psi_0}}\leq \eps.$$
\end{proof}

Another easy to derive corollary of our machinery is robust oblivious amplitude amplification.\footnote{Note that we could also easily derive a fixed-point version of oblivious amplitude amplification based on Theorem~\ref{thm:singularVecTrans}, but we state the usual version instead for readability.}

\begin{theorem}[Robust oblivious amplitude amplification]\label{cor:oblivious}
	Let $n\in\mathbb{N}_+$ be odd, let $\eps\in\R_+$, let $U$ be a unitary, let $\widetilde{\Pi},\Pi$ be orthogonal projectors, and let $W:\img{\Pi}\mapsto\img{\widetilde{\Pi}}$ be an isometry, such that
	\begin{equation}\label{eq:apxUnifNorm2}
	\nrm{\sin\left(\frac{\pi}{2n}\right)W\ket{\psi}-\widetilde{\Pi} U \ket{\psi}}\leq \eps
	\end{equation} 
	for all $\ket{\psi}\in\img{\Pi}$.
	Then we can construct a unitary $\tilde{U}$ such that for all $\ket{\psi}\in \img{\Pi}$
	\begin{equation*}
		\nrm{W\ket{\psi}-\widetilde{\Pi}\tilde{U} \ket{\psi}}\leq 2n\eps,
	\end{equation*} 	
	which uses a single ancilla qubit, with $n$ uses of $U$ and $U^\dagger$, $n$ uses of C$_\Pi$NOT and $n$ uses of C$_{\widetilde{\Pi}}$NOT gates and $n$ single qubit gates.
\end{theorem}
\begin{proof}
	First we prove the $\eps=0$ case. We prove this case by reproducing the polynomials stemming from ordinary amplitude amplification. Let $T_n\in\R[x]$ be the degree-$n$ Chebyshev polynomial of the first kind. As discussed after Corollary~\ref{cor:refAchievableP} there is an easy to describe $\Phi\in\R^n$ which corresponds to $T_n$ in equation~\eqref{eq:IterRef}.

	Now observe that by \eqref{eq:apxUnifNorm2} we have that $\widetilde{\Pi} U\Pi=\sin\left(\frac{\pi}{2n}\right)W$. We can apply singular value transformation using $T_n$ to obtain $U_\Phi$ such that 
	\begin{equation*}
		\widetilde{\Pi} U_\Phi\Pi=T_n\left(\sin\left(\frac{\pi}{2n}\right)\right)W
		=T_n\left(\cos\left(\frac\pi2-\frac{\pi}{2n}\right)\right)W
		=\cos\left(\frac{n-1}{2}\pi\right)W
		=\pm W.
	\end{equation*}
	After correcting the global phase $\pm 1$ (which depends on the parity of $(n-1)/2$) we get $\tilde{U}:=\pm U_\Phi$ such that for all $\ket\psi\in\img\Pi$ we have $\tilde{U}\ket{\psi}=W\ket{\psi}$. The complexity statement follows from Lemma~\ref{lemma:implementingPhasedSeq}.
	
	In the $\eps\neq 0$ case we first handle some trivial cases. If $n=1$ or $\eps>\frac13$ we simply take $\tilde{U}:=U$.
	Otherwise if $n\geq 3$ and $\eps\in[0,\frac13]$
	the error bounds follow from Lemma~\ref{lem:PolyNormDiff2}, in the following way: 
	Let $A:=\sin\left(\frac{\pi}{2n}\right)W$ and let $\tilde{A}:=\widetilde{\Pi} U\Pi$, by \eqref{eq:apxUnifNorm2} we have that $\nrm{A-\tilde{A}}\leq  \eps$. Then
	$$\nrm{A+\tilde{A}}\leq \nrm{A} +\nrm{A }+\nrm{\tilde{A}-A}=2\sin\left(\frac{\pi}{2n}\right)+\eps\leq 2\sin\left(\frac\pi6\right)+\frac13=\frac{4}{3},$$ thus $\nrm{\frac{A+\tilde{A}}{2}}^2\leq \frac{4}{9}$ and $\nrm{A-\tilde{A}} + \nrm{\frac{A+\tilde{A}}{2}}^2\leq  \frac79<1$. This also implies that $\sqrt{\frac{2}{1-\nrm{\frac{A+\tilde{A}}{2}}^2}}\leq \sqrt{\frac{2}{1-\frac49}}=\sqrt{\frac{18}{5}}<2$, and therefore by Lemma~\ref{lem:PolyNormDiff2} we get that $\nrm{W-\widetilde{\Pi} \tilde{U}\Pi}\leq 2n \eps$.
\end{proof}

Now we turn to solving the linear singular value amplification problem. That is, given a matrix in a projected encoding form, construct a projected encoding of a matrix which have singular values that are $\gamma$ times larger than the original singular values. 

In order to proceed we first construct some polynomials similarly that can be used in combination with our singular value transformation results.

\begin{lemma}[Polynomial approximations of the rectangle function]\label{lemma:polyRect}
	Let $\delta',\eps'\in(0,\frac12)$ and $t\in[-1,1]$. There exist an even polynomial $P'\in\R[x]$ of degree $\bigO{\log(\frac1{\eps'})/\delta'}$, such that $|P'(x)|\leq 1$ for all $x\in[-1,1]$, and
		\begin{equation}\label{eq:polyRect}
		\left\{\begin{array}{rcl} P'(x)\in & [0,\eps'] &\text{for all }x\in[-1,-t-\delta']\cup[t+\delta',1]\text{, and}\\
		P'(x)\in & [1-\eps',1] &\text{for all }x\in[-t+\delta',t-\delta'].\end{array}\right.
		\end{equation}
\end{lemma}
\begin{proof}
	First let us take a real polynomial $P$ which $\frac{\eps'}{2}$-approximates the sign function on the interval $[-2,2]\setminus(-\delta',\delta')$, moreover for all $x\in[-2,2]\colon |P(x)|\leq 1$. Such a polynomial of degree $\bigO{\frac{1}{\delta'}\log\left(\frac1{\eps'}\right)}$ can be efficiently constructed by Lemma~\ref{lemma:signApx}.
	Now take the polynomial $$P'(x):=(1-\eps')\frac{P\left(x+t\right)+P\left(-x+t\right)}{2} + \eps'.$$ It is easy to see that by construction $P'(x)$ is an even polynomial of degree $\bigO{\frac1{\delta'}\log\left(\frac1{\eps'}\right)}$. Moreover $|P'(x)|\leq 1$ for all $x\in[-1,1]$ and \eqref{eq:polyRect} also holds.
\end{proof}

Now we prove our result about uniform singular value amplification, which is a common generalization of the results of Low and Chuang \cite[Theroems 2,8]{low2017HamSimUnifAmp}.

\begin{theorem}[Uniform singular value amplification]\label{thm:singularValAmp}
	Let $U,\Pi, \widetilde{\Pi}$ be as in Theorem~\ref{thm:singValTransformation}, let $\gamma>1$ and let $\delta,\eps\in (0,\frac12)$. Suppose that $\widetilde{\Pi}U\Pi=W \Sigma V^\dagger=\sum_{i}\varsigma_i\ket{w_i}\bra{v_i}$ is a singular value decomposition. Then there is an $m= \bigO{\frac{\gamma}{\delta}\log\left(\frac\gamma\eps\right)}$ and an efficiently computable $\Phi\in\R^m$ such that\footnote{Here we implicitly assumed that $U_\Phi$ is implemented as in Figure~\ref{fig:qubitization}, with the phase gates as in Figure~\ref{fig:Piphi} and the the $\ket+$ ancilla state corresponds to the ancilla qubit in Figure~\ref{fig:Piphi}.\label{foot:plusQubitLocation}}
	\begin{equation}\label{eq:linearAmp}
		\left(\bra+\otimes\widetilde{\Pi}_{\leq\frac{1-\delta}{\gamma}}\right)U_\Phi \left(\ket+\otimes\Pi_{\leq\frac{1-\delta}{\gamma}}\right)=\sum_{i\colon\varsigma_i\leq \frac{1-\delta}{\gamma} }\tilde{\varsigma}_i\ket{w_i}\bra{v_i} , \text{ where }\nrm{\frac{\tilde{\varsigma}_i}{\gamma\varsigma_i}-1}\leq \eps.
	\end{equation}
	Moreover, $U_\Phi$ can be implemented using a single ancilla qubit with $m$ uses of $U$ and $U^\dagger$, $m$ uses of C$_\Pi$NOT and $m$ uses of C$_{\widetilde{\Pi}}$NOT gates and $m$ single qubit gates.
\end{theorem}
\begin{proof}
	Let us set in Lemma~\ref{lemma:polyRect} $t:=\frac{1-\delta/2}{\gamma}$, $\delta':=\frac{\delta}{2\gamma}$ and $\eps':=\frac\eps\gamma$ in order to get an even polynomial $P$ of degree $\bigO{\frac{\gamma}{\delta}\log\left(\frac\gamma\eps\right)}$ that is an $\frac\eps\gamma$-approximation of the rectangle function.
	Let $P_\Re(x):=\gamma\cdot x\cdot P(x)$, which is an odd polynomial of degree $m=\bigO{\frac{\gamma}{\delta}\log\left(\frac\gamma\eps\right)}$. It is easy to see that $P_\Re$ approximates the linear function $\gamma\cdot x$ with $\eps$-multiplicative-precision on the domain $\left[\frac{-1+\delta}{\gamma},\frac{1-\delta}{\gamma}\right]$, and observe that $|P_\Re(x)|\leq 1$ for all $x\in[-1,1]$, thereby it satisfies the requirements of Corollary~\ref{cor:matchingParity}. We use singular value transformation Corollary~\ref{cor:matchingParity} to construct a $\Phi\in\R^m$ such that
	$$(\bra+\otimes\widetilde{\Pi})U_\Phi\left(\ket+\otimes\Pi\right)=P_\Re^{(SV)}\!\left(\widetilde{\Pi}U\Pi\right)=\sum_{i}P_\Re(\sigma_i)\ket{w_i}\bra{v_i}$$
	which shows that equation \eqref{eq:linearAmp} is satisfied because $\frac{P_\Re(x)}{\gamma\cdot x}$ is $\eps$-close to $1$ on the domain $\left[\frac{-1+\delta}{\gamma},\frac{1-\delta}{\gamma}\right]$.
	We conclude the gate complexity using Lemma~\ref{lemma:implementingPhasedSeq}.
\end{proof}

Finally, note that if $\nrm{\Sigma}\leq\frac{1-\delta}{\gamma}$ in the above theorem then we get that $$\nrm{\gamma\widetilde{\Pi}U\Pi - (\bra+\otimes\widetilde{\Pi})U_\Phi\left(\ket+\otimes\Pi\right)}\leq \eps,$$
thereby this procedure gives an efficient way to magnify a projected unitary encoding.

\subsection{Singular value discrimination, quantum walks and the fast OR lemma}\label{subsec:singWalk}

First we show how to efficiently implement approximate singular value threshold projectors, which will be the main tool of this section.
\begin{theorem}[Implementing singular value threshold projectors]\label{thm:singularVecProjs}
	Let $U,\Pi, \widetilde{\Pi}$ be as in Theorem~\ref{thm:singValTransformation} and let $t,\delta>0$. Suppose that $\widetilde{\Pi}U\Pi=W \Sigma V^\dagger$ is a singular value decomposition. Then there is an $m= \bigO{\frac{\log(1/\eps)}{\delta}}$ and a $\Phi\in\R^m$ such that we have 
	$\nrm{\Pi_{\geq t+\delta}U_\Phi\Pi_{\geq t+\delta}-\Pi_{\geq t+\delta}}\leq \eps,$ and\textsuperscript{\ref{foot:plusQubitLocation}} $\nrm{\left(\bra{+}\otimes\Pi_{\leq t-\delta}\right)U_\Phi\left(\ket{+}\otimes\Pi_{\leq t-\delta}\right)}\leq \eps.$  Moreover, $U_\Phi$ can be implemented using a single ancilla qubit with $m$ uses of $U$ and $U^\dagger$, $m$ uses of C$_\Pi$NOT and $m$ uses of C$_{\widetilde{\Pi}}$NOT gates and $m$ single qubit gates.
\end{theorem}
\begin{proof}
	By Lemma~\ref{lemma:polyRect} we can construct an even polynomial $P_\Re\in\R[x]$ of degree $m=\bigO{\frac{\log(1/\eps^2)}{\delta}}$ that approximates the rectangle function with $\eps^2/2$ precision on the domain $[-1,1]\setminus(-t-\delta,-t+\delta)\cup(t-\delta,t+\delta)$. By Corollary~\ref{cor:realP} we know that there exists a polynomial $P$ of the same degree as $P_\Re$ such that $\Re[P]=P_\Re$, moreover $P$ satisfies the conditions of Corollary~\ref{cor:refAchievableP}. Use singular value transformation Theorem~\ref{thm:singValTransformation} to construct a $\Phi\in\R^m$ such that $\widetilde{\Pi}U_\Phi\Pi=P^{(SV)}\left(\widetilde{\Pi}U\Pi\right)$ up to precision $\eps$ and observe that 
	$\nrm{\Pi_{\geq t+\delta}U_\Phi\Pi_{\geq t+\delta}-\Pi_{\geq t+\delta}}\leq \eps,$ and $\nrm{\left(\bra{+}\otimes\Pi_{\leq t-\delta}\right)U_\Phi\left(\ket{+}\otimes\Pi_{\leq t-\delta}\right)}\leq \eps.$
	Conclude the gate complexity using Lemma~\ref{lemma:implementingPhasedSeq}.
\end{proof}

Note that the above complexity can be improved up to quadratically in terms of scaling with $\delta$, when the threshold $t$ is close to $1$, see, e.g., Lemma~\ref{cor:DCpx}. For the error of the optimal polynomial approximation of the step function see the results of Eremenko and Yuditskii~\cite{eremenko2010PolyApxSgnTwoIntervals}.

We define the singular value discrimination problem as to find out whether a given quantum state has singular value at most $a$ or it at least $b$. As we indicated above whenever $a$ and $b$ are $\bigO{|a-b|}$ close to $1$ we can get a quadratic improvement. A simple way to achieve this quadratic improvement is to perform singular value projection on the complementary singular values rather than on the original ones, by replacing the matrix $\widetilde{\Pi}U\Pi$ by the complementary projection $(I-\widetilde{\Pi})U\Pi$. 
\begin{theorem}[Efficient singular value discrimination]\label{thm:singValDiscr}
	Let $0\leq a<b \leq 1$, and let $A=\widetilde{\Pi}U\Pi$ be a projected unitary encoding. Let $\ket{\psi}$ be a given unknown quantum state, with the promise that $\ket{\psi}$ is a right singular vector of $A$ with singular value at most $a$ or at least $b$. Then we can distinguish the two cases with error probability at most $\eps$ using singular value transformation of degree $$\bigO{\frac{1}{\max[b-a,\sqrt{1-a^2}-\sqrt{1-b^2}]}\log\left(\frac{1}{\eps}\right)}.$$
	Moreover, if $a=0$ or $b=1$ we can make the error one sided.
\end{theorem}
\begin{proof}
	Let us assume that $b-a\geq \sqrt{1-a^2}-\sqrt{1-b^2}$. First we apply an $\sqrt{\eps}$-approximate singular value projector on $\ket{\psi}$ using Theorem~\ref{thm:singularVecProjs}, with choosing $t:=\frac{a+b}{2}$ and $\delta:=\frac{b-a}{2}$, at the end measuring the projector $\ketbra{+}{+}\otimes \Pi$. If we find the state in the image of $\ketbra{+}{+}\otimes \Pi$ we conclude that the singular value is at least $b$, otherwise we conclude that it is at most $a$. The correctness and the complexity follows from Theorem~\ref{thm:singularVecProjs}. If $a=0$ then we make the error one-sided by using singular vector transformation Theorem~\ref{thm:singularVecTrans} with setting $\delta:=b$, and measuring $\widetilde{\Pi}$ at the end. Similarly as before if we find the state in the image of $\widetilde{\Pi}$ we conclude that the singular value is at lest $b$, otherwise we conclude that it is $0$. The error becomes one-sided because Theorem~\ref{thm:singularVecTrans} uses an odd-degree singular value transformation which always preserves $0$ singular values.
	
	The proof of the $b-a< \sqrt{1-a^2}-\sqrt{1-b^2}$ case works analogously just changing $\widetilde{\Pi}$ to $\Pi':=I-\widetilde{\Pi}$ in the proof, which leads to $A':=\Pi'U\Pi$. It is easy to see by Lemma~\ref{def:SVDInvSubspaces} that $\ket{\psi}$ is a singular vector of $A'$ with singular value at least $\sqrt{1-a^2}$ in the first case or with singular value at most $\sqrt{1-b^2}$ in the second case. Also if $b=1$ we can make the error one sided since then the corresponding singular value of $\ket{\psi}$ with respect to $A'$ is $0$.
	
	Finally note that if $a=0$, then $b-a=b\geq 1-\sqrt{1-b^2}=\sqrt{1-a^2}-\sqrt{1-b^2}$,
	and if $b=1$, then $b-a=1-a\leq \sqrt{1-a^2}=\sqrt{1-a^2}-\sqrt{1-b^2}$,
	therefore we covered all cases.
\end{proof}

The above result can also be used when the input state is promised to be a superposition of singular values, with the promised bounds. Also in order to distinguish the two cases with constant success probability it is enough if most of the amplitude is on singular vectors with singular vectors satisfying the promise.

\subsubsection{Relationship to quantum walks}

Now we show how to quickly derive the quadratic speed-ups of Markov chain  based search algorithms using our singular value transformation and discrimination results. Before doing so we introduce some definitions and notation for Markov Chains.

Let $\PM\in\R^{n\times n}$ be a time-independent Markov chain on discrete state space $X$ with $|X|=n$, which sends an element $x\in X$ to $y\in X$ with probability $p_{xy}$, thereby $\PM$ is a row-stochastic matrix. We say that $\PM$ is \emph{ergodic} if for a large enough $t\in \N$ all elements of $\PM^t$ are non-zero. For an ergodic $\PM$ there exists a unique stationary distribution $\pi$ such that $\pi \PM = \pi$, and we define the \emph{time-reversed} Markov chain as $\PM^*:=\diag{\pi}^{-1}\cdot\PM^T\cdot\diag{\pi}$. We say that $\PM$ is \emph{reversible} if it is ergodic and $\PM^*=\PM$. For an ergodic Markov chain $\PM$ we define the discriminant matrix $D(\PM)$ such that its $xy$ entry is $\sqrt{p_{xy}p^*_{yx}}$, where $p^*_{yx}$ stands for entries of the time-reversed chain. It is easy to see that $$D(\PM)=\diag{\pi}^{\frac12}\cdot\PM\cdot\diag{\pi}^{-\frac12}.$$
This form has several important consequences. First of all the spectrum of $\PM$ and $D(\PM)$ coincide, moreover the vector $\sqrt{\pi}$, that we get from $\pi$ by taking the square root element-wise, is a left eigenvector of $D(\PM)$ with eigenvalue $1$. Also from the definition $\sqrt{p_{xy}p^*_{yx}}$ of the $xy$ entry it follows that for reversible Markov chains $D(\PM)$ is a symmetric matrix, therefore its singular values and eigenvalues coincide after taking their absolute value.

In the literature~\cite{szegedy2004QMarkovChainSearch,magniez2006SearchQuantumWalk,krovi2010QWalkFindMarkedAnyGraph} quantum walk based search methods are usually analyzed with the help of this discriminant matrix. Here we directly use the discriminant matrix as opposed to the associated quantum walk, significantly simplifying the analysis. Before deriving our versions of the Markov chain speed-up results we introduce some definitions regarding sets of marked elements. 

For a set of marked element $M\subseteq X$, we denote by $D_M(\PM)$ the matrix that we get after setting to zero the rows and columns of $D(\PM)$ corresponding to the marked elements. For an ergodic Markov chain $\PM$ we define the hitting time $HT(\PM,M)$ as the expected number of step of the Markov chain before reaching the first marked element, if started from the stationary\footnote{Here we follow the convention of Szegedy~\cite[Equation (15)]{szegedy2004QMarkovChainSearch}, and define hitting time without conditioning on the stationary distribution to unmarked vertices.} distribution $\pi$. We denote the probability that an element is marked in the stationary distribution by $p_M:=\sum_{x\in M}\pi_x$. Now we invoke some results about the connection between the hitting time and the discriminant matrix, which are proven for example in \cite[Proposition 2]{krovi2010QWalkFindMarkedAnyGraph} and~\cite[Lemma 10]{gilyen2014MScThesis}.

\begin{lemma}[Relationship between the hitting time and the discriminant matrix]\label{lemma:hittingTime}
	Let $\PM$ be a reversible Markov chain and $M$ a set of marked elements. Let $(v_i,\lambda_i)$ be the eigenvector-eigenvalue pairs of $D_M(\PM)$, then 
	\begin{align}
		HT(\PM,M)=\sum_{i=1}^{n}\frac{|\left\langle v_i, \sqrt{\pi}\right\rangle|^2}{1-\lambda_i} - p_M, & & \text{ and }& & 
		\sum_{i=1}^{n}\frac{|\left\langle v_i, \sqrt{\pi}\right\rangle|^2}{1-|\lambda_i|}
		\leq 2\sum_{i=1}^{n}\frac{|\left\langle v_i, \sqrt{\pi}\right\rangle|^2}{1-\lambda_i}
	\end{align} 
\end{lemma}

The following result shows how the presence of marked elements can be detected quadratically faster using singular value discrimination compared to using the corresponding classical Markov chain. A slightly less general version of this result was proven by Szegedy~\cite{szegedy2004QMarkovChainSearch}.

\begin{corollary}[Detecting marked elements in a reversible Markov chain]
	Let $\PM$ be a reversible Markov chain, and $M\subseteq X$ a set of marked elements.
	Let $U$ be a unitary and $\tilde{\Pi},\Pi$ orthogonal projectors. Suppose that $B,\tilde{B}$ are orthogonal bases, such that representing the matrix of $\tilde{\Pi} U \Pi$ in the bases $B\rightarrow \tilde{B}$ we have that 
	$$\tilde{\Pi} U \Pi=\left[\begin{array}{cc} D_M(\PM) & 0 \\ 0 & .\end{array}\right].$$
	Suppose that we are given a copy of $\ket{\pi}:=\sum_{x\in X}\sqrt{\pi_x}\ket{x}$, where $\ket{x}\colon x \in X$ are the first $n$ basis elements in $B$. Then we can distinguish with constant one sided error the case $\mathrm{HT}(P,M)\leq K$ from the case $M=\emptyset$ (i.e., $\mathrm{HT}(P,M)=\infty$) with singular value transformation of degree $\bigO{\sqrt{K+1}}$.
\end{corollary}
\begin{proof}
	Suppose that $M\neq \emptyset$. Let $(\ket{v_i},\lambda_i)\colon i\in[n]$ be the eigenvector and eigenvalue pairs of the $D_M(\PM)$ block of $\tilde{\Pi} U \Pi$. By Lemma~\ref{lemma:hittingTime} we have that $$\sum_{i=1}^{n}\frac{|\braket{v_i}{\pi}|^2}{1-|\lambda_i|}\leq 2\mathrm{HT}(P,M)+2p_M\leq 2(K+1).$$ 
	By Markov's inequality we have that 
	$$\sum_{i\colon |\lambda_i|\geq 1-\frac{1}{12(K+1)}}|\braket{v_i}{\pi}|^2\leq \frac{1}{6},$$
	and so $\nrm{\Pi_{\leq 1-\frac{1}{12(K+1)}}\ket{\pi}}^2\geq \frac{5}{6}$. On the other hand if $M=\emptyset$, then $D_M(\PM)=D(\PM)$, and $\nrm{D(\PM)\ket{\pi}}=1$. Therefore we can apply our singular value discrimination result Theorem~\ref{thm:singValDiscr} to distinguish the two cases $M=\emptyset$ and $\mathrm{HT}(P,M)\leq K$ using singular value transformation of degree $\bigO{\sqrt{K+1}}$.
\end{proof}

The above result shows how to detect the presence of marked elements quadratically faster than the classical hitting time. In practice one usually also wants to find a marked element, and quantum walks are also good at solving this problem. In order to show the connection to the literature of quantum walk based search algorithms we define some additional notation. 

First we define the standard implementation procedures for Markov chains together with their associated costs following Magniez et al.~\cite{magniez2006SearchQuantumWalk}. We slightly generalize the usual approach fitting our singular value transformation framework. Let us fix an orthogonal basis $B$, such that the first $n$ elements of $B$ are labeled by $\ket{x}_d\colon x\in X$. We define the following costs an operations with their matrices represented in the basis $B$.
\begin{itemize}
	\item[$\mathsf{U}$]: Update cost $\mathsf{U}$. The cost of implementing C$_\Pi$NOT and $U$ gates such that
	\begin{equation}\label{eq:generalU}
	\Pi U \Pi=\left[\begin{array}{cc} D(\PM) & 0 \\ 0 & .\end{array}\right].
	\end{equation}
	\item[$\mathsf{C}$]: Checking cost $\mathsf{C}$. The cost of implementing a C$_{\Pi_M}$NOT gate such that  for all $x\in M\colon \Pi_M\ket{x}_d=\ket{x}_d$ and for all $x\in  X\setminus M\colon \Pi_M\ket{x}_d=0$. This implies that 
	$$\left(I-\Pi_M\right)\Pi U \Pi\left(I-\Pi_M\right)=\left[\begin{array}{cc} D_M(\PM) & 0 \\ 0 & .\end{array}\right].$$	
	\item[$\mathsf{S}$]: Setup cost $\mathsf{S}$. The cost of preparing the stationary state in basis $B$: $\ket{\pi}:=\sum_{x\in X}\sqrt{\pi_x}\ket{x}_d$.
\end{itemize}

First we would like to describe how these operators are usually implemented in the literature. Usually $\PM$ is represented using basis elements $\ket{x}_d=\ket{0}\ket{x}\ket{d_x}$, where the $\ket{d_x}$ register stores some data associated with the vertex $x$, which enables efficient implementation of the update procedure. The unitary $U$ is usually implemented using a product of state preparation unitaries $U=U_L^\dagger U_R$:
\begin{align}
U_R:\ket{0}\ket{x}\ket{d_x}\rightarrow \sum_{y\in [n]}\sqrt{p_{xy}}\ket{x}\ket{y}\ket{d_{xy}} \qquad\forall x \in X\label{eq:walkR}\\
U_L:\ket{0}\ket{y}\ket{d_y}\rightarrow \sum_{x\in [n]}\sqrt{p^*_{yx}}\ket{x}\ket{y}\ket{d_{xy}} \qquad\forall y \in X\label{eq:walkL}
\end{align}
We assume for simplicity that $0\notin X$, resulting in a helpful ``free'' label, and also let us assume that the first register is $n+1$ dimensional and the second register is $n$ dimensional. If the third register is one-dimensional (i.e., we can just trivially omit it), by equations \eqref{eq:walkR}-\eqref{eq:walkL} we get that 
\begin{equation*}
(\ketbra{0}{0}\otimes I) U (\ketbra{0}{0}\otimes I)=\left[\begin{array}{cc} D(\PM)  & 0 \\ 0 & 0\end{array}\right].
\end{equation*}
 
If the data structure register is non-trivial, we can still conclude that 
\begin{equation*}
(\ketbra{0}{0}\otimes I) U (\ketbra{0}{0}\otimes I)=\left[\begin{array}{cc} D(\PM)  & . \\ . & .\end{array}\right],
\end{equation*}
however we need a slightly stronger assumption about $U$. We assume that $U_L,U_R$ are implemented such that the block-matrices next to $D(\PM)$ are $0$ as in \eqref{eq:generalU}. This is implicitly assumed\footnote{This assumption is necessary for the correctness of the analysis in~\cite{magniez2006SearchQuantumWalk}, however it is not explicitly stated.} in~\cite{magniez2006SearchQuantumWalk}, and a sufficient condition is presented in the work of Childs et al.~\cite{childs2013arXivTimeEfficientQW3Distintness}.

Before solving the above problem, we invoke a useful polynomial approximation result due to Dolph~\cite{dolph1946ChebWindowFunc}. 
\begin{lemma}[Optimal polynomial approximation of a windowing function on a bounded interval]\label{cor:DCpx}
	For all $\eps\in(0,1]$ and $n\in \N$ we have that\footnote{The standard generalization of Chebyshev polynomials to non-integer degree $y$ is $T_{y}(x)\equiv\cosh\left(y\arccosh\left(x\right)\right)\equiv\cos\left(y\arccos\left(x\right)\right)$.} 
	\begin{align*}
		\underset{P\in \R[x]}{\operatorname{argmax}}\left(\max\big\{\lambda\colon\nrm{P(x)}_{[-\lambda,\lambda]}\leq \eps\big\}\right)=T_n(x T_{1/n}(1/\eps)),
	\end{align*}
	where $\operatorname{argmax}$ is over all real degree-$n$ polynomials satisfying
		$\nrm{P}_{[-1,1]}\leq 1, \text{ and } P(\pm 1)=(\pm 1)^n.$
	Moreover, for any $\delta\in(0,1)$ for some $n=\bigO{\frac{1}{\sqrt{\delta}}\log(\frac1\eps)}$ we have that 
	$$\nrm{T_n(x T_{1/n}(1/\epsilon))}_{[-1+\delta,1-\delta]}\le\eps.$$
\end{lemma}

Notably, the phase sequence required to implement this polynomial using quantum signal processing is expressed in closed-form in the work of Yoder et al.~\cite{yoder2014FixedPointSearch}.

\begin{theorem}[Quadratic speed-up for finding marked elements of a Markov chain]
	Let $\PM$ be a reversible Markov chain, such that the singular value gap\footnote{We define the singular value gap as the difference between the two largest singular values. For a reversible Markov chain the singular values of $D(\PM)$ are the same as the absolute values of the eigenvalues of $\PM$. Note however, that this is not strictly necessary, we could work with the eigenvalues too using eigenvalue transformation results, such as Theorem~\ref{thm:arbParity}.} of $D(\PM)$ is at least $\delta$, and the set of marked elements $M$ is such that $p_M\geq \eps$. Then we can find a marked element with high probability in complexity of order $\mathsf{S}+\frac{1}{\sqrt{\eps}}\left(\mathsf{C}+\sqrt{\frac{1}{\delta}}\log\left(\frac{1}{\eps}\right)\mathsf{U}\right)$.
\end{theorem}
\begin{proof}
	First we apply singular value transform on $\Pi U\Pi$ using an $\eps$-approximation of the zero-function given by Lemma~\ref{cor:DCpx} in order to get $U_\Phi$ with all $\neq 1$ singular values of $D(\PM)$ shrunk below a level of $\bigO{\eps}$. Then the top-left block of $\Pi U_\Phi\Pi$ is $\bigO{\eps}$-close to $\ketbra{\pi}{\pi}$, and the implementation of $U_\Phi$ has complexity $\bigO{\sqrt{\frac{1}{\delta}}\log\left(\frac{1}{\eps}\right)\mathsf{U}}$. We pretend that the top-left block is $\ketbra{\pi}{\pi}$, in which case we seem to have that $\Pi_M\Pi U_\Phi \Pi=\sqrt{p_M}\ketbra{\pi_M}{\pi}$, where $\ket{\pi_M}:=\frac{\sum_{x\in M}\sqrt{\pi_x}\ket{x}_d}{\sqrt{p_M}}$. Then we apply singular vector transform to get a constant approximation of $\ketbra{\pi_M}{\pi}$ in the top-left block, and apply it to the state $\ket{\pi}$ in order to find a marked element with high probability. Finally, we use the robustness of singular value transformation Lemma~\ref{lem:PolyNormDiff} to show that we can indeed dismiss the $\bigO{\eps}$ discrepancy between $\Pi U_\Phi\Pi$ and  $\ketbra{\pi}{\pi}$.
\end{proof}

Note that the above algorithm is simpler and more efficient than the phase estimation based algorithm of~\cite{magniez2006SearchQuantumWalk}. However note, that Magniez et al.~\cite{magniez2006SearchQuantumWalk} showed how to remove the $\log(\frac1\eps)$ factor completely using a more involved procedure.

Finally note that it is known that a unique marked element can be found using a quantum walk quadratically faster than the hitting time. However, it is an open question whether in the presence of multiple marked elements the quadratic advantage can be retained.\footnote{One can show that $p_M=\Omega\left(\frac{1}{HT(\PM,M)}\right)$, therefore using amplitude amplification one can find a marked element quadratically faster, however with a large $\mathsf{S}$ cost. The appeal of the quantum walk algorithms is that they only use the setup procedure a very few times.} For more details see the work of Krovi et al.~\cite{krovi2010QWalkFindMarkedAnyGraph}. They use an interpolated matrix between $D(\PM)$ and $D_M(\PM)$ -- which is an idea very naturally fitting our framework, see for example Lemma~\ref{lem:linCombBlocks}. We believe that the algorithms of Krovi et al.~\cite{krovi2010QWalkFindMarkedAnyGraph} for finding marked elements can also be cast in our singular value transformation framework, but we leave the discussion of these algorithms for future work.

\subsubsection{Fast \texorpdfstring{$\mathsf{QMA}$}{QMA} amplification and fast quantum OR lemma}

We show how the fast $\mathsf{QMA}$ amplification result of Nagaj et al.~\cite{nagaj2009FastAmpQMA} follows directly from our singular value discrimination results. In order to state the result we invoke the definition of the language class $\mathsf{QMA}$.

\begin{definition}[The language class $\mathsf{QMA}$]\label{def:QMA}
	Let $L\subseteq \{0,1\}^*$ be a language of yes and no instances $L=L_{\text{yes}}\cupdot L_{\text{yes}}$. The language $L$ belongs to the class $\mathsf{QMA}$ if there exists a uniform family of quantum verifier circuits $V$ working on $n=\poly{|x|}$ qubits using $m=\poly{|x|}$ ancillae and two numbers $0\leq b < a\leq 1$ satisfying $\frac{1}{a-b}=\bigO{\poly{|x|}}$ such that for all $x$ in
	\begin{itemize}
		\item[$L_{\text{yes}}:$] there exists an $n$-qubit witness $\ket{\psi}$ such that upon measuring the state $V\ket{\psi}\ket{0}^{m}$ the probability of finding the first qubit in state $\ket{1}$ has probability at least $a$.
		\item[$L_{\text{no}}:$] for any $n$-qubit state $\ket{\phi}$ upon measuring the state $V\ket{\phi}\ket{0}^{m}$ the probability of finding the first qubit in state $\ket{1}$ has probability at most $b$.
	\end{itemize}
\end{definition}

Now we are ready to reprove the result of Nagaj et al.~\cite{nagaj2009FastAmpQMA}:

\begin{theorem}[Fast $\mathsf{QMA}$ amplification]
	Suppose that we have a language in $\mathsf{QMA}$ as in Definition~\ref{def:QMA}. We can modify the verifier circuit $V$ such that the acceptance probabilities become $a':=1-\eps$ and $b':=\eps$ using singular value transformation of degree $\bigO{\frac{1}{\max[\sqrt{a}-\sqrt{b},\sqrt{1-b}-\sqrt{1-a}]}\log\left(\frac{1}{\eps}\right)}$.
\end{theorem}
\begin{proof}
	Observe that by Definition~\ref{def:QMA} for all $x\in L_{\text{yes}}$ we have that 
	$$\nrm{\left(\ketbra{1}{1}\otimes I_{n+m-1}\right)V\left(I_n\otimes\ketbra{0}{0}^{\otimes m}\right)}\geq \sqrt{a},$$
	and for all $x\in L_{\text{no}}$ we have that  
	$$\nrm{\left(\ketbra{1}{1}\otimes I_{n+m-1}\right)V\left(I_n\otimes\ketbra{0}{0}^{\otimes m}\right)}\leq \sqrt{b}.$$
	After applying a singular value discrimination circuit for discriminating singular values below $\sqrt{b}$ and above $\sqrt{a}$ we get a circuit that in the former case accepts some witness $\ket{\psi}$ with probability at least $1-\eps$ and in the latter case rejects every state $\ket{\phi}$ with probability at least $1-\eps$.
\end{proof}

Finally we show how to quickly derive a slightly improved version of the fast quantum OR lemma of Brand{\~a}o et al.~\cite{brandao2017QSDPSpeedupsLearning}. We use the main ideas of the proof of the original quantum OR lemma of Harrow et al.~\cite{harrow2017SeqMeasPropTest} in a way similar to the approach of Brand{\~a}o et al.~\cite{brandao2017QSDPSpeedupsLearning}.

\begin{theorem}[Fast quantum OR lemma]\label{lem:fastQuantumOR}
	Let $m\in\N$, let $\Pi_i\colon i\in [m]$ be orthogonal projectors and let $\eta,\nu \in (0,\frac{1}{2}]$. Suppose we are given one copy of a quantum state $\rho$ with the promise that either 
	\begin{enumerate}[label=(\roman*)]
		\item there exists some $i\in [m]$ such $\tr{\rho \Pi_i}\geq 1-\eta$, or \label{it:bigProj}
		\item $\frac{1}{m}\sum_{j=1}^{m}\Tr[\rho\Pi_j]\leq\nu$.  \label{it:smallAvgProj}
	\end{enumerate}
	Suppose that we have access\footnote{Brand{\~a}o et al.~\cite{brandao2017QSDPSpeedupsLearning} assumes access to a multiply-controlled reflection operator instead of $V$, but it is easy to see that such an operator can be easily transformed to the operator required here using a single qubit by applying phase-kickback.} to an operator $V$ such that $(\bra{i}\otimes I)V(\ket{i}\otimes I)=$C$_{\Pi_i}$NOT for all $i\in [m]$. Then for all $\eps\in(0,\frac{1}{2}]$ we can construct an algorithm which, in case \ref{it:bigProj} accepts $\rho$ with probability at least $\frac{(1-\eta)^{2}}{4} - \eps$, and in case \ref{it:smallAvgProj} it accepts $\rho$ with probability at most $5m\nu + \eps$.
	Moreover, the algorithm uses $V$ and its inverse a total number of $\bigO{\sqrt{m}\log\left(\frac{1}{\eps}\right)}$ times and uses $\bigO{\sqrt{m}\log(m)\log\left(\frac{1}{\eps}\right)}$ other gates and $\bigO{\log(m)}$ ancilla qubits.
\end{theorem}
\begin{proof}
	Let us define $A:=\frac{1}{m}\sum_{i=1}^{m}(I-\Pi_i)$. First observe that $I-\Pi_i=\left(\ketbra{0}{0}\otimes I\right)$C$_{\Pi_i}$NOT$\left(\ketbra{0}{0}\otimes I\right)$. Let $a:=\lceil\log_2(m+1)\rceil+1$ and let $U$ be a unitary that implements the map $\ket{0}^{a-1}\rightarrow \frac{1}{\sqrt{m}}\sum_{i=1}^{m}\ket{i}$, and let us define $\tilde{V}:=\left(U^\dagger\otimes I\right)V\left(U\otimes I\right)$ and $\Pi:=\ketbra{0}{0}^a\otimes I$. Then it is easy to see that $A=\Pi\tilde{V}\Pi$.
	
	Now let $\lambda:=\frac{1-\eta}{2m}$, in case \ref{it:bigProj} Harrow et al.~\cite[Corollary 11]{harrow2017SeqMeasPropTest} proved that
	\begin{equation}\label{eq:projLowerBound}
		\tr{\rho \Pi_{\leq 1-\lambda}}\geq (1-\eta)^2/4.
	\end{equation}
		
	On the other hand in case \ref{it:smallAvgProj} we have that $\tr{\rho A}\geq 1-\nu$. Using Markov's inequality we get
	\begin{equation}\label{eq:projUpperBound}
		\tr{\rho \Pi_{\leq 1- \frac{4}{5}\lambda}}\leq \frac{\nu}{\frac{4}{5}\lambda}=\frac{5m\nu}{2(1-\eta)}\leq 5m\nu.
	\end{equation}

	Finally, we apply $\eps$-precise singular value discrimination on $\rho$ with $a:=1-\lambda$ and $b:=1-\frac{4}{5}\lambda$.
	The correctness follows from \eqref{eq:projLowerBound}-\eqref{eq:projUpperBound} and Theorem~\ref{thm:singValDiscr}. The complexity statement follows from Theorem~\ref{thm:singValDiscr}, Lemma~\ref{lemma:implementingPhasedSeq} and the fact that $U$ can be implemented using $\bigO{\log(m)}$ one- and two-qubit gates.
\end{proof}

\subsection{``Non-commutative measurements'' and singular value estimation}\label{subsec:NonCommMeasAndSVE} 

Preparing ground states of local Hamiltonians is a notoriously hard problem. However, under the conditions of the quantum Lovász Local Lemma, the local Hamiltonian is guaranteed to be frustration-free as shown by Ambainis et al.~\cite{ambainis2009QLLL}. As shown by Sattath and Arad~\cite{sattath2013ConstrQLLLComm} and Schwarz et al.~\cite{schwarz2013QInfoProofCommQLLL} under same conditions the problem becomes efficiently solvable when the local Hamiltonian terms commute. The non-commuting case is more difficult, but under a gap-promise Gilyén and Sattath~\cite{gilyen2016PrepGapHamEffQLLL} showed that a ground state can be efficiently prepared.

Gilyén and Sattath~\cite{gilyen2016PrepGapHamEffQLLL} essentially reduced the state preparation problem to solving the following task: Given two (non-commuting) orthogonal projectors $\Pi^F$ and $\Pi_c$ and quantum state $\ket{\psi}\in \img{\Pi^F}$ perform a ``non-commutative'' measurement in the following sense. If $\ket{\psi}\in \img{\Pi^F}\cap \ker(\Pi_c)$ then output $0$ and leave the state intact, otherwise if $\ket{\psi}$ is a right singular vector of $\Pi_c\Pi^F$ with singular value greater than $0$, then output $1$ and ``rotate'' the state $\ket{\psi}$ to the corresponding left singular vector of $\Pi_c\Pi^F$ which in turn lies in $\img{\Pi_c}$. They showed how to implement such a quantum channel using a combination of weak measurements and the quantum Zeno effect, however their quantum channel does not preserve coherence between singular vectors with different singular values. The complexity of their implementation is essentially $\bigO{\frac{\log(1/\eps)}{\eps \varsigma^2}}$, where $\varsigma$ is the smallest non-zero singular value of $\Pi_c\Pi^F$, and $\eps$ is the desired maximum failure probability. 

Gilyén and Sattath~\cite{gilyen2016PrepGapHamEffQLLL} called exact quantum channel the procedure which solves the above problem, but also preserves coherence between singular vectors with different singular values. In their paper it was unclear how to efficiently implement such a ``non-commutative'' measurement. However note, that the techniques developed in this paper result in an efficient implementation. Indeed, by setting $A:= \Pi_c\Pi^F$ this task can be solved with maximal failure probability $\eps$ using singular value transformation Theorem~\ref{thm:singularVecTrans}, with $\bigO{\frac{\log(1/\eps)}{\varsigma}}$ uses of C$_{\Pi^F}$NOT, C$_{\Pi_c}$NOT and other two-qubit gates. This improves the $\eps$ dependence exponentially and improves the $\varsigma$ dependence quadratically, while solves a qualitatively stronger problem. These improvements also greatly improve the final complexity of the main algorithm presented in~\cite{gilyen2016PrepGapHamEffQLLL}.

Finally, we turn to the singular value estimation results of Kerenidis an Prakash~\cite{kerenidis2016QRecSys}. Kerenidis and Prakash mostly use singular value estimation in order to implement singular vector projectors, with similar complexity to that of Theorem~\ref{thm:singularVecProjs}. However, to our knowledge, there is an unresolved issue in their implementation procedure, stemming from the ambiguity of the phase labels produced by phase estimation. Using our techniques combined with ideas of Chakraborty et al.~\cite{chakraborty2018BlockMatrixPowers}, we show an alternative approach to singular value estimation.

Suppose that $A=\widetilde{\Pi} U \Pi$, and we would like to perform singular value estimation of $A$. The main idea is to first implement an operator $V$ such that $\left(I\otimes \Pi \right) V \left(I\otimes \Pi \right)=\sum_{t=0}^{2^n-1}\ketbra{t}{t}\otimes T_{2t}^{(SV)}(A)$. This can be done by using controlled quantum walk steps, i.e., using controlled alternating phase modulation sequences with phases as in Lemma~\ref{lemma:chebyshev}. 
Suppose that $\ket{\psi_j}\in\img{\Pi}$ is a right singular vector of $A$ with singular value $\cos(\theta_j)$, then 
$$\left(I\otimes \Pi \right) V \left(H^{\otimes n}\otimes I\right)\ket{0}^n\ket{\psi_j}
=\left(I\otimes \Pi \right) V \frac{1}{\sqrt{2^n}}\sum_{t=0}^{2^n-1}\ket{t}\ket{\psi_j}
=\frac{1}{\sqrt{2^n}}\sum_{t=0}^{2^n-1}\cos(2t\theta_j)\ket{t}\ket{\psi_j}.$$
One can show that the norm of this state $N_j$ is always bigger than some constant $c$. However, the problem is that the norm $N_j$ depends on the singular value $\cos(\theta_j)$. Fortunately we can use singular vector transformation using the projected unitary encoding $\left(I\otimes \Pi \right) V \left(H^{\otimes n}\otimes I\right)\left(\ketbra{0}{0}^{\otimes n}\otimes \Pi\right)$ by Theorem~\ref{thm:singularVecTrans} in order to $\eps$-approximate the map $\ket{0}^n\ket{\psi_j}\rightarrow \frac{1}{\sqrt{2^n}}\sum_{j=0}^{2^n-1}\frac{\cos(2t\theta_j)}{N_j}\ket{t}\ket{\psi_j}$. Applying a Fourier transform on the first (time) register, and taking half of the absolute value of the resulting estimation solves the singular value estimation problem. The correctness can be seen using the usual analysis of quantum phase estimation~\cite{cleve1997QAlgsRevisited} by utilizing the identity $\cos(x)=\frac{e^{i x}+e^{-i x}}{2}$. For more details see~\cite[version 2]{chakraborty2018BlockMatrixPowers}.

\subsection{Direct implementation of the Moore-Penrose pseudoinverse}\label{subsec:DirectPseudoInverse} 
Suppose $\widetilde{\Pi}U\Pi=A$ and $A=W \Sigma V^\dagger$ is a singular value decomposition. Then the pseudo-inverse of $A$ is $A^+=V\Sigma^{+}W^\dagger$, where $\Sigma^{+}$ contains the inverses of the diagonal elements of $\Sigma$, except for $0$ entries which remain $0$.

Now it is pretty straightforward to proceed using our singular value transformation methods. 
Suppose that all non-zero singular values are at least $\delta$. Let $P_\Re$ be an odd real polynomial that $\eps$-approximates the function $\delta/(2x)$ on the domain $[-1,1]\setminus(-\delta,\delta)$, then $P_\Re^{(SV)}(A^\dagger)=\Pi U^\dagger_\Phi\widetilde{\Pi}$ \linebreak $\eps$-approximates $\frac{\delta}{2} A^{+}$.
The only thing remaining is to construct a relatively low-degree odd polynomial $P_\Re$ that achieves the desired approximation, and which is bounded by $1$ in absolute value on $[-1,1]$, in order to be able to apply Corollary~\ref{cor:matchingParity}. Childs et al.~\cite[Lemmas 17-19]{childs2015QLinSysExpPrec} constructed a useful polynomial for improving the HHL algorithm, which we can use after some adjustments.
\begin{lemma}\emph{(Polynomial approximations of $1/x$, \cite[Lemmas 17-19]{childs2015QLinSysExpPrec})}\label{lemma:clidsPoly}
	Let $\kappa>1$ and $\eps\in(0,\frac12)$. For $b = \left\lceil\kappa^2\log(\kappa/\eps)\right\rceil$ the odd function 
	$$f(x) := \frac{1 - (1 - x^2)^b}{x}$$
	is $\eps$-close to $1/x$ on the domain $[-1,1]\setminus(-\frac1\kappa,\frac1\kappa)$.	
	Let $J:= \left\lceil\sqrt{b\log(4b/\eps)}\right\rceil$, then the $\bigO{\kappa\log(\frac\kappa\eps)}$-degree odd real polynomial  
	$$g(x) := 4\sum_{j=0}^{J}(-1)^j\left[\frac{\sum_{i=j+1}^{b}\binom{2b}{b+i}}{2^{2b}}\right]	T_{2j+1}(x)$$
	is $\eps$-close to $f(x)$ on the interval $[-1, 1]$, moreover $|P(x)|\leq 4 J=\bigO{\kappa\log(\frac\kappa\eps)}$ on this interval.
\end{lemma}
\begin{theorem}[Implementing the Moore-Penrose pseudoinverse]\label{thm:pseudoinverse}
	Let $U,\Pi, \widetilde{\Pi}$ be as in Theorem~\ref{thm:singValTransformation} and let $0<\eps\leq\delta\leq \frac12$. Suppose that $A=\widetilde{\Pi}U\Pi=W \Sigma V^\dagger$ is a singular value decomposition. Let $\Pi_{0,\geq\delta}:=\Pi_{=0}+\Pi_{\geq\delta}$ and similarly $\widetilde{\Pi}_{0,\geq\delta}:=\widetilde{\Pi}_{=0}+\widetilde{\Pi}_{\geq\delta}$. Then there is an $m= \bigO{\frac{1}{\delta}\log(\frac{1}{\eps})}$ and an efficiently computable $\Phi\in\R^m$ such that\textsuperscript{\ref{foot:plusQubitLocation}}
	\begin{equation}\label{eq:pseudoinverse}
	\nrm{\Big(\bra+\otimes\Pi_{0,\geq\delta}\Big)U_\Phi \Big(\ket+\otimes\widetilde{\Pi}_{0,\geq\delta}\Big)-\Pi_{0,\geq\delta} \left(\frac{\delta}{2}\cdot A^+\!\right) \widetilde{\Pi}_{0,\geq\delta}}\leq \eps.
	\end{equation}
	Moreover, $U_\Phi$ can be implemented using a single ancilla qubit with $m$ uses of $U$ and $U^\dagger$, $m$ uses of C$_\Pi$NOT and $m$ uses of C$_{\widetilde{\Pi}}$NOT gates and $m$ single qubit gates.
\end{theorem}
\begin{proof}
	Using Lemma~\ref{lemma:clidsPoly} we can construct an odd polynomial $P(x)$ of degree $\bigO{\log(1/\eps)/\delta}$ that $\frac\eps3$-approximates the function $\frac{\delta}{2x}$ on the domain $[-1,1]\setminus(-\frac\delta2,\frac\delta2)$, and is less than $1$ on this domain. Let us define $P_{\max}:=\max_{x\in[-1,1]}|P(x)|$ and observe that $P_{\max}=\bigO{\log(\frac1\eps)}$. Let us also construct an even polynomial $P'$ of degree $\bigO{\log(1/\eps)/\delta}$ using Lemma~\ref{lemma:polyRect} setting $t:=\frac34\delta$, $\delta':=\frac\delta4$ and $\eps':=\min\left(\frac\eps3,\frac{1}{P_{\max}}\right)$ that $\eps'$-approximates the rectangle function. Finally let $P_\Re:=P\cdot (1-P')$, which is and odd real polynomial of degree $m=\bigO{\log(1/\eps)/\delta}$. It is easy to see that $P_\Re$ $\eps$-approximates $\frac{\delta}{2x}$ on the domain $[-1,1]\setminus(-\frac\delta2,\frac\delta2)$, moreover $P_\Re$ is bounded by $1$ in absolute value on $[-1,1]$. Finally, we apply real singular value transformation on $A^\dagger=\Pi U^\dagger\widetilde{\Pi}$ using the polynomial $P_\Re$ by Corollary~\ref{cor:matchingParity}, and conclude the gate complexity using Lemma~\ref{lemma:implementingPhasedSeq}.
\end{proof}

Note that the $\eps\leq\delta$ assumption in the above statement is quite natural, but it is not necessary, and can be removed by using our general polynomial approximation results, see Corollary~\ref{cor:oneOverX}.

\subsection{Applications in quantum machine learning}
\label{subsec:QML}
The ability to transform singular values is central to the operation of many popular machine learning methods.  Quantum machine learning methods such as quantum support vector machines~\cite{rebentrost2014QSVM}, principal component analysis~\cite{lloyd2013QPrincipalCompAnal,wiebe2017hardeningQMLAgainstAdv}, regression~\cite{harrow2009QLinSysSolver,wiebe2012QDataFitting,childs2015QLinSysExpPrec,chakraborty2018BlockMatrixPowers}, slow feature analysis~\cite{kerenidis2018SlowFeatureAnalysis}, Gibbs sampling~\cite{chowdhury2016QGibbsSampling,apeldoorn2017QSDPSolvers} and in turn training Boltzmann machines~\cite{amin2016QBoltzMachine,kieferova2016tomographyTrainQBoltzMachine} all hinge upon this idea.  These results are among the most celebrated in quantum machine learning, showing that singular value transformation has substantial impact on this field of quantum computing.  

Many quantum algorithms for basic machine learning problems, such as ordinary least squares, weighted least squares, generalized least squares, were studied in a series of works \cite{harrow2009QLinSysSolver,wiebe2012QDataFitting,childs2015QLinSysExpPrec,chakraborty2018BlockMatrixPowers}. We do not examine these problems case-by-case, but point out that they can all be reduced to implementing the Moore-Penrose pseudoinverse and matrix multiplication, therefore they can be straightforwardly implemented by Theorem~\ref{thm:pseudoinverse} and Lemma~\ref{lemma:disjointAncillaProduct} (to be discussed in Subsection~\ref{subsec:mult}) within our framework.

We demonstrate how to apply our singular value transformation techniques to quantum machine learning by developing a new quantum algorithm for principal component regression.  This machine learning algorithm is closely related to principal component analysis (PCA), which is a tool that is commonly used to reduce the effective dimension of a model by excising irrelevant features from it.  The PCA method is quite intuitive, it simply involves computing the covariance matrix for a data set and then diagonalizing it.  The eigenvectors of the covariance matrix then represent the independent directions of least or greatest variation in the data.  Dimension reduction can be achieved by culling out any components that have negligibly small variation over them.  This technique has many applications ranging from anomaly detection to quantitative finance.  Quantum algorithms are known for this task and can lead to substantial speed-ups under appropriate assumptions about the data and the oracles used to provide it to the algorithm~\cite{lloyd2013QPrincipalCompAnal,kimmel2016hamiltonian}.

Principal component regression is identical in spirit to principal component analysis.  Rather than simply truncating small eigenvalues of the covariance matrix for a data set, principal component regression aims to approximately reconstruct a target vector on a domain/range that is spanned by the right or left singular vectors with large singular values.  A least-squares type estimation of the target vector within this subspace of the data can be found by performing a singular vector transformation.  This can provide a more flexible and powerful approach to dimensionality reduction than ordinary principal component analysis permits.

	 The problem of principal component regression can be formally stated as follows \cite{frostig2016PrincipalCompProjwoPCA}: given a matrix $A\in\R^{n\times d}$, a vector $b\in\R^{n}$ and a threshold value $0<\varsigma$, find $x\in\R^d$ such that
	\begin{equation}\label{eq:defPCR}
	x=\mathrm{argmin}_{x\in\R^d}\nrm{\widetilde{\Pi}_{\geq\varsigma}A\Pi_{\geq\varsigma}x-b},
	\end{equation}
	where $\widetilde{\Pi}_{\geq\varsigma}, \Pi_{\geq\varsigma}$ denote left and right singular value threshold projectors. A closed-form expression for the optimal solution of \eqref{eq:defPCR} is given by $x=\Pi_{\geq\varsigma}A^{+}\widetilde{\Pi}_{\geq\varsigma}b=A^{+}\widetilde{\Pi}_{\geq\varsigma}b$. 
	
	As the following corollary shows, our singular value transformation techniques give rise to an efficient quantum algorithm for implementing $\Pi_{\geq\varsigma}A^{+}\widetilde{\Pi}_{\geq\varsigma}$, and thus principal component regression. 
	
	\begin{corollary}[Implementing the threshold pseudoinverse]\label{cor:thresholdinverse}
		Let $U,\Pi, \widetilde{\Pi}$ form a projected unitary encoding of the matrix $A$, and let $\eps,\delta\in(0,\frac12]$ and $0<\varsigma<1$. Suppose that $A=\widetilde{\Pi}U\Pi=W \Sigma V^\dagger$ is a singular value decomposition of the projected unitary encoding of $A$. 
		Then there is an $m= \bigO{\frac{1}{\delta}\log(\frac{1}{\eps})}$ and an efficiently computable $\Phi\in\R^m$ such that\textsuperscript{\ref{foot:plusQubitLocation}}
		\begin{equation}\label{eq:thresholdinverse}
		\nrm{\Big(\bra+\otimes\big(\Pi-\Pi_{[\varsigma-\delta,\varsigma+\delta]}\big)\Big)U_\Phi \Big(\ket+\otimes\big(\widetilde{\Pi}-\widetilde{\Pi}_{[\varsigma-\delta,\varsigma+\delta]}\big)\Big)-\Pi_{\geq\varsigma} \left(\frac{\varsigma}{2}A^+\right) \widetilde{\Pi}_{\geq\varsigma}}\leq \eps.
		\end{equation}
		Moreover, $U_\Phi$ can be implemented using a single ancilla qubit with $m$ uses of $U$ and $U^\dagger$, $m$ uses of C$_\Pi$NOT and $m$ uses of C$_{\widetilde{\Pi}}$NOT gates and $m$ single qubit gates.
	\end{corollary}
	\begin{proof}
	We can implement this operator by first applying a singular value threshold projector $\widetilde{\Pi}_{\geq\varsigma}$ according to Theorem~\ref{thm:singularVecProjs}, followed by performing the Moore-Penrose pseudoinverse $A^+$ as in Theorem~\ref{thm:pseudoinverse}.
	
	Implementing these two operations separately is actually suboptimal. In order to get the stated result we simply take the polynomials used for singular value transformation in Theorem~\ref{thm:singularVecProjs} and Theorem~\ref{thm:pseudoinverse}, then take their product and implement singular value transformation according to the product polynomial. The complexity statement can be proven similarly to the proofs of Theorem~\ref{thm:singularVecProjs} and Theorem~\ref{thm:pseudoinverse}.
	\end{proof}
	
	Given a unitary preparing a quantum state $\ket{b}$ we can approximately solve principal component regression by applying an approximation of $\Pi_{\geq\varsigma}\left(\frac{\varsigma}{2}A^+\right)\widetilde{\Pi}_{\geq\varsigma}$ to $\ket{b}$, and then applying amplitude amplification to get $\ket{x}$. Strictly speaking, in order for this to work as required by \eqref{eq:defPCR}, we would need to have the promise that $\ket{b}$ does not have an overlap with left singular vectors that have eigenvalues in $[\varsigma-\delta,\varsigma+\delta]$, while it does have a non-negligible overlap with left singular vectors having singular value $>\varsigma+\delta$. In fact, due to the nature of singular value transformation, for a left singular vector $\ket{w_j}$ with singular value $\varsigma_j\in[\varsigma-\delta,\varsigma+\delta]$ the procedure still performs a meaningful operation: it maps $\ket{w_j}\rightarrow f(\varsigma_j)\ket{v_j}$, such that $f(\varsigma_j)\in[-1,1]$. 
	It is plausible to believe that the transition behavior on $[\varsigma-\delta,\varsigma+\delta]$ would in practice not significantly degrade the performance of typical machine learning applications, therefore the promise of not having singular values in $[\varsigma-\delta,\varsigma+\delta]$ is probably not crucial.
	Also note that an essentially quadratic improvement to the runtime of the above procedure can be achieved using variable-time amplitude amplification techniques~\cite{ambainis2010VTAA,childs2015QLinSysExpPrec,chakraborty2018BlockMatrixPowers}.
	
	Finally, we briefly discuss a recently developed quantum machine learning algorithm which is significantly more complex then the previous algorithm, but can still be easily fitted to our framework. Kerenidis and Luongo recently proposed a quantum algorithm for slow feature analysis~\cite{kerenidis2018SlowFeatureAnalysis}. The main ingredient of their algorithm is to apply a threshold projection on some input state, i.e., to project onto the subspace spanned by the singular vectors of a matrix with singular values smaller than a certain threshold. Their algorithm is based on singular value estimation, whereas our Theorem~\ref{thm:singularVecProjs} approaches the same problem in a more direct way, by transforming the singular values according to a threshold function. 
	
	In the first step of the quantum algorithm of Kerenidis and Luongo, the task is to implement $Y:=V\Sigma^{-1} V^\dagger$ for a given input matrix $X=W\Sigma V^\dagger$. In our framework this can be performed analogously to Theorem~\ref{thm:pseudoinverse} using singular value transformation; the only difference is that one needs to use an even polynomial approximation of $\frac{1}{x}$, for example given by Corollary~\ref{cor:negatiwePower}. In the next step, one needs to implement singular value threshold projection using the matrix $\dot{X}Y$ for a given ``derivative'' matrix $\dot{X}$. Taking the product\footnote{In case we would have a subnormalized version of $\dot{X}$, in order to get maximal efficiency, it usually worth amplifying $\dot{X}$ using Theorem~\ref{thm:singularValAmp} before taking the product $\dot{X}Y$.} of the two matrices can be implemented using Lemma~\ref{lemma:disjointAncillaProduct}, after which we can use our Theorem~\ref{thm:singularVecProjs} to implement singular value threshold projection. 
	
\section{Matrix Arithmetics using blocks of unitaries}
\label{sec:matArith}

In this section we describe a generic toolbox for implementing matrix calculations on a quantum computer in an operational way, representing the vectors as quantum states. The matrix arithmetics methodology we propose carries out all calculations in an operational way, such that the matrices are represented by blocks of unitary operators of the quantum system, thereby can in principle result in exponential speed-ups in terms of the dimension of the matrices. The methodology we describe is a distilled version of the results of a series of works on quantum algorithms \cite{harrow2009QLinSysSolver,berry2014HamSimTaylor,childs2015QLinSysExpPrec,low2016HamSimQubitization,apeldoorn2017QSDPSolvers,chakraborty2018BlockMatrixPowers}.

We present the results in an intuitively structured way. First we define how to represent arbitrary matrices as blocks of unitaries, and show how to efficiently encode various matrices this way. Then we show how to implement addition and subtraction of these matrices, and finally show how to efficiently obtain products of block-encoded matrices. In order to make the results maximally reusable we also give bounds on the propagation of errors arising from inaccurate encodings.

\subsection{Block-encoding}
We introduce a definition of block-encoding which we are going to work with in the rest of the paper.
The main idea is to represents a subnormalized matrix as the upper-left block of a unitary.
$$ U=\left[\begin{array}{cc} A/\alpha & . \\ . & .\end{array}\right] \kern10mm\Longrightarrow\kern10mm A =\alpha (\bra{0}\otimes I)U(\ket{0}\otimes I)$$
\begin{definition}[Block-encoding]\label{def:standardForm}
	Suppose that $A$ is an $s$-qubit operator, $\alpha,\eps\in\R_+$ and $a\in \N$, then we say that the $(s+a)$-qubit unitary $U$ is an $(\alpha,a,\eps)$-block-encoding of $A$, if 
	$$ \nrm{A - \alpha(\bra{0}^{\otimes a}\otimes I)U(\ket{0}^{\otimes a}\otimes I)}\leq \eps. $$
\end{definition}
Note that since $\nrm{U}=1$ we necessarily have $\nrm{A}\leq \alpha+\eps$. Also note that using the above definition it seems that we can only represent square matrices of size $2^s\times 2^s$. However, this is not really a restriction. Suppose that $A\in\C^{n\times m}$, where $n,m\leq 2^s$. Then we can define an embedding matrix denoted by $A_e\in\C^{2^s\times2^s}$ such that the top-left block of $A_e$ is $A$ and all other elements are $0$. This embedding is a faithful representation of the matrices. Suppose that $A,B\in\C^{n\times m}$ are matrices, then $A_e+B_e=(A+B)_e$. Moreover, suppose $C\in\C^{m\times k}$ for some $k\leq 2^s$, then $A_e \cdot C_e=(A\cdot C)_e$.

The above defined block-encoding is a special case of the projected-encoding of Definition~\ref{def:singDec}, therefore we can later apply our singular value transformation results for block-encoded matrices. In this manner the advantage of block-encoding is that the C$_\Pi$NOT gate which is required in order to implement the gates of Figure~\ref{fig:qubitization} is just a Toffoli gate on $a+1$ qubits, which can be implemented by $\bigO{a+1}$ two-qubit gates and using a single additional ancilla qubit~\cite{he2017ToffoliLinearGateComplexity}.

\subsection{Constructing block-encodings}	

\begin{definition}[Trivial block-encoding]
	A unitary matrix is a $(1,0,0)$-block-encoding of itself. 
\end{definition}	
If we $\eps$-approximately implement a unitary $U$ using $a$ ancilla qubits via a unitary $\tilde{U}$ acting jointly on the system and the ancilla qubits, then $\tilde{U}$ is an $(1,a,\eps)$-block-encoding of $U$. This is also a rather trivial encoding. Note that we make a slight distinction between ancilla qubits that are exactly returned to their original state after the computation and the ones that might pick up some error. The latter qubits we will treat as part of the encoding, and the former qubits we usually treat separately as purely ancillary qubits.

Now we present some non-trivial ways for constructing block-encodings, which will serve as a toolbox for efficiently inputting and representing matrices for arithmetic computations on a quantum computer. We will denote by $I_w$ a $w$-qubit identity operator, and let $\mathrm{SWAP}_w$ denote the swap operation of two $w$-qubit register. We denote by $\mathrm{CNOT}$ the controlled not gate that targets the first qubit. When clear from the context we use simply notation $\ket{0}$ to denote $\ket{0}^{\otimes w}$.

First we show following Low and Chuang~\cite{low2016HamSimQubitization}, how to create a block-encoding of a purified density operator. This technique can be used in combination with the optimal block-Hamiltonian simulation result Theorem~\ref{thm:blockHamSim}, in order to get much better simulation performance, compared to density matrix exponentiation techniques~\cite{lloyd2013QPrincipalCompAnal,kimmel2016hamiltonian} which does not use purification. This result can be generalized for subnormalized density operators too, for more details see~\cite{apeldoorn2018ImprovedQSDPSolving}.
\begin{lemma}[Block-encoding of density operators]
	Suppose that $\rho$ is an $s$-qubit density operator and $G$ is an $(a+s)$-qubit unitary that on the $\ket{0}\ket{0}$ input state prepares a purification $\ket{0}\ket{0}\rightarrow \ket{\rho}$, s.t. $\mathrm{Tr}_{a}{\ketbra{\rho}{\rho}}=\rho$. Then $(G^\dagger\otimes I_s)(I_a\otimes \mathrm{SWAP}_s)(G\otimes I_s)$ is a $(1,a+s,0)$-block-encoding of $\rho$.
\end{lemma}
\begin{proof}
	Let $r$ be the Schmidt-rank of $\rho$, let $\{\ket{\psi_k}\colon k\in [2^s]\}$ be an orthonormal basis, let $\{\ket{\phi_k}\colon k\in [r]\}$ be an orthonormal system and let $p\in[0,1]^{2^s}$ be such that $\ket{\rho}=\sum_{k=1}^{r}\sqrt{p_k}\ket{\phi_k}\ket{\psi_k}$ and $p_\ell=0$ for all $\ell \in [2^s]\setminus [r]$. Then for all $i,j\in [2^s]$ we have that
	\begin{align*}
		&\bra{0}^{\otimes a+s}\bra{\psi_i}(G^\dagger\otimes I_s)(I_a\otimes \mathrm{SWAP}_s)(G^\dagger\otimes I_s)\ket{0}^{\otimes a+s}\ket{\psi_j}=\\
		&\kern40mm=\bra{\rho}\bra{\psi_i}(I_a\otimes \mathrm{SWAP}_s)\ket{\rho}\ket{\psi_j}\\
		&\kern40mm=\left(\sum_{k=1}^{r}\sqrt{p_k}\bra{\phi_k}\bra{\psi_k}\right)\bra{\psi_i}(I_a\otimes \mathrm{SWAP}_s)\left(\sum_{\ell=1}^{r}\sqrt{p_\ell}\ket{\phi_\ell}\ket{\psi_\ell}\right)\ket{\psi_j}\\
		&\kern40mm=\left(\sum_{k=1}^{r}\sqrt{p_k}\bra{\phi_k}\bra{\psi_k}\bra{\psi_i}\right)\left(\sum_{\ell=1}^{r}\sqrt{p_\ell}\ket{\phi_\ell}\ket{\psi_j}\ket{\psi_\ell}\right)\\	
		&\kern40mm=\sqrt{p_j p_i}\delta_{ij}\\	
		&\kern40mm=\bra{\psi_i}\rho\ket{\psi_j}.\\				
	\end{align*}
	\vskip-12mm
\end{proof}

Apeldoorn and Gilyén~\cite{apeldoorn2018ImprovedQSDPSolving} recently also showed that an implementation scheme for a POVM operator can also be easily transformed to block-encoding of the POVM operators. By an implementation scheme we mean a quantum circuit $U$ that given input $\rho$ and $a$ ancilla qubits, it sets a flag qubit to $0$ with probability $\tr{\rho M}$.
\begin{lemma}[Block-encoding of POVM operators]
	Suppose that $U$ is an $a+s$ qubit unitary, which implements a POVM operator $M$ with $\eps$-precision such that for all $s$-qubit density operator $\rho$
	\begin{equation}\label{eq:POVMApx}
	\left|\tr{\rho M}-\tr{U\left(\ketbra{0}{0}^{\otimes a}\otimes \rho\right)U^\dagger\left(\ketbra{0}{0}^{\otimes 1}\otimes I_{a+s-1}\right)}\right|\leq\eps. 	
	\end{equation}
	Then $(I_1\otimes U^\dagger)(\mathrm{CNOT}\otimes I_{a+s-1})(I_1\otimes U)$ is a $(1,1+a,\eps)$-block-encoding of the matrix $M$.
\end{lemma}
\begin{proof}
	First observe that by the cyclicity of trace we have that
	\begin{align*}
		\tr{U\left(\ketbra{0}{0}^{\otimes a}\otimes \rho\right)U^\dagger\left(\ketbra{0}{0}\otimes I_{a+s-1}\right)}
		&=\tr{U\left(\ket{0}^{\otimes a}\otimes I\right) \rho \left(\bra{0}^{\otimes a}\otimes I\right) U^\dagger\left(\ketbra{0}{0}\otimes I_{a+s-1}\right)}\\
		&=\tr{\rho \left(\bra{0}^{\otimes a}\otimes I\right) U^\dagger\left(\ketbra{0}{0}\otimes I_{a+s-1}\right)U\left(\ket{0}^{\otimes a}\otimes I\right) }.		
	\end{align*}
	Together with \eqref{eq:POVMApx} this implies that for all $\rho$ density operator
	\begin{equation*}
		\left|\tr{\rho\left(M- \left(\bra{0}^{\otimes a}\otimes I\right) U^\dagger\left(\ketbra{0}{0}\otimes I_{a+s-1}\right)U\left(\ket{0}^{\otimes a}\otimes I\right)\right) }\right|\leq \eps,
	\end{equation*}
	which is equivalent to saying that $\nrm{M- \left(\bra{0}^{\otimes a}\otimes I\right) U^\dagger\left(\ketbra{0}{0}\otimes I_{a+s-1}\right)U\left(\ket{0}^{\otimes a}\otimes I\right)}\leq \eps$. We can conclude by observing that 
	\begin{align*}
	&\left(\bra{0}^{\otimes a}\otimes I\right) U^\dagger\left(\ketbra{0}{0}\otimes I_{a+s-1}\right)U\left(\ket{0}^{\otimes a}\otimes I\right)=\\
	&\kern48mm =\left(\bra{0}^{\otimes 1+a}\otimes I\right) \left(I_1\otimes U^\dagger\right)\left(\mathrm{CNOT}\otimes I_{a+s-1}\right)\Big(I_1\otimes U\Big)\left(\ket{0}^{\otimes 1+a}\otimes I\right).
	\end{align*}	
\end{proof}

Now we turn to a more traditional way of constructing block-encodings via state preparation. This is a common technique for example to implement quantum walks. Note that we introduce the notation $[n]-1$ to denote the set $\{0,1,\ldots, n-1\}$.

\begin{lemma}[Block-encoding of Gram matrices by state preparation unitaries]\label{lemma:GramBolck}
	Let $U_L$ and $U_R$ be ``state preparation'' unitaries acting on $a+s$ qubits preparing the vectors $\{\ket{\psi_i}\colon i\in[2^s]-1\}$, $\{\ket{\phi_j}\colon j\in[2^s]-1\}$, s.t.
	\begin{align*}
		&U_L\colon \ket{0}\ket{i}\rightarrow \ket{\psi_i}\\ 
		&U_R\colon \ket{0}\ket{j}\rightarrow \ket{\phi_j}. 		
	\end{align*}
	Then $U=U_L^\dagger U_R$ is an $(1,a,0)$-block-encoding of the Gram matrix $A$ such that $A_{ij}=\braket{\psi_i}{\phi_j}$.
\end{lemma}

Based on the above idea one can efficiently implement block-encodings of sparse-access matrices.

\begin{lemma}[Block-encoding of sparse-access matrices]\label{lemma:sparseMatrixBlock}
	Let $A\in\C^{2^w\times 2^w}$ be a matrix that is $s_r$-row-sparse and $s_c$-column-sparse, and each element of $A$ has absolute value at most $1$. Suppose that we have access to the following sparse-access oracles acting on two $(w+1)$ qubit registers
	\begin{align*}
		&\mathrm{O}_r\colon \ket{i}\ket{k} \rightarrow \ket{i}\ket{r_{ik}}& &\kern-30mm\forall i\in[2^w]-1, k\in [s_r],\text{ and}\\
		&\mathrm{O}_c\colon \ket{\ell}\ket{j} \rightarrow \ket{c_{\ell j}}\ket{j}& &\kern-30mm\forall \ell\in [s_c], j\in[2^w]-1,\text{ where}		
	\end{align*}
	$r_{ij}$ is the index for the $j$-th non-zero entry of the $i$-th row of $A$, or if there are less than $i$ non-zero entries, then it is $j+2^w$, and similarly $c_{ij}$ is the index for the $i$-th non-zero entry of the $j$-th column of $A$, or if there are less than $j$ non-zero entries, then it is $i+2^w$.
	Additionally assume that we have access to an oracle $\mathrm{O}_A$ that returns the entries of $A$ in a binary description
	\begin{align*}
		&\mathrm{O}_A\colon \ket{i}\ket{j}\ket{0}^{\!\otimes b} \rightarrow \ket{i}\ket{j}\ket{a_{ij}}& &\kern-30mm\forall i,j\in[2^{w}]-1,\text{ where}		
	\end{align*}	
	$a_{ij}$ is a $b$-bit binary description\footnote{For simplicity we assume here that the binary representation is exact.} of the $ij$-matrix element of $A$. 
	Then we can implement a $(\sqrt{s_r s_c},w+3,\eps)$-block-encoding of $A$ with a single use of $\mathrm{O}_r, \mathrm{O}_c$, two uses of $\mathrm{O}_A$ and additionally using $\bigO{w+\log^{2.5}(\frac{s_r s_c}{\eps})}$ one and two qubit gates while using $\bigO{b,\log^{2.5}(\frac{s_r s_c}{\eps})}$ ancilla qubits.
\end{lemma}
\begin{proof}
	We proceed by constructing state preparation unitaries in the spirit of Lemma~\ref{lemma:GramBolck}. We will work with $3$-registers the first of which is a single qubit register, and the other two registers have $(w+1)$ qubits. Let $D_{s}$ be a $(w+1)$-qubit unitary that implements the map $\ket{0}\rightarrow \sum_{k=1}^{s}\frac{\ket{k}}{\sqrt{s}}$, it is known that this operator $D_{s}$ can be implemented with $\bigO{w}$ quantum gates using $\bigO{1}$ ancilla qubits. Then we define the $2(w+1)$ qubit unitary
	$V_L:= \mathrm{O}_r(I_{w+2}\otimes D_{s_r}) \mathrm{SWAP}_{w+1}$ such that
	\begin{align*}
		V_L\colon \ket{0}^{w+2}\ket{i}\rightarrow \sum_{k=1}^{s_r}\frac{\ket{i}\ket{r_{ik}}}{\sqrt{s_r}} \qquad\forall i\in[2^w]-1.
	\end{align*}
	We implement the operator $V_R:= \mathrm{O}_c(D_{s_c}\otimes I_{w+1})$ in a similar way acting as
	\begin{align*}
		V_R\colon \ket{0}^{w+2}\ket{j}\rightarrow \sum_{\ell=1}^{s_c}\frac{\ket{c_{\ell j}}\ket{j}}{\sqrt{s_c}} \qquad\forall j\in[2^w]-1.
	\end{align*}
	It is easy to see that the above unitaries are such that 
	\begin{align*}
		\bra{0}^{w+2}\bra{i}V_L^\dagger V_R\ket{0}^{w+2}\ket{j}=\frac{1}{\sqrt{s_r s_c}}\text{ if } a_{ij}\neq 0 \text{ and } 0 \text{ otherwise.}
	\end{align*}
	Now we define $U_L:=I_1\otimes V_L$ and define $U_R$ as performing the unitary $I_1\otimes V_R$ followed by some extra computation. After performing $V_R$ we get a superposition of index pairs $\ket{i}\ket{j}$. Given an index pair $\ket{i}\ket{j}$ we query the matrix element $\ket{a_{ij}}$ using the oracle $\mathrm{O}_A$. Then we do some elementary computations in order to implement a single qubit gate $\ket{0}\rightarrow a_{ij}\ket{0}+\sqrt{1-|a_{ij}|^2}\ket{1}$ on the first qubit, with precision $\bigO{\poly{\frac{\eps}{s_r s_c}}}$. This can be executed with the stated complexity, for more details see, e.g., the work of Berry et al.~\cite{berry2015HamSimNearlyOpt}. Finally we also need to uncompute everything which requires one more use of $\mathrm{O}_A$. This way we get a good approximation of 
	\begin{align*}
		U_R\colon \ket{0}^{w+3}\ket{j}\rightarrow \sum_{\ell=1}^{s_c}\frac{\left(a_{c_{\ell j}j}\ket{0}+\sqrt{1-|a_{c_{\ell j}j}|^2}\ket{1}\right)\ket{c_{\ell j}}\ket{j}}{\sqrt{s_c}} \qquad\forall j\in[2^w]-1.
	\end{align*}
	\vskip-5mm	
\end{proof}

Note that in the above method the matrix gets subnormalized by a factor of $\frac{1}{\sqrt{s_r s_c}}$. If we would know that for example $\nrm{A}\leq \frac{1}{2}$, then we could amplify the block-encoding in order to remove this unwanted subnormalization using singular value amplification Theorem~\ref{thm:singularValAmp} using the block-encoding roughly $\sqrt{s_r s_c}$ times. However, under some circumstances one can defeat the subnormalization more efficiently by doing an amplification at the level of the state preparation unitaries. The idea comes from Low and Chuang~\cite{low2017HamSimUnifAmp}, who called this technique ``Uniform spectral gap amplification''. We generalize their results combining with ideas of Kerenidis and Prakash~\cite{kerenidis2017QGradDesc} and Chakraborty et al.~\cite{chakraborty2018BlockMatrixPowers}, who used similar ideas but assumed QROM-access to matrices rather than sparse-access.

\begin{lemma}[Preamplified block-encoding of sparse-access matrices]
	Let $A\in\C^{2^w\times 2^w}$ be a matrix that is $s_r$-row-sparse and $s_c$-column-sparse, and is given using the input oracles defined in Lemma~\ref{lemma:sparseMatrixBlock}. Let $a_{i.}$ denote the $i$-th row of $A$ and similarly $a_{.j}$ the $j$-th column. Let $q\in [0,2]$ and suppose that $n_r\in[1,s_r]$ is an upper bound on $\nrm{a_{i.}}_q^q$ and $n_c\in[1,s_c]$ is an upper bound on $\nrm{a_{.j}}_{2-q}^{2-q}$.
	
	Let $m=\max[\frac{s_r}{n_r},\frac{s_c}{n_c}]$. Then we can implement a $\big(\sqrt{\frac{1}{2n_r n_c}},w+6,\eps\big)$-block-encoding of $A$ with $\bigO{\sqrt{\frac{s_r}{n_r}}\log(\frac{s_r s_c}{\eps})}$ uses of $\mathrm{O}_r$, $\bigO{\sqrt{\frac{s_c}{n_c}}\log(\frac{s_r s_c}{\eps})}$ uses of $\mathrm{O}_c$, $\bigO{\sqrt{m}\log(\frac{s_r s_c}{\eps})}$ uses of $\mathrm{O}_A$, and additionally using $\bigO{\sqrt{m}\left(w\log(\frac{s_r s_c}{\eps})+\log^{3.5}(\frac{s_r s_c}{\eps})\right)}$ one and two qubit gates while using $\bigO{b,\log^{2.5}(\frac{s_r s_c}{\eps})}$ ancilla qubits.
\end{lemma}
\begin{proof}
	The idea is very similar to the proof of Lemma~\ref{lemma:sparseMatrixBlock}, we implement the unitaries $V_L,V_R$ the same way.
	However, we define $U_R,U_L$ slightly differently. Using a similar method than in Lemma~\ref{lemma:sparseMatrixBlock}, we implement $\bigO{\poly{\frac{\eps}{s_r s_c}}}$-approximations of the maps
	\begin{align*}
		&U_L\colon \ket{0}^{w+4}\ket{i}\rightarrow \sum_{k=1}^{s_r}\frac{\left(|a_{ir_{ik}}|^{\frac{q}2}\ket{0}+\sqrt{1-|a_{ir_{ik}}|^{q}}\ket{1}\right)\ket{0}\ket{i}\ket{r_{ik}}}{\sqrt{s_r}} &\forall i\in[2^w]-1,\\
		&U_R\colon \ket{0}^{w+4}\ket{j}\rightarrow \sum_{\ell=1}^{s_c}\frac{\frac{a_{c_{\ell j}j}}{|a_{c_{\ell j}j}|}\ket{0}\left(|a_{c_{\ell j}j}|^{1-\frac{q}2}\ket{0}+\sqrt{1-|a_{c_{\ell j}j}|^{2-q}}\ket{1}\right)\ket{c_{\ell j}}\ket{j}}{\sqrt{s_c}} &\forall j\in[2^w]-1.		
	\end{align*}
	It is easy to see that the above unitaries are such that 
	\begin{align*}
	\bra{0}^{w+4}\bra{i}U_L^\dagger U_R\ket{0}^{w+4}\ket{j}=\frac{a_{ij}}{\sqrt{s_r s_c}}
	\qquad\forall i,j\in[2^w]-1.		
	\end{align*}
	We can see that for all $i\in [2^w]-1$ the modified row vector $\sum_{k=1}^{s_r}\frac{|a_{ir_{ik}}|^{\frac{q}2}\ket{0}\ket{0}\ket{i}\ket{r_{ik}}}{\sqrt{s_r}}$ has squared norm at most $\frac{n_r}{s_r}$, and a similar $\frac{n_c}{s_c}$ upper bound holds for the squared norm of the modified column vector.
	Also observe that 
	\begin{align*}
		(\ketbra{0}{0}\otimes I_{2w+3})U_L(\ketbra{0}{0}^{w+4}\otimes I_{w})=\sum_{j=0}^{2^w-1}\left(\sum_{k=1}^{s_r}\frac{|a_{ir_{ik}}|^{\frac{q}2}\ket{0}\ket{0}\ket{i}\ket{r_{ik}}}{\sqrt{s_r}}\right)\bra{0}^{w+4}\bra{j},
	\end{align*}
	which is a singular value decomposition with the singular values being the modified row norms. Therefore we can apply singular value amplification Theorem~\ref{thm:singularValAmp} to with amplification $\gamma_r=\sqrt{\frac{s_r}{\sqrt{2}n_r}}$ and precision $\bigO{\poly{\frac{\eps}{s_r s_c}}}$ resulting in an $\bigO{\poly{\frac{\eps}{s_r s_c}}}$ approximation of $\tilde{U}_L$ such that 
	\begin{align*}
	(\bra{+}\otimes \ketbra{0}{0}\otimes I_{2w+3})\tilde{U}_L(\ket{+}\otimes\ketbra{0}{0}^{w+4}\otimes I_{w})=\gamma_r\sum_{j=0}^{2^w-1}\left(\sum_{k=1}^{s_r}\frac{|a_{ir_{ik}}|^{\frac{q}2}\ket{0}\ket{0}\ket{i}\ket{r_{ik}}}{\sqrt{s_r}}\right)\bra{0}^{w+4}\bra{j}.
	\end{align*}
	Similarly we apply singular value amplification with amplification $\gamma_c=\sqrt{\frac{s_c}{\sqrt{2}n_c}}$ and precision $\bigO{\poly{\frac{\eps}{s_r s_c}}}$ resulting in a $\bigO{\poly{\frac{\eps}{s_r s_c}}}$ approximation of $\tilde{U}_R$  such that 
	\begin{align*}
	\bra{++}\bra{0}^{w+4}\bra{i}\tilde{U}_L^\dagger \tilde{U}_R\ket{++}\ket{0}^{w+4}\ket{j}=\gamma_r\gamma_c\frac{a_{ij}}{\sqrt{s_r s_c}}=\frac{a_{ij}}{\sqrt{2n_r n_c}}
	\qquad\forall i,j\in[2^w]-1.		
	\end{align*}
	Finally adding $4$ Hadamard gates we can change the $\ket{+}$ states above to $\ket{0}$ states, resulting in the $\big(\sqrt{\frac{1}{2n_r n_c}},w+6,\eps\big)$-block-encoding of $A$.
	The complexity statement follows similarly as in the proof of Lemma~\ref{lemma:sparseMatrixBlock}, with the extra observation that the singular value amplifications of $U_L$ and $U_R$ can be performed using degree $\bigO{\gamma_r\log(\frac{s_r s_c}{\eps})}$ and $\bigO{\gamma_c\log(\frac{s_r s_c}{\eps})}$ singular value transformations respectively.
\end{proof}

Finally, for completeness we invoke the results of Kerenidis and Prakash~\cite{kerenidis2017QGradDesc} and Chakraborty et al.~\cite{chakraborty2018BlockMatrixPowers}, who showed how to efficiently implement block-encodings of matrices that are stored in a clever quantum data structures in a quantum accessible RAM.

For $q\in [0,2]$ let us define $\mu_q(A)=\sqrt{n_{q}(A)n_{(2-q)}(A^T)}$, where $n_{q}(A):=\max_i\nrm{a_{i.}}_q^q$ is the $q$-th power of the maximum $q$-norm of the rows of $A$. Let $A^{(q)}$ denote the matrix of the same dimensions as $A$, with\footnote{For complex values we define these non-integer powers using the principal value of the complex logarithm function.} $A^{(q)}_{ij}=\sqrt{a_{ij}^q}$. The following was proven in \cite{kerenidis2017QGradDesc}, although not in the language of block-encodings, and was stated in this form by Chakraborty et al.~\cite{chakraborty2018BlockMatrixPowers}.

\begin{lemma}[Block-encodings of matrices stored in quantum data structures]\label{lem:kp}
	Let $A\in\C^{2^w\times 2^w}$.
	\begin{enumerate}
		\item Fix $q\in [0,2]$. If $A^{(q)}$ and $(A^{(2-q)})^\dagger$ are both stored in quantum accessible data structures\footnote{\label{foot:precData}Here we assume that the data-structure stores the matrices with sufficient precision, cf.~\cite{chakraborty2018BlockMatrixPowers}.}, then there exist unitaries $U_R$ and $U_L$ that can be implemented in time $\bigO{\poly{w\log(1/\eps)}}$ such that 
		$U_R^\dagger U_L$ is a $(\mu_q(A),w+2,\eps)$-block-encoding of $A$. 
		\item On the other hand, if $A$ is stored in a quantum accessible data structure\textsuperscript{\emph{\ref{foot:precData}}}, then there exist unitaries $U_R$ and $U_L$ that can be implemented in time $\bigO{\poly{w\log(1/\eps)}}$ such that $U_R^\dagger U_L$ is an $(\nrm{A}_F,w+2,\eps)$-block-encoding of $A$.
	\end{enumerate}
\end{lemma}

\subsection{Addition and subtraction: Linear combination of block-encoded matrices}	
\label{subsec:add}

We use a simple but powerful method for implementing linear combinations of unitary operators on a quantum computer.
This technique was introduced by Berry et al.~\cite{berry2014HamSimTaylor} for exponentially improving the precision of Hamiltonian simulation. Later it was adapted by Childs et al.~\cite{childs2015QLinSysExpPrec} for exponentially improving the precision of quantum linear equation solving. Here we present this method from the perspective of block-encoded matrices.

First we define state preparation unitaries in order to conveniently state our the result in the following lemma.

\begin{definition}[State preparation pair]
	Let $y\in \C^m$ and $\nrm{y}_1\leq \beta$, the pair of unitaries $(P_L,P_R)$ is called a $(\beta, b, \eps)$-state-preparation-pair if $P_L\ket{0}^{\otimes b}= \sum_{j=0}^{2^b-1} c_j \ket{j}$ and  $P_R\ket{0}^{\otimes b}= \sum_{j=1}^{2^b-1} d_j \ket{j}$ such that $\sum_{j=0}^{m-1}|\beta (c^*_j d_j) -y_j | \leq \eps$ and for all $j\in m,\ldots, 2^b-1$ we have $c^*_j d_j =0$.
\end{definition}

Now we show how to implement a block-encoding of a linear combination of block-encoded operators.

\begin{lemma}[Linear combination of block-encoded matrices]\label{lem:linCombBlocks}
	Let $A=\sum_{j=1}^{m}y_j A_j$ be an $s$-qubit operator and $\eps\in \R_+$. Suppose that $(P_L,P_R)$ is a $(\beta, b, \eps_1)$-state-preparation-pair for $y$, $W=\sum_{j=0}^{m-1} \ketbra{j}{j}\otimes U_j+((I-\sum_{j=0}^{m-1}\ketbra{j}{j})\otimes I_a\otimes I_s)$ is an $s+a+b$ qubit unitary such that for all $j\in 0,\ldots, m$ we have that $U_j$ is an $(\alpha,a,\eps_2)$-block-encoding of $A_j$. Then we can implement a $(\alpha\beta,a+b,\alpha\eps_1+\alpha\beta\eps_2)$-block-encoding of $A$, with a single use of $W$, $P_R$ and $P_L^\dagger$.
\end{lemma}
\begin{proof}
	Observe that $\widetilde{W}=(P_L^\dagger\otimes I_a\otimes I_s) W  (P_R\otimes I_a\otimes I_s)$ is a	$(\alpha\beta,a+b,\alpha\eps_1+\alpha\beta\eps_2)$-block-encoding of $A$:
	\begin{align*}
	\kern-1mm
	\nrm{A-\alpha\beta(\bra{0}^{\otimes b}\otimes\bra{0}^{\otimes a}\otimes I)\widetilde{W}(\ket{0}^{\otimes b}\otimes\ket{0}^{\otimes a}\otimes I)}
	&=\nrm{A -\alpha\sum_{j=0}^{m-1}\beta (c^*_j d_j) (\bra{0}^{\otimes a}\otimes I)U_j(\ket{0}^{\otimes a}\otimes I)}\\
	&\leq\alpha\eps_1 +\nrm{A -\alpha\sum_{j=0}^{m-1}y_j(\bra{0}^{\otimes a}\otimes I)U_j(\ket{0}^{\otimes a}\otimes I)}\kern-1mm\\
	&\leq\alpha\eps_1 +\alpha\sum_{j=0}^{m-1}y_j\nrm{A_j-(\bra{0}^{\otimes a}\otimes I)U_j(\ket{0}^{\otimes a}\otimes I)}\kern-3mm\\		
	&\leq\alpha\eps_1 +\alpha\sum_{j=0}^{m-1}y_j\eps_2\\			
	&\leq\alpha\eps_1 +\alpha\beta\eps_2.
	\end{align*}
	\vskip-6mm
\end{proof} 

\subsection{Multiplication: Product of block-encoded matrices}	
\label{subsec:mult}

In general if we want to take the product of two block encoded matrices we need to treat their ancilla qubits separately. In this case as the following lemma shows the errors simply add up and the block encoding does not introduce any additional errors.
 
\begin{lemma}[Product of block-encoded matrices]\label{lemma:disjointAncillaProduct}
	If $U$ is an $(\alpha,a,\delta)$-block-encoding of an $s$-qubit operator $A$, and $V$ is an $(\beta,b,\eps)$-block-encoding of an $s$-qubit operator $B$ then\footnote{The identity operators act on each others ancilla qubits, which is hard to express properly using simple tensor notation, but the reader should read this tensor product this way.} $(I_b\otimes U)(I_a\otimes V)$ is an $(\alpha\beta,a+b,\alpha\eps+\beta\delta)$-block-encoding of $AB$.
\end{lemma}
\begin{proof}
	\begin{align*}
	&\nrm{AB - \alpha\beta(\bra{0}^{\otimes a+b}\otimes I)(I_b\otimes U)(I_a\otimes V)(\ket{0}^{\otimes a+b}\otimes I)}\\
	=&\Big\lVert AB - 
	\underset{\tilde{A}}{\underbrace{\alpha(\bra{0}^{\otimes a}\otimes I)U(\ket{0}^{\otimes a}\otimes I)}}
	\underset{\tilde{B}}{\underbrace{\beta(\bra{0}^{\otimes b}\otimes I)V(\ket{0}^{\otimes b}\otimes I)}}\Big\rVert\\
	=&\nrm{AB -\tilde{A}B+\tilde{A}B-\tilde{A}\tilde{B}}\\
	=&\nrm{(A-\tilde{A})B+\tilde{A}(B-\tilde{B})}\\
	\leq&\nrm{A-\tilde{A}}\beta+\alpha\nrm{B-\tilde{B}}\\		
	\leq& \alpha\eps+\beta\delta.
	\end{align*}
	\vskip-6mm
\end{proof}

In the special case when the encoded matrices are unitaries and their block-encoding does not use any extra scaling factor, then we might reuse the ancilla qubits, however it introduces an extra error term, which can be bounded by the geometrical mean of the two input error bounds.

\begin{lemma}[Product of two block-encoded unitaries]
	If $U$ is an $(1,a,\delta)$-block-encoding of an $s$-qubit unitary operator $A$, and $V$ is an $(1,a,\eps)$-block-encoding of an $s$-qubit unitary operator $B$ then $UV$ is a $(1,a,\delta+\eps+2\sqrt{\delta\eps})$-block-encoding of the unitary operator $AB$.
\end{lemma}
\begin{proof}
	It is enough to show that for all $s$-qubit pure states $\ket{\phi},\ket{\psi}$ we have that 
	$$ \left|\bra{\phi}AB\ket{\psi}-\bra{\phi}(\bra{0}^{\otimes a}\otimes I)UV(\ket{0}^{\otimes a}\otimes I)\ket{\psi}\right|\leq \delta+\eps +2\sqrt{\delta\eps}.$$
	Observe that 
	\begin{align*}
	&\bra{\phi}(\bra{0}^{\otimes a}\otimes I)UV(\ket{0}^{\otimes a}\otimes I)\ket{\psi}\\
	=&\bra{\phi}(\bra{0}^{\otimes a}\otimes I)U\left((\ketbra{0}{0}^{\otimes a}\otimes I)+\left(\left(I-\ketbra{0}{0}^{\otimes a}\right)\otimes I\right)\right)V(\ket{0}^{\otimes a}\otimes I)\ket{\psi}\\
	=&\bra{\phi}(\bra{0}^{\otimes a}\otimes I)U(\ket{0}^{\otimes a}\otimes I)(\bra{0}^{\otimes a}\otimes I)V(\ket{0}^{\otimes a}\otimes I)\ket{\psi}\\
	&+\bra{\phi}(\bra{0}^{\otimes a}\otimes I)U\left(\left(I-\ketbra{0}{0}^{\otimes a}\right)\otimes I\right)V(\ket{0}^{\otimes a}\otimes I)\ket{\psi}\\
	\end{align*}
	Now we can see that similarly to the proof of Lemma~\ref{lemma:disjointAncillaProduct} we have
	\begin{align*}
	&\left|\bra{\phi}AB\ket{\psi}-\bra{\phi}(\bra{0}^{\otimes a}\otimes I)U(\ket{0}^{\otimes a}\otimes I)(\bra{0}^{\otimes a}\otimes I)V(\ket{0}^{\otimes a}\otimes I)\ket{\psi}\right|\\
	=&\left|\bra{\phi}\left(AB-(\bra{0}^{\otimes a}\otimes I)U(\ket{0}^{\otimes a}\otimes I)(\bra{0}^{\otimes a}\otimes I)V(\ket{0}^{\otimes a}\otimes I)\right)\ket{\psi}\right|\\
	\leq&\nrm{AB-(\bra{0}^{\otimes a}\otimes I)U(\ket{0}^{\otimes a}\otimes I)(\bra{0}^{\otimes a}\otimes I)V(\ket{0}^{\otimes a}\otimes I)}\\
	\leq& \delta+\eps.
	\end{align*}	
	Finally note that
	\begin{align*}
	&\left|\bra{\phi}(\bra{0}^{\otimes a}\otimes I)U\left(\left(I-\ketbra{0}{0}^{\otimes a}\right)\otimes I\right)V(\ket{0}^{\otimes a}\otimes I)\ket{\psi}\right|\\
	=&\left|\bra{\phi}(\bra{0}^{\otimes a}\otimes I)U\left(\left(I-\ketbra{0}{0}^{\otimes a}\right)\otimes I\right)^2V(\ket{0}^{\otimes a}\otimes I)\ket{\psi}\right|\\
	\leq&\nrm{\left(\left(I-\ketbra{0}{0}^{\otimes a}\right)\otimes I\right)U(\ket{0}^{\otimes a}\otimes I)\ket{\phi}}
	\cdot \nrm{\left(\left(I-\ketbra{0}{0}^{\otimes a}\right)\otimes I\right)V(\ket{0}^{\otimes a}\otimes I)\ket{\psi}}\\
	=&\sqrt{1-\nrm{\left(\ketbra{0}{0}^{\otimes a}\otimes I\right)U(\ket{0}^{\otimes a}\otimes I)\ket{\phi}}^2}
	\cdot \sqrt{1-\nrm{\left(\ketbra{0}{0}^{\otimes a}\otimes I\right)V(\ket{0}^{\otimes a}\otimes I)\ket{\psi}}^2}\\	
	\leq&\sqrt{1-(1-\delta)^2}
	\cdot \sqrt{1-(1-\eps)^2}\\			
	\leq& 2\sqrt{\delta\eps}.
	\end{align*}
	\vskip-6mm
\end{proof}

The following corollary suggest that if we multiply together multiple block-encoded unitaries, the error may grow super-linearly, but it increases at most quadratically with the number of factors in the product.

\begin{corollary}[Product of multiple block-encoded unitaries]\label{cor:blockProductPrecision}
	Suppose that $U_j$ is an $(1,a,\eps)$-block-encoding of an $s$-qubit unitary operator $W_j$ for all $j\in [K]$.
	Then $\prod_{j=1}^K U_j$ is an $(1,a,4K^2\eps)$-block-encoding of $\prod_{j=1}^K W_j$.
\end{corollary}
\begin{proof}
	First observe that for the product of two matrices we get the precision bound $4\eps$ by the above lemma. 
	If $K=2^k$ for some $k\in \N$. Then we can apply the above observation in a recursive fashion in a binary tree structure, to get the upper bound $4^k \eps$ on the precision, and observe that $4^k=K^2$.
	
	If $2^{k-1}\leq K< 2^{k}$ we can just add identity operators so that we have $2^k$ matrices to multiply, which gives the precision bound $4^k\eps\leq 4^{1+\log_2 K}\eps= 4K^2\eps$. 
\end{proof}

\section{Implementing smooth functions of Hermitian matrices}
\label{sec:smooth}
	In the previous section we developed an efficient methodology to perform basic matrix arithmetics, such as addition and multiplication. In principle all smooth functions of matrices can be approximated arbitrarily precisely using such basic arithmetic operations. 
	In this section we show how to more efficiently transform Hermitian matrices according to smooth functions using singular value transformation techniques. The key observation is that for a Hermitian matrix $A$ we have that $P^{(SV)}(A)=P(A)$, i.e., singular value transformation and eigenvalue transformation coincide.
	
	The following theorem is our improvement of Corollary~\ref{cor:matchingParity} removing the counter-intuitive parity constraint at the expense of a subnormalization factor $1/2$, which is not a problem in most applications.
	
	\begin{theorem}[Polynomial eigenvalue transformation of arbitrary parity]\label{thm:arbParity}
		Suppose that $U$ is an $(\alpha,a,\eps)$-encoding of a Hermitian matrix $A$. 
		If $\delta\geq 0$ and $ P_{\Re}\in\R[x]$ is a degree-$d$ polynomial satisfying that
		\begin{itemize}
			\item for all $x\in[-1,1]\colon$ $| P_{\Re}(x)|\leq \frac12$.
		\end{itemize}
		Then there is a quantum circuit $\tilde{U}$, which is an $(1,a+2,4d\sqrt{\eps/\alpha}+\delta)$-encoding of $ P_{\Re}(A/\alpha)$, and consists of $d$ applications of $U$ and $U^\dagger$ gates, a single application of controlled-$U$ and $\bigO{(a+1)d}$ other one- and two-qubit gates.
		Moreover we can compute a description of such a circuit with a classical computer in time $\bigO{\poly{d,\log(1/\delta)}}$.
	\end{theorem}
	\begin{proof}
		First note that for a Hermitian matrix $A$ and for any even/odd polynomial $P\in\C[x]$ we have that $P^{(SV)}(A)=P(A)$.
		Now let $P^{(\text{even})}_\Re(x):=P_\Re(x)+P_\Re(-x)\leq 1$, and let $P^{(\text{odd})}_\Re(x):=P_\Re(x)-P_\Re(-x)\leq 1$. Then we can implements both $P^{(\text{even})}_\Re(x)$ and $P^{(\text{odd})}_\Re(x)$ using Corollary~\ref{cor:matchingParity} and we can take an equal $\frac12$ linear combination of them by Lemma~\ref{lem:PolyNormDiff}. Using the notation of Figure~\ref{fig:qubitization} the final circuit is simply $(H\otimes H\otimes I)U_{\Phi^{(c)}}(H\otimes H\otimes I)$, where $H$ denotes the Hadamard gate.
	\end{proof}

	Note that a similar statement can be proven for arbitrary $P\in\C[x]$ that satisfy $|P(x)|\leq \frac14$ for all $x\in[-1,1]$. The only difference is that for implementing the complex part one needs to add a (controlled) phase $e^{i\frac\pi2}$. The $\leq\frac14$ constraint comes form the fact that the implementation is a sum of $4$ different terms (even/odd component of the real/imaginary part).
	
\subsection{Optimal Hamiltonian simulation}
\label{subsec:opt}

	Let us define for $t\in \R_+$ and $\eps\in(0,1)$ the number $r(t,\eps)\geq t$ as the solution to the equation
	\begin{equation}\label{eq:lambertR}
		\eps = \left(\frac{t}{r}\right)^{\!\!r} \colon r\in(t,\infty).
	\end{equation}
	This equation is closely related to the Lambert-$W$ function, and unfortunately we cannot give the solution in terms of elementary functions.
	However, one can see that the function $\left(\frac{t}{r}\right)^{\!r}$ is strictly monotone decreasing for $r\in[t,\infty)$ and in the limit $r\rightarrow \infty$ it tends to $0$. Since for $r=t$ the function value is $1$, therefore the equation \eqref{eq:lambertR} has a unique solution.
	In particular for any $r,R \in[t,\infty)$ such that $\left(\frac{t}{r}\right)^{\!r} \geq \eps \geq \left(\frac{t}{R}\right)^{\!R} $ we have that $r\leq r(t,\eps) \leq R$.
	This is an important expression for this section, since Low and Chuang proved~\cite{low2016HamSimQSignProc,low2016HamSimQubitization} that the complexity of Hamiltonian simulation for time $t$ with precision $\eps$ is $\Theta(r(|t|,\eps))$.
	
	Low and Chuang also claimed~\cite{low2016HamSimQSignProc} that for all $t\geq 1$ one gets $r(t,\eps)=\Theta\left(t+\frac{\log(1/\eps)}{\log(\log(1/\eps))}\right)$, which led to their complexity statement. 
	Note that we found a subtle issue in their calculations making this formula invalid for some range of values of $t$. We show in Lemma~\ref{lemma:lamgertR} how to correct the formula by a slight modification of the $\log(\log(1/\eps))$ term. This is the reason why we need to give more complicated expressions for the complexity of block-Hamiltonian simulation. First we show what is the connection between equation \eqref{eq:lambertR} and the complexity of Hamiltonian simulation by constructing polynomials of degree $\bigO{r(|t|,\eps)}$, which $\eps$-approximate trigonometric functions with $t$-times rescaled argument.
	
	\begin{lemma}[Polynomial approximations of trigonometric functions]\label{lemma:polyTrig}
		Let $t\in\R\setminus\{0\}$, $\eps\in(0,\frac{1}{e})$, and let $R=\left\lfloor r\left(\frac{e|t|}{2},\frac{5}{4}\eps\right)/2 \right\rfloor$, then the following $2R$ and $2R+1$ degree polynomials satisfy
		\begin{align*}
			&\nrm{\cos(tx)-J_{0}(t)+2\sum_{k=1}^{R}(-1)^k J_{2k}(t)T_{2k}(x)}_{[-1,1]}\leq\eps,\text{ and }\\
			&\nrm{\sin(tx)-2\sum_{k=0}^{R}(-1)^k J_{2k+1}(t)T_{2k+1}(x)}_{[-1,1]}\leq\eps,
		\end{align*}
		where $J_m(t)\colon m\in\N$ denote Bessel functions of the first kind.
	\end{lemma}
	\begin{proof}
		We use the Fourier-Chebyshev series of the trigonometric functions given by the Jacobi-Anger expansion~\cite[9.1.44-45]{abramowitz1974HandbookMathFuns}:
		\begin{align*}
			\cos(tx)&=J_{0}(t)+2\sum_{k=1}^{\infty}(-1)^k J_{2k}(t)T_{2k}(x)\\
			\sin(tx)&=2\sum_{k=0}^{\infty}(-1)^k J_{2k+1}(t)T_{2k+1}(x).
		\end{align*}
		The Jacobi-Anger expansion implies that
		\begin{align}
			\nrm{\cos(tx)-J_{0}(t)+2\sum_{k=1}^{R}(-1)^k J_{2k}(t)T_{2k}(x)}_{[-1,1]}
			&=\nrm{2\sum_{k=R+1}^{\infty}(-1)^k J_{2k}(t)T_{2k}(x)}_{[-1,1]}\nonumber\\
			&\leq 2 \sum_{k=R+1}^{\infty}(-1)^k |J_{2k}|\nonumber\\
			&=2 \sum_{\ell=0}^{\infty}(-1)^k |J_{2R+2+2\ell}|,\label{eq:trunCos}
		\end{align}	
		and similarly we can derive that 
		\begin{align}
			\nrm{\sin(tx)-2\sum_{k=0}^{R}(-1)^k J_{2k+1}(t)T_{2k+1}(x)}_{[-1,1]}
			&\leq2 \sum_{\ell=0}^{\infty}(-1)^k |J_{2R+3+2\ell}|.\label{eq:trunSin}
		\end{align}			
		It is known~\cite[9.1.62]{abramowitz1974HandbookMathFuns} that for all $m\in \N_+$ and $t\in \R$ we have
		\begin{align}\label{eq:BesselBound}
			|J_{m}(t)|\leq \frac{1}{m!}\left|\frac{t}{2}\right|^{\!m}.
		\end{align}
		Following~\cite{berry2015HamSimNearlyOpt} we show that it implies that for any positive integer $q\geq |t|-1$ we have that 
		\begin{align}\label{eq:trunr}
			2\sum_{\ell=0}^\infty|J_{(q+2\ell)}(t)|
			\overset{\eqref{eq:BesselBound}}{\leq} 2\sum_{\ell=0}^\infty\frac{|t/2|^{(q+2\ell)}}{(q+2\ell)!} 
			\leq 2\frac{|t/2|^{q}}{q!} \sum_{\ell=0}^\infty\left(\frac{1}{4} \right)^{\!\!\ell}
			= \frac{8}{3}\frac{|t/2|^{q}}{q!}\leq\frac{1.07}{\sqrt{q}} \left(\frac{e|t|}{2 q}\right)^{\!\!q},
		\end{align}				
		where in the last inequality we used that by Stirling's approximation $q!\geq \sqrt{2\pi q} \left(\frac{q}{e}\right)^{\!q}$.
		In the inequalities \eqref{eq:trunCos}-\eqref{eq:trunSin} we can apply the bound of \eqref{eq:trunr} with $q\geq 2(R+1)\geq r\left(\frac{e|t|}{2},\eps\right)$, so we get the upper bound
		\begin{align*}
			\frac{1.07}{\sqrt{q}} \left(\frac{e|t|}{2 q}\right)^{\!\!q}
			\leq \frac{1.07}{\sqrt{2}} \left(\frac{e|t|}{2 q}\right)^{\!\!q}
			\leq \frac{5}{4} \left(\frac{e|t|}{2 q}\right)^{\!\!q}
			\leq \eps.
		\end{align*}	
	\end{proof}	

	Now we are ready to prove the optimal block-Hamiltonian simulation result of Low and Chuang. The optimality is discussed in an earlier work of the same authors~\cite{low2016HamSimQSignProc}. 

	\begin{theorem}\emph{(Optimal block-Hamiltonian simulation \cite{low2016HamSimQubitization})}\label{thm:blockHamSim}
		Let $t\in\R\setminus\{0\}$, $\eps\in(0,1)$ and let $U$ be an $(\alpha,a,0)$-block-encoding of the Hamiltonian $H$. Then we can implement an $\eps$-precise Hamiltonian simulation unitary $V$ which is an $(1,a+2,\eps)$-block-encoding of $e^{itH}$, with $3r\left(\frac{e\alpha|t|}{2},\frac{\eps}{6}\right)$ uses of $U$ or its inverse, $3$ uses of controlled-$U$ or its inverse and with $\bigO{ar\left(\frac{e\alpha|t|}{2},\frac{\eps}{6}\right)}$ two-qubit gates and using $\bigO{1}$ ancilla qubits.
	\end{theorem}
	\begin{proof}
		Use the polynomials of Lemma~\ref{lemma:polyTrig} and combine the even real polynomial $\frac{\eps}{6}$-approximating $\cos(\alpha t)$ with the odd imaginary polynomial $\frac{\eps}{6}$-approximating $i\cdot \sin(\alpha t)$ using the same method as Theorem~\ref{thm:arbParity} in order to get an $(1,a+2,\frac{\eps}{6})$-block-encoding of $e^{itH}/2$. Then use robust oblivious amplitude amplification Corollary~\ref{cor:oblivious} in order get an $(1,a+2,\eps)$-block-encoding of $e^{itH}$.
	\end{proof}

	Now we prove some bounds on $r(t,\eps)$, in order to make the above result more accessible.
	\begin{lemma}[Bounds on $r(t,\eps)$]\label{lemma:lamgertR}
		For $t\in \R_+$ and $\eps\in(0,1)$ 
		$$ r(t,\eps)=\Theta\left(t+\frac{\ln(1/\eps)}{\ln(e+\ln(1/\eps)/t)}\right). $$
		Moreover, for all $q\in \R_+$ we have that
		$$r(t,\eps)<e^q t+\frac{\ln(1/\eps)}{q}.$$
	\end{lemma}
	\begin{proof}
		First consider the case $t\geq \frac{\ln(1/\eps)}{e}$ and set $r:=e t$, then we get that 
		\begin{align}
			\forall t\geq \frac{\ln(1/\eps)}{e} \colon \left(\frac{t}{et}\right)^{\!\!et}=  \left(\frac{1}{e}\right)^{\!\!et}\leq \eps
			\qquad \Longrightarrow\qquad  \forall t\geq \frac{\ln(1/\eps)}{e} \colon r(t,\eps) \leq et.
			\label{eq:larget}
		\end{align}
		Now we turn to the case $t\leq \frac{\ln(1/\eps)}{e}$, and try to find $r=r(t,\eps)$.
		\begin{align}
			\left(\frac{t}{r}\right)^{\!\!r}= \eps  \qquad \Longleftrightarrow \qquad  
			\left(\frac{r}{t}\right)^{\!\!\frac{r}{t}}= \left(\frac{1}{\eps}\right)^{\!\!\frac{1}{t}}
			\qquad \Longleftrightarrow \qquad  			
			\frac{r}{t}\ln\left(\frac{r}{t}\right)= \ln\left(\frac{1}{\eps}\right)\frac{1}{t}.
			\label{eq:lambertForm}
		\end{align}	
		Let us define $x:=\frac{r}{t}\geq 1$ and $c:=\ln\left(\frac{1}{\eps}\right)\frac{1}{t}\geq e$. We will examine the solution of the equation $x\ln(x)=c$ for $c\geq e$. We see that the function $x\ln(x)$ is monotone increasing on $[1,\infty)$, takes value $0$ at $1$ and in the $x\rightarrow \infty$ limit it tends to infinity, therefore the equation $x\ln(x)=c$ has a unique solution for all $c\in\R_+$. 
		Moreover, if $b,B\in [1,\infty)$ are such that $b\ln(b)\leq c \leq B\ln(B)$, then $b\leq x \leq B$. Therefore we can see that $\frac{c}{\ln(c)}\leq x$ since 
		$$\frac{c}{\ln(c)}\ln\left(\frac{c}{\ln(c)}\right)
		=\frac{c}{\ln(c)}\left(\ln(c)-\ln(\ln(c))\right)
		=c\left(1-\frac{\ln(\ln(c))}{\ln(c)}\right)
		\leq c.$$
		By a similar argument we can see that $x\leq \frac{5}{3}\frac{e+c}{\log(e+c)}\leq \frac{4c}{\log(e+c)}$, since 
		\begin{align*}
			\frac{5}{3}\frac{e+c}{\ln(e+c)}\ln\left(\frac{5}{3}\frac{e+c}{\ln(e+c)}\right)
			&> \frac{5}{3}\frac{e+c}{\ln(e+c)}\ln\left(\frac{e+c}{\ln(e+c)}\right)\\
			&=\frac{5}{3}\frac{e+c}{\ln(e+c)}\left(\ln(e+c)-\ln(\ln(e+c))\right)\\
			&=\frac{5}{3}(e+c)\left(1-\frac{\ln(\ln(e+c))}{\ln(e+c)}\right)\\			
			&\geq\frac{5}{3}(e+c)\left(1-\frac{1}{e}\right)\tag*{$\left(\forall y\in \R_+\colon \frac{\ln(y)}{y}\leq\frac{1}{e}\right)$}\\				
			&> e+c\\
			&> c.
		\end{align*}
		Thus for $x\geq 1 ,\, c\geq e$ we get that the solution of the equation $x\log(x)=c$ satisfies 
		\begin{equation}\label{eq:xcsolution}
			\frac{c}{\log(e+c)}	\leq \frac{c}{\ln(c)}\leq x\leq \frac{4c}{\log(e+c)}.
		\end{equation}
		Using $x=\frac{r}{t}\Rightarrow r= t x$ and $c=\ln\left(\frac{1}{\eps}\right)\frac{1}{t}$ from \eqref{eq:lambertForm}-\eqref{eq:xcsolution} we get that
		\begin{align}
			\forall t\leq \frac{\ln(1/\eps)}{e} \colon \frac{\ln(1/\eps)}{\ln(e+\ln(1/\eps)/t)} \leq r(t,\eps) \leq \frac{4\ln(1/\eps)}{\ln(e+\ln(1/\eps)/t)}.
			\label{eq:smallt}	
		\end{align}
		Combining \eqref{eq:larget} and \eqref{eq:smallt} while observing $t\leq r(t,\eps)$ and $\frac{\ln(1/\eps)}{\ln(e+\ln(1/\eps)/t)}\leq \ln(1/\eps)$, we get that 
		\begin{equation}\label{eq:lagrangesolution}
			\forall \eps\in(0,1) \forall t\in \R_+ \colon r(t,\eps) = \Theta\left(t+ \frac{\ln(1/\eps)}{\ln(e+\ln(1/\eps)/t)} \right).
		\end{equation}		
		Finally note that for $r_q:=e^q t+\ln(1/\eps)/q$, then we get
		\begin{align*}
			\left(\frac{t}{r_1}\right)^{\!\!r_q}
			\leq \left(e^{-q}\right)^{\!r_q}
			\leq e^{-\ln(1/\eps)}
			=\eps \qquad \Longrightarrow\qquad r(t,\eps)\leq r_q.
		\end{align*}
	\end{proof}	

	This enables us to conclude the complexity of block-Hamiltonian simulation. Note that for $t\leq \eps$ Hamiltonian simulation with $\eps$-precision is trivial if $\nrm{H}\leq 1$, therefore we should assume that $t=\Omega(\eps)$ in order to avoid this trivial situation. Apart from this we can conclude the complexity of block-Hamiltonian simulation for entire range of interesting parameters.
	\begin{corollary}[Complexity of block-Hamiltonian simulation]
		Let $\eps\in(0,\frac{1}{2})$, $t\in\R$ and $\alpha\in\R_+$. 
		Let $U$ be an $(\alpha,a,0)$-block-encoding of the unknown Hamiltonian $H$. In order to implement an $\eps$-precise Hamiltonian simulation unitary $V$ which is an $(1,a+2,\eps)$-block-encoding of $e^{itH}$, it is necessary and sufficient to use the unitary $U$ a total number of times $$\Theta\left(\alpha|t|+\frac{\log(1/\eps)}{\log(e+\log(1/\eps)/(\alpha |t|))}\right).$$ 
	\end{corollary}
	\begin{proof}
		The upper bound follows from Theorem~\ref{thm:blockHamSim} and Lemma~\ref{lemma:lamgertR}. The lower bound follows from the argument laid out in \cite{low2016HamSimQSignProc} using Lemma~\ref{lemma:lamgertR}.
	\end{proof}

	Note that the above corollary also covers the range $t\ll 1$, unlike the result of Low and Chuang~\cite{low2016HamSimQubitization} who assumed $t=\Omega(1)$. Also note that this result does not entirely match the complexity stated by Low and Chuang~\cite{low2016HamSimQSignProc,low2016HamSimQubitization}. For example in the case $t=\frac{\log(1/\eps)}{\log(\log(1/\eps))}$ the above corollary shows that the complexity is $\Theta\left(\frac{\log(1/\eps)}{\log(\log(\log(1/\eps)))}\right)$, whereas the expression of \cite{low2016HamSimQSignProc,low2016HamSimQubitization} claims complexity $\bigO{\frac{\log(1/\eps)}{\log(\log(1/\eps))}}$.

	The following lemma of Chakraborty et al.~\cite[Appendix A]{chakraborty2018BlockMatrixPowers} helps us to understand error accumulation in Hamiltonian simulation, which enables us to present a slightly improved claim in Theorem~\ref{cor:blockHamSimRob}.
	\begin{lemma}\label{lemma:hamSimDiff}
		Let $t\in\R$ and $H,H'\in\C^{n\times n}$ Hermitian operators, then 
		$$\nrm{e^{it H}-e^{it H'}}\leq |t|\nrm{H-H'}.$$
	\end{lemma}
	
	Now we prove a robust version of Theorem~\ref{thm:blockHamSim} using Lemma~\ref{lemma:hamSimDiff}, and also substitute a simple expression of Lemma~\ref{lemma:lamgertR} bounding $r(t,\eps)$ in order to get explicit constants.
	\begin{corollary}\emph{(Robust block-Hamiltonian simulation)}\label{cor:blockHamSimRob}
		Let $t\in\R$, $\eps\in(0,1)$ and let $U$ be an $(\alpha,a,\eps/|2t|)$-block-encoding of the Hamiltonian $H$. Then we can implement an $\eps$-precise Hamiltonian simulation unitary $V$ which is an $(1,a+2,\eps)$-block-encoding of $e^{itH}$, with $6\alpha|t|+9\log(12/\eps)$
		uses of $U$ or its inverse, $3$ uses of controlled-$U$ or its inverse, using $\bigO{a\left(\alpha|t|+\log(2/\eps)\right)}$ two-qubit gates and using $\bigO{1}$ ancilla qubits.
	\end{corollary}
	\begin{proof}
		Let $H'=\alpha(\bra{0}^{\otimes a}\otimes I)U(\ket{0}^{\otimes a}\otimes I)$, then $\nrm{H'-H}\leq \eps/|2t|$.
		By Theorem~\ref{thm:blockHamSim} we can implement $V$ an $(1,a+2,\eps/2)$-block-encoding of $e^{itH'}$, with $3r\left(\frac{e\alpha|t|}{2},\frac{\eps}{12}\right)$ uses of $U$ or its inverse, $3$ uses of controlled-$U$ or its inverse and with $\bigO{ar\left(\frac{e\alpha|t|}{2},\frac{\eps}{12}\right)}$ two-qubit gates and using $\bigO{1}$ ancilla qubits.
		By Lemma~\ref{lemma:hamSimDiff} we get that $V$ is an $(1,a+2,\eps)$-block-encoding of $e^{itH}$.
		Finally by Lemma~\ref{lemma:lamgertR} choosing $q:=\frac13$ we get that 
		$$r\left(\frac{e\alpha|t|}{2},\frac{\eps}{12}\right)
		\leq e^q \frac{e\alpha|t|}{2}+\frac{\ln(12/\eps)}{q}
		\leq 2\alpha|t| + 3 \ln(12/\eps).$$
	\end{proof}	
	
	\subsection{Bounded polynomial approximations of piecewise smooth functions}\label{subsec:polyApprox}

	We begin by invoking a slightly surprising result showing how to efficiently approximate monomials on the interval $[-1,1]$ with essentially quadratically smaller degree polynomials than the monomial itself. The following theorem can be found in the survey of Sachdeva and Vishnoi~\cite[Theorem 3.3]{sachdeva2014FasterAlgsViaApxTheory}.
	\begin{theorem}[Efficient approximation of monomials on ${[-1,1]}$]\label{thm:momialApx}
		For any positive integers $s$ and $d$, there exists an efficiently computable degree-$d$ polynomial $P_{s,d}\in\R[x]$ that satisfies
		$$\nrm{P_{s,d}(x) - x^s}_{[-1,1]} \leq 2e^{-d^2/(2s)}.$$
	\end{theorem}
	
	If one wants to approximate smooth functions on the entire $[-1,1]$ interval this result gives essentially quadratic savings. For example one can easily derive Corollary~\ref{cor:exponentialApx} using the above result as shown in~\cite{sachdeva2014FasterAlgsViaApxTheory}.
	
	\begin{corollary}[Polynomial approximations of the exponential function]\label{cor:exponentialApx}
		Let $\beta\in\R_+$ and $\eps\in(0,\frac{1}{2}]$. There exists an efficiently constructable polynomial $P\in\R[x]$ such that 
		$$ \nrm{e^{-\beta(1-x)}-P(x)}_{[-1,1]}\leq \eps, $$
		and the degree of $P$ is $\bigO{\sqrt{\max\left[\beta,\log(\frac1\eps)\right]\log(\frac1\eps)}}$.
	\end{corollary}
	
	However, we often want to implement functions that are smooth only on some compact subset of $C\subseteq[-1,1]$, which requires different techniques. The main difficulty is to achieve a good approximation on $C$, while keeping the norm of the approximating polynomial bounded on the whole $[-1,1]$ interval. We overcome this difficulty by using Fourier approximations on $C$, which give rise to bounded functions naturally. Later we convert these Fourier series to a polynomial using Lemma~\ref{lemma:polyTrig} from the previous subsection.
	
	Apeldoorn et al~\cite[Appendix B]{apeldoorn2017QSDPSolvers} developed techniques that make it possible to implement smooth-functions of a Hamiltonian $H$, based on Fourier series decompositions and using the Linear Combinations of Unitaries (LCU) Lemma~\cite{berry2015HamSimNearlyOpt}. The techniques developed in~\cite[Appendix B]{apeldoorn2017QSDPSolvers} access $H$ only through controlled-Hamiltonian simulation, which step can be omitted using singular value transformation techniques by constructing the corresponding bounded low-degree polynomials.

	Now we invoke one of the main technical results of~\cite[Lemma~37]{apeldoorn2017QSDPSolvers} about approximating smooth functions by low-weight Fourier series. By low weight we mean that the $1$-norm of the coefficients is small. A notable property of the following result is that the bounds on the Fourier series do not depend on the degrees of the polynomials terms. This can however be expected since the terms that have large degree make negligible contribution due to the restricted domain $x\in [-1+\delta, 1-\delta]$, and therefore we can drop them without loss of generality.
	
	\begin{lemma}[Low weight approximation by Fourier series]\label{lemma:LowWeightAPX}
		Let $\delta,\eps\in\!(0,1)$ and $f:\mathbb{R}\rightarrow \mathbb{C}$ s.t. $\left|f(x)\!-\!\sum_{k=0}^K a_k x^k\right|\leq \eps/4$ for all $x\in\![-1+\delta,1-\delta]$.
		Then $\exists\, c\in\mathbb{C}^{2M+1}$ such that 
		$$
		\left|f(x)-\sum_{m=-M}^M c_m e^{\frac{i\pi m}{2}x}\right|\leq \eps
		$$ 
		for all $x\in\![-1+\delta,1-\delta]$, where $M=\max\left(2\left\lceil \ln\left(\frac{4\nrm{a}_1}{\eps}\right)\frac{1}{\delta} \right\rceil,0\right)$ and $\nrm{c}_1\leq \nrm{a}_1$. Moreover $c$ can be efficiently calculated on a classical computer in time $\text{poly}(K,M,\log(1/\eps))$.
	\end{lemma}

	Our main idea is to combine this result with the polynomial approximations of the trigonometric functions as in Lemma~\ref{lemma:polyTrig}. The low weights are useful because they let us reduce the precision required for approximating the Fourier terms.
	
	\begin{corollary}[Bonded polynomial approximations based on a local Taylor series] \label{cor:boundedTaylorApx}
		Let $x_0\in[-1,1]$, $r\in(0,2]$, $\delta\in(0,r]$ and let  $f\colon [-x_0-r-\delta,x_0+r+\delta]\rightarrow \C$ and be such that $f(x_0+x)=\sum_{\ell=0}^{\infty} a_\ell x^\ell$ for all $x\in\![-r-\delta,r+\delta]$. 
		Suppose $B>0$ is such that $\sum_{\ell=0}^{\infty}(r+\delta)^\ell|a_\ell|\leq B$. 
		Let $\eps\in\!\left(0,\frac{1}{2B}\right]$, then there is an efficiently computable polynomial $P\in \C[x]$ of degree $\bigO{\frac{1}{\delta}\log\left(\frac{B}{\eps}\right)}$ such that
		\begin{align}
			\nrm{ f(x)- P(x) }_{[x_0-r,x_0+r]}&\leq \eps\label{eq:truncApx}\\
			\nrm{P(x)}_{[-1,1]}&\leq  \eps + \nrm{f(x)}_{[x_0-r-\delta/2,x_0+r+\delta/2]}\leq \eps + B\\
			\nrm{P(x)}_{[-1,1]\setminus [x_0-r-\delta/2,x_0+r+\delta/2]}&\leq \eps.\label{eq:truncBnd}
		\end{align}
	\end{corollary}
 	\begin{proof}
 		We proceed similarly to the proof of~\cite[Theorem~40]{apeldoorn2017QSDPSolvers}.
		Let $L(x):=\frac{x-x_0}{r+\delta}$ be the linear transformation taking $[x_0-r-\delta,x_0+r+\delta]$ to $[-1,1]$.
 		Let $g(y):=f(L^{-1}(y))$, and $b_\ell:=a_\ell(r+\delta)^\ell$ such that $g(y)=\sum_{\ell=0}^{\infty} b_\ell y^\ell$.
 		Let $\delta':=\frac{\delta}{2(r+\delta)}$ and let $J=\left\lceil\frac{1}{\delta'}\log(\frac{12 B}{\eps})\right\rceil$, then for all $y\in[-1,1]$ we have that 
		\begin{align*}
			\left|g(y)-\sum_{j=0}^{J-1}b_j y^j\right| 
			= \left|\sum_{j=J}^{\infty}b_j y^j\right|
			\leq \sum_{j=J}^{\infty}\left|b_j (1-\delta')^j\right|
			\leq (1-\delta')^J \sum_{j=J}^{\infty}\left|b_j\right|		
			\leq \left(1-\delta'\right)^{J}B	
			\leq e^{-\delta'J}B
			\leq \frac{\eps}{12}.					
		\end{align*}
		Now we construct a Fourier-approximation of $g$ for all $y\in[-1+\delta',1-\delta']$, with precision $\frac{\eps}{3}$. 
		Let $b':=(b_0,b_1,\ldots,b_{J-1})$ and observe that $\nrm{b'}_1\leq\nrm{b}_1\leq B$. We apply Lemma~\ref{lemma:LowWeightAPX} to the function $g$, using the polynomial approximation corresponding to the truncation to the first $J$ terms, i.e., using the coefficients in $b'$. 	
		Then we obtain a Fourier $\frac{\eps}{3}$-approximation $\tilde{g}(y):=\sum_{m=-M}^{M}\tilde{c}_m e^{\frac{i\pi m}{2}y}$ of $g$, with 
		$$
		M=\bigO{\frac{1}{\delta'}\log\left(\frac{\nrm{b'}_1}{\eps'}\right)}=\bigO{\frac{r}{\delta}\log\left(\frac{B}{\eps}\right)},
		$$ 
		such that the vector of coefficients $\tilde{c}\in\mathbb{C}^{2M+1}$ satisfies $\nrm{\tilde{c}}_1\leq\nrm{b'}_1 \leq \nrm{b}_1 \leq B$.   Let
		$$\tilde{f}(x):=\tilde{g}\left(L(x)\right)=\tilde{g}\left(\frac{x-x_0}{r+\delta}\right)=\sum_{m=-M}^{M}\tilde{c}_m e^{\frac{i\pi m}{2(r+\delta)}(x-x_0)}=\sum_{m=-M}^{M}\tilde{c}_m e^{-\frac{i\pi m }{2(r+\delta)}x_0} e^{\frac{i\pi m}{2(r+\delta)}x}.$$ 
		Since $g(y)=f(L^{-1}(y))$ we have that $f(x)=g\left(L(x)\right)$ thus we can see that $\tilde{f}$ is an $\frac{\eps}{3}$-precise Fourier approximation of $f$ on the interval $[x_0-r-\frac{\delta}{2},x_0+r+\frac{\delta}{2}]$. Now we define $\tilde{P}$ as the polynomial that we get by replacing each of the Fourier terms $e^{\frac{i\pi m}{2(r+\delta)}x}$ by $\frac{\eps}{3B}$-approximating polynomials given by Lemma~\ref{lemma:polyTrig}. Using a tiny rescaling we can assure that the polynomial approximations of $e^{\frac{i\pi m}{2(r+\delta)}x}$ have absolute value at most $1$ on $[-1,1]$. Moreover by Lemma~\ref{lemma:lamgertR} we know that the degree of these polynomials are
		$\bigO{\frac{M}{r+\delta}+\log\left(\frac{B}{\eps}\right)}=\bigO{\frac{1}{\delta}\log\left(\frac{B}{\eps}\right)}$.
		Since $\nrm{\tilde{c}}\leq B$, we get that the absolute value of the polynomial $\tilde{P}$ is bounded by $B$ on the interval $[-1,1]$.
		Finally we define $P$ as the product of $\tilde{P}$ and an approximation polynomial of the rectangle function that is $\frac{\eps}{3B}$-close to $1$ on the interval $[x_0-r,x_0+r]$, and is $\frac{\eps}{3B}$-close to $0$ on the interval $[-1,1]\setminus [x_0-r-\frac{\delta}{2},x_0+r+\frac{\delta}{2}]$, finally which is bounded by $1$ on the interval $[-1,1]$ in absolute value. By Lemma~\ref{lemma:polyRect} we can construct such a polynomial of degree $\bigO{\frac{1}{\delta}\log\left(\frac{B}{\eps}\right)}$.
		As we can see $P$ has degree $\bigO{\frac{1}{\delta}\log\left(\frac{B}{\eps}\right)}$, and by construction satisfies the required properties \eqref{eq:truncApx}-\eqref{eq:truncBnd}.
 	\end{proof}

	Combining this polynomial approximation result with Theorem~\ref{thm:arbParity} we can efficiently implement smooth functions of Hermitian matrices. As an application, motivated by the work of Chakraborty et al.~\cite{chakraborty2018BlockMatrixPowers} we show how to construct low-degree polynomial approximations of power functions.
	\begin{corollary}[Polynomial approximations of negative power functions]\label{cor:negatiwePower}
		Let $\delta,\eps\in(0,\frac{1}{2}]$, $c>0$ and let $f(x):=\frac{\delta^c}{2}x^{-c}$, then there exist even/odd polynomials $P,P'\in\R[x]$, such that $\nrm{P-f}_{[\delta,1]}\leq\eps$, $\nrm{P}_{[-1,1]}\leq1$ and similarly $\nrm{P'-f}_{[\delta,1]}\leq\eps$, $\nrm{P'}_{[-1,1]}\leq1$, moreover the degree of the polynomials are $\bigO{\frac{\max[1,c]}{\delta}\log\left(\frac{1}{\eps}\right)}$.
	\end{corollary}
	\begin{proof}
		First note that for all $y\in (-1,1)$ we have that $(1+y)^{-c}=\sum_{k=0}^\infty \binom{-c}{k}y^k$. We first find a polynomial $\tilde{P}\in\C[x]$ such that $\nrm{\tilde{P}-f}_{[\delta,1]}\leq \frac{\eps}{2}$, $\nrm{\tilde{P}}_{[-1,0]}\leq \frac{\eps}{2}$ and $\nrm{\tilde{P}}_{[-1,1]}\leq 1$. We construct such a polynomial of degree $\bigO{\frac{\max[1,c]}{\delta}\log\left(\frac{1}{\eps}\right)}$ using Corollary~\ref{cor:boundedTaylorApx}, with choosing $x_0:=0$, $r:=1-\delta$, $\delta':=\frac{\delta}{2\max[1,c]}$ and $B:=1$. The choice of $B$ is justified by the observation that 
		\begin{align*}
			\frac{\delta^{c}}{2}\sum_{k=0}^\infty\left|\binom{-c}{k}\right|(r+\delta')^k
			&=\frac{\delta^{c}}{2}\sum_{k=0}^\infty\binom{-c}{k}(-r-\delta')^k
			=\frac{\delta^{c}}{2}(1-r-\delta')^{-c}\\			
			&=\frac{\delta^{c}}{2}(\delta-\delta')^{-c}	
			=\frac{1}{2}\left(1-\frac{\delta'}{\delta}\right)^{\!\!-c}\\		
			&=\frac{1}{2}\left(1-\frac{1}{2\max[1,c]}\right)^{\!\!-c}
			\leq 1.\\	
		\end{align*}
		Finally, we define $P$ as the even real part of $\tilde{P}$, and define $P'$ as the odd real part of $\tilde{P}$.
	\end{proof}

	Given a $(1,a,0)$-block-encoding of $A$, with the promise that the spectrum of $A$ lies in $[\delta,1]$, using the above polynomials and Theorem~\ref{thm:arbParity} we can implement a $(1,a+2,\eps)$-block-encoding of $\frac{\delta^c}{2}A^{-c}$ with $\bigO{\frac{\max[1,c]}{\delta}\log\left(\frac{1}{\eps}\right)}$ uses of the block-encoding of $A$. Since the derivative of the function $\frac{\delta^c}{2}x^{-c}$ at $x=\delta$ is $-\frac{c}{2\delta}$, we get by Theorem~\ref{thm:lowerBoundEVT} that the $\delta$ and $c$ dependence of the complexity of this procedure is optimal.

	Finally we develop a theorem that is analogous to  \cite[Corollary~42]{apeldoorn2017QSDPSolvers}, and shows that any function that has quickly converging local Taylor-series can in principle be $\eps$-approximated with complexity $\propto \log\left(\frac{1}{\eps}\right)$.

	\begin{theorem}[Bounded polynomial approximation based on multiple local Taylor series]\label{thm:boundedTaylorApx}
		Let $J\in \N$, $(x_j,r_j,\delta_j)\in[-1,1]^J\times(0,2]^J\times(0,1]^J$, such that $x_j\colon j\in [J]$ is monotone increasing, and $\delta_j\leq r_j$ for all $j\in[J]$. Let $I:=\bigcup_{j\in [J]}[x_j-r_j,x_j+r_j]$ be the union of the intervals $[x_j-r_j,x_j+r_j]$, and suppose that for all $i<j\in[J]$ such that $j-i\geq 2$ we have that $r_j+r_j< x_j-x_j$. Let $\delta=\min\left[\min_{j\in [J]} \delta_j,\min_{j\in [J-1]} |x_{j+1}-x_{j} - (r_{j+1} + r_j )|\right]$. Let $f:I+[-\frac{\delta}2,\frac{\delta}2]\rightarrow \C$, $B\in \R_+$ be such that for all $j\in[J]$ we have $f(x_j+x)=\sum_{k=0}^{\infty}a^{(j)}_k x^k$ for all $x\in [x_j-r_j-\frac{\delta_j}2,x_j+r_j+\frac{\delta_j}2]$ and $\sum_{k=0}^{\infty}(r_j+\delta_j)^k |a^{(j)}_k|\leq B$. Let $\eps\in\!\left(0,\frac{1}{2BJ}\right]$, then there is an efficiently computable polynomial $P\in \C[x]$ of degree $\bigO{\frac{J}{\delta}\log\left(\frac{BJ}{\eps}\right)}$ such that
		\begin{align*}
		\nrm{ f(x)- P(x) }_{I}&\leq \eps\\
		\nrm{P(x)}_{[-1,1]}&\leq  \nrm{f(x)}_{I+[-\delta/2,\delta/2]}\\
		\nrm{P(x)}_{[-1,1]\setminus \left(I+[-\delta/2,\delta/2]\right)}&\leq \eps.
		\end{align*}
	\end{theorem}
	\begin{proof}
		Use Corollary \ref{cor:boundedTaylorApx} to construct polynomials $f_j \colon j\in [J]$ of degree $\bigO{\frac{1}{\delta}\log\left(\frac{BJ}{\eps}\right)}$ such that 
		\begin{align*}
			\nrm{ f(x)- f_j(x) }_{[x_j-r_j,x_j+r_j]}&\leq \frac{\eps}{4J}\\
			\nrm{f_j(x)}_{[-1,1]}&\leq   \nrm{f(x)}_{I+[-\delta/2,\delta/2]}\\
			\nrm{f_j(x)}_{[-1,1]\setminus \left([x_j-r_j,x_j+r_j]+[-\delta/2,\delta/2]\right)}&\leq \eps.
		\end{align*}
		Let us introduce a notation for the union of the intervals $[x_i-r_i,x_i+r_i]\colon i\in \{j,j+1,\ldots,k\}$ as  
		$$I_{[j,k]}:=\bigcup_{i\in \{j,j+1,\ldots,k\}}[x_i-r_i,x_i+r_i].$$
		We show inductively how to construct polynomials $f_{[j,k]}$  of degree $\bigO{\frac{k-j+1}{\delta}\log\left(\frac{BJ}{\eps}\right)}$
		such that
		\begin{align}
			\nrm{ f(x)- f_{[j,k]}(x) }_{I_{[j,k]}}&\leq \frac{2(k-j+1)\eps}{2J}\label{eq:complP1}\\
			\nrm{f_{[j,k]}(x)}_{[-1,1]}&\leq  \nrm{f(x)}_{I+[-\delta/2,\delta/2]}\label{eq:complP2}\\
			\nrm{f_{[j,k]}(x)}_{[-1,1]\setminus \left(I_{[j,k]}+[-\delta/2,\delta/2]\right)}&\leq \eps.\label{eq:complP3}
		\end{align}
		We already showed how to construct $f_{[j,j]}:=f_j\colon j\in[J]$. Suppose that we already constructed $f_{[1,j]}$, then we construct  $f_{[1,j+1]}$ as follows. We take a polynomial $S(x)$ of degree $\bigO{\frac{1}{\delta}\log\left(\frac{BJ}{\eps}\right)}$ that approximates the shifted sign function s.t.  $\nrm{S(x)-\sign{x-\frac{x_i+x_j}{2}}}_{[-1,1]\setminus \left[\frac{x_i+x_j-\delta}{2},\frac{x_i+x_j+\delta}{2}\right]}\leq \frac{\eps}{8BJ}$, moreover $\nrm{S}_{[-1,1]}\leq 1$. Then we define $f_{[1,j+1]}:=\frac{1-S(x)}{2}f_{[1,j]}+\frac{1+S(x)}{2}f_{[j+1,j+1]}$. It satisfies \eqref{eq:complP2}-\eqref{eq:complP3}, since $f_{[1,j+1]}$ is a point-wise convex combination of $f_{[1,j]}$ and $f_{[j+1,j+1]}$. Similarly \eqref{eq:complP1} is also easy to verify. Therefore by induction we can finally construct $P:=f_{[1,J]}$, which satisfies \eqref{eq:complP1}-\eqref{eq:complP3} and therefore also the requirements of the theorem.\footnote{Note that this approach could be further improved to produce a degree $\bigO{\frac{\log(J)}{\delta}\log\left(\frac{B\log(J)}{\eps}\right)}$ approximating polynomial by combining the polynomial approximations on the different intervals in a binary tree structure. Since $J=\bigO{\frac{1}{\delta}}$, $\log(J)=\log\left(\frac{1}{\delta}\right)$ and then this gives at most a logarithmic overhead.}
	\end{proof}
	
	A direct corollary of this theorem is for example that for all $\eps\in(0,\frac12]$ the function $\frac{\delta}{x}$ can be $\eps$-approximated on the domain $[-1,1]\setminus [-\delta,\delta]$ with a polynomial of degree $\bigO{\frac{1}{\delta}\log\left(\frac{1}{\eps}\right)}$. Although this also follows from Corollary~\ref{cor:negatiwePower}, we prove it directly using Theorem~\ref{thm:boundedTaylorApx}, in order to illustrate its usefulness.
	
	\begin{corollary}[Polynomial approximations of $\frac{1}{x}$]\label{cor:oneOverX}
		Let $\eps,\delta \in(0,\frac12]$, then there is an odd polynomial $P\in\R[x]$ of degree $\bigO{\frac{1}{\delta}\log\left(\frac{1}{\eps}\right)}$ that is $\eps$-approximating $f(x)=\frac{3}{4}\frac{\delta}{x}$ on the domain $I=[-1,1]\setminus[-\delta,\delta]$, moreover it is bounded $1$ in absolute value.
	\end{corollary}	
	\begin{proof}
		Take $J:=2$, $(x_1:=-1,r_1:=1-\delta,\delta_1:=\frac{\delta}{2})$, $(x_2:=1,r_2:=1-\delta,\delta_2:=\frac{\delta}{2})$ and $B=1$ in Theorem~\ref{thm:boundedTaylorApx}, observing that $f(1+x)=\frac{3\delta}{4}\sum_{k=0}^{\infty}-(1)^k x^k=-f(-1+x)$. Define $P$ as the odd real part of the polynomial given by Theorem~\ref{thm:boundedTaylorApx}.
	\end{proof}

	\subsection{Applications: fractional queries and Gibbs sampling}
	\label{subsec:applications}
	Scott Aaronson listed as one of ``The ten most annoying questions in quantum computing''~\cite{aaronson2006TenMostAnnoying} the following problem: given a unitary $U$, can we implement $\sqrt{U}$? This was positively answered by Sheridan et al.~\cite{sheridan2008ApxFractTimeQEvol}. We show how to improve the complexity of the result of Sheridan et al. exponentially in terms of the error dependence. We proceed follow ideas of Low and Chuang~\cite{low2017HamSimUnifAmp}.
	
	Suppose that we have access to a unitary $U=e^{iH}$, where $H$ is a Hamiltonian of norm at most $\frac{1}{2}$. Low and Chuang~\cite{low2017HamSimUnifAmp} showed how to get a $(1,2,\eps)$-block-encoding of $H$ with $\bigO{\log\left(\frac{1}{\eps}\right)}$ uses of $U$. We reprove this result; our proof becomes quite simple thanks to Corollary~\ref{cor:boundedTaylorApx}.
	
	\begin{lemma}[Polynomial approximations of $\arcsin(x)$] 
		Let $\delta,\eps\in(0,\frac12]$, there is an efficiently computable odd real polynomial $P\in\R[x]$ of degree $\bigO{\frac{1}{\delta}\log\left(\frac{1}{\eps}\right)}$ such that $\nrm{P}_{[-1,1]}\leq 1$ and 
		$$\nrm{P(x)-\frac{2}{\pi}\arcsin(x)}_{[-1+\delta,1-\delta]}\leq \eps.$$
	\end{lemma}
	\begin{proof}
		Observe that $\frac{2}{\pi}\arcsin(x)=\sum_{\ell=0}^{\infty}\binom{2\ell}{\ell}\frac{2^{-2\ell}}{2\ell+1}\frac{2}{\pi}x^{2\ell+1}$ for all $x\in[-1,1]$. Therefore we also have $\sum_{\ell=0}^{\infty}\left|\binom{2\ell}{\ell}\frac{2^{-2\ell}}{2\ell+1}\frac{2}{\pi}\right|=1$.
		The result immediately follows by taking the odd real part of the polynomial given by Corollary~\ref{cor:boundedTaylorApx}.
	\end{proof}

	\begin{corollary}[Implementing the logarithm of unitaries]\label{cor:logarithm}
		Suppose that $U=e^{iH}$, where $H$ is a Hamiltonian of norm at most $\frac{1}{2}$. Let $\eps\in(0,\frac12]$, then we can implement a $(\frac{2}{\pi},2,\eps)$-block-encoding of $H$ with $\bigO{\log\left(\frac{1}{\eps}\right)}$ uses of controlled-$U$ and its inverse, using $\bigO{\log\left(\frac{1}{\eps}\right)}$ two-qubit gates and using a single ancilla qubit.
	\end{corollary}
	\begin{proof}
		Let $cU$ denote the controlled version of $U$ controlled by the first qubit. Then $$\sin(H)=-i(\bra{+}\otimes I) cU^\dagger\left(ZX\otimes I\right)cU (\ket{+}\otimes I).$$
		Now we apply singular value transformation Corollary~\ref{cor:matchingParity} using an $\eps$-approximating polynomial of $\frac{2}{\pi}\arcsin(x)$ on the domain $[-\frac{1}{2},\frac{1}{2}]$.
	\end{proof}
	Combining the above result with block-Hamiltonian simulation techniques Corollary~\ref{cor:blockHamSimRob} we can implement fractional queries of unitaries with complexity $\bigO{\log^2\left(\frac{1}{\eps}\right)}$. As we show in the following corollary this complexity can be reduced to $\bigO{\log\left(\frac{1}{\eps}\right)}$ by directly implementing\footnote{The method we describe uses block-encoding formalism, but in fact one could implement it more directly using a Fourier series based approach similarly to the one used for Hamiltonian simulation by Low and Chuang~\cite{low2016HamSimQSignProc}.} Hamiltonian simulation using a block-encoding of $\sin(H)$ rather than $H$. 

	\begin{corollary}[Implementing fractional queries]\label{cor:fractionalQuery}
		Suppose that $U=e^{iH}$, where $H$ is a Hamiltonian of norm at most $\frac{1}{2}$. Let $\eps\in(0,\frac12]$ and $t\in[-1,1]$, then we can implement an $\eps$-approximation of $U^t=e^{itH}$ with $\bigO{\log\left(\frac{1}{\eps}\right)}$ uses of controlled-$U$ and its inverse, using $\bigO{\log\left(\frac{1}{\eps}\right)}$ two-qubit gates and using $\bigO{1}$ ancilla qubits.
	\end{corollary}
	\begin{proof}
		As we have shown in the proof of Corollary~\ref{cor:logarithm}, one can implement a block-encoding of $\sin(H)$ with a constant number of queries to $U$.
		Let us look at the Taylor series of $e^{it\arcsin(x)}$. One can see that the $1$-norm of the coefficients of the Taylor series of $t\arcsin(x)$ is $|t|\arcsin(1)=|t|\frac{\pi}{2}$. Therefore, for $t\in[-\frac{2}{\pi},\frac{2}{\pi}]$ we get that the $1$-norm of the Taylor series of $e^{it\arcsin(x)}$ is at most $e^1=e$. Thereby, using Theorem~\ref{thm:boundedTaylorApx} we can construct polynomial $\bigO{\eps}$-approximations of $\sin(t\arcsin(x))$ and $\cos(t\arcsin(x))$ of degree $\bigO{\log\left(\frac{1}{\eps}\right)}$, which are bounded by $1$ in absolute value on the interval $[-1,1]$. We can combine these polynomials in a similar way as in the proof of Theorem~\ref{thm:blockHamSim}. This way we can implement an $\eps$-approximation of $U^t$ for all $t\in[-\frac{2}{\pi},\frac{2}{\pi}]$ with complexity $\bigO{\log\left(\frac{1}{\eps}\right)}$. Implementing $U^t$ for all $t\in[-1,1]$ can be achieved by implementing $U^{\frac{t}{2}}$ twice an taking their product.
	\end{proof}
	
	Note that the above technique can be combined with an initial phase estimation in order to implement fractional queries under the weaker promise $\nrm{H}\leq \pi-\delta$. It suffices to perform a $\delta$-precise phase estimation with success probability $1-\poly{\eps}$, then implement fractional queries using Corollary~\ref{cor:fractionalQuery} and then undo the initial phase estimation. This leads to complexity $\bigO{\frac{1}{\delta}\log\left(\frac{1}{\eps}\right)}$, which exponentially improves the complexity $\bigO{\max[\frac{1}{\delta},\frac{1}{\eps}]\log\left(\frac{1}{\eps}\right)}$ of Sheridan et al.~\cite{sheridan2008ApxFractTimeQEvol} in the case of $\delta=\Theta(1)$. Note that Sheridan et al.~\cite{sheridan2008ApxFractTimeQEvol} also proved a lower bound on this problem, which shows that the $\delta$ dependence of this algorithm is actually optimal. We believe that the $\log(\frac{1}{\eps})$ dependence in the runtime is also necessary, therefore this algorithm is probably fairly close to optimal.
	
	After implementing a fractional query, such that $\nrm{H}\leq\frac12$ is satisfied, one can use Corollary~\ref{cor:fractionalQuery} to implement the logarithm of the unitary. Also note the gap promise on the spectrum of $U$ is necessary for implementing the fractional queries, but it is not important that the gap is exactly at $e^{i\pi}$, one can just add a phase gate for example to $U$ in order to rotate the spectrum. 
	
	Finally, we briefly describe how these techniques can be used for Gibbs sampling.  If one first prepares a maximally entangled state on two registers and applies the map $e^{-\frac{\beta}{2} (H+I)}$ on the first register, then one gets a subnormalized Gibbs state $e^{\beta (H+I)}$ on the first register.  Then, using (fixed-point) amplitude amplification one gets a purification of a Gibbs-state.  Each of these steps can be compactly performed using singular value transformation techniques, providing an efficient implementation. 
	
	An $\eps$-approximation of the map $e^{-\frac{\beta}{2} (H+I)}$ can be implemented using Theorem~\ref{thm:arbParity} and Corollary~\ref{cor:exponentialApx} with query complexity $\bigO{\sqrt{\beta}\log{(1/\epsilon)}}$, and it suffices to use $\bigO{\sqrt{\frac{n}{\mathcal{Z}}}}$ amplitude amplification steps in order to prepare the Gibbs-state with constant success probability, where $n$ is dimension of $H$ and $\mathcal{Z}:=\tr{e^{-\beta (H+I)}}$ is the partition function.
	In the case when $H$ does not have an eigenvalue close to $-1$, but say $\lambda_{\min}$ is the smallest eigenvalue, then one should implement an approximation of $e^{-\beta (H - \lambda_{\min} I)}$ on the domain $[\lambda_{\min},1]$ in order to avoid unnecessary subnormalization.  However note, that this in general increases to complexity and gives a linear dependence on $\beta$.  For more details see, e.g., the work of Appeldoorn et al.~\cite{apeldoorn2017QSDPSolvers,apeldoorn2018ImprovedQSDPSolving}.
    
    In the special case when one has access to the square root of $H$, and $H$ has an eigenvalue close to $0$, then one can still achieve quadratically improved scaling with $\beta$ as shown by Chowdhury and Somma~\cite{chowdhury2016QGibbsSampling}. This can be easily shown using our techniques observing that $e^{-\beta H}=e^{-\beta \left(\sqrt{H}\right)^2}\!,$ and that the function $e^{-\beta x^2}$ can be $\eps$-approximated on the interval $[0,1]$ using a polynomial of degree $\bigO{\sqrt{\beta}\log\left(\frac{1}{\eps}\right)}$ as follows from Theorem~\ref{thm:momialApx} or Corollary~\ref{cor:exponentialApx}.

	\section{Limitations of the smooth function techniques}
	\label{sec:lowerBd}
	In the classical literature there are many good techniques for lower bounding the degrees of approximation polynomials~\cite{sachdeva2014FasterAlgsViaApxTheory}. There is a intimate relationship between the degrees of approximation polynomials quantum query complexity~\cite{beals2001QLowerBoundPoly}. In a recent result Arunachalam et al.~\cite{arunachalam2017QuQueryAlgComplBond} showed that for discrete problems certain polynomial approximations characterize the quantum query complexity. 
	There are also some result about lower bounds for continuous problems~\cite{aaronson2009QCopyProt,belovs2015GeneralAdv,gilyen2017OptQOptAlgGrad}, however the literature to this end is much more sparse.
	
	To advance the knowledge on lower bounds in the continuous regime, we prove a conceptually simple lower bound on eigenvalue transformations, which guides our intuition about what sort of transformations are possible. Intuitively speaking if a function has derivative $d$ on the domain of interest then we need to use the block-encoding $\Omega(d)$-times in order to implement the eigenvalue transformation corresponding to $f$. This suggests that Theorem~\ref{thm:boundedTaylorApx} applied together with Theorem~\ref{thm:arbParity} often gives optimal results, since the $\delta$ parameter usually turns out to be $\propto \frac{1}{d}$, where $d$ is the maximal derivative of the function on the domain of interest.
	
	\begin{theorem}[Lower bound for eigenvalue transformation]\label{thm:lowerBoundEVT}
		Let $I\subseteq[-1,1]$, $a\geq 1$ and suppose $U$ is a $(1,a,0)$-block-encoding of an unknown Hermitian matrix $H$ with the only promise that the spectrum of $H$ lies in $I$. Let $f:I\rightarrow \R$, and suppose that we have a quantum circuit $V$ that implements a $(1,b,\eps)$-block-encoding of $f(H)$ using $T$ applications of $U$, for all $U$ fulfilling the promise. Then for all $x\neq y\in I\cap[-\frac12,\frac12]$ we have that 
		$$
			T=\Omega\left(\frac{|f(x)-f(y)|-2\eps}{|x-y|}\right).
		$$
		More precisely for all $x,y\in I$ we have that 
		\begin{align}
			T&\geq \frac{\max\left[f(x)-f(y)-2\eps,\sqrt{1-(f(y)-\eps)^2}-\sqrt{1-(f(x)+\eps)^2}\right]}{\sqrt{2} \sqrt{1-x y-\sqrt{\left(1-x^2\right) \left(1-y^2\right)}}}\label{eq:preciseQueryBound}\\
			&\geq \frac{\max\left[f(x)-f(y)-2\eps,\sqrt{1-(f(y)-\eps)^2}-\sqrt{1-(f(x)+\eps)^2}\right]}{\sqrt{2}\max\left[|x-y|,\left|\sqrt{1-x^2}-\sqrt{1-y^2}\right|\right]}.\label{eq:lessPreciseQueryBound}
		\end{align}		
	\end{theorem}
	\begin{proof}
		First let us examine the case when $H$ is a $d\times d$ matrix, $a=1$ and $U$ is of size $2d\times 2d$.
		Recall that in \eqref{eq:2DReflection} we defined the two-dimensional reflection operator
		$$R(x)=\left[\begin{array}{cc} x  & \sqrt{1-x^2} \\ \sqrt{1-x^2} & -x\end{array}\right],$$
		and note, that for all $x,y\in[0,1]$ we have that 
		\begin{equation}\label{eq:normForm}
			\nrm{R(x)-R(y)}=\sqrt{2} \sqrt{1-x y-\sqrt{\left(1-x^2\right) \left(1-y^2\right)}}\leq \sqrt{2}\max\left[|x-y|,\left|\sqrt{1-x^2}-\sqrt{1-y^2}\right|\right].
		\end{equation}
		For all $z\in[0,1]$ let $U_z:=\bigoplus_{i=1}^d R(z)$, where the direct sum structure is arranged in such a way that $U_z$ is a $(1,1,0)$-block-encoding of $H_z:=z I$. Let $V[U_z]$ denote the circuit $V$ when using the input unitary $U_z$.
		Since $V[U_z]$ uses $U_z$ a total number of $T$ times we have that 
		\begin{equation}\label{eq:VupperBound}
			\nrm{V[U_x]-V[U_y]}\leq T \nrm{U_x-U_y}= T\nrm{R(x)-R(y)}.
		\end{equation}
		By the promise on $V$ we get that $V[U_z]$ is a $(1,b,\eps)$-block-encoding of $f(H_z)=f(z)I$. Let $\varsigma_{\max/\min}$ denote the maximal/minimal singular value of a matrix. Using this notation we get that 
		\begin{align}
			\varsigma_{\max}\left[(\bra{0}^{\otimes b}\!\otimes\! I)V[U_y](\ket{0}^{\!\otimes b}\!\otimes\! I)\right]&\leq f(y)+\eps,\label{eq:VyupperBound}\\
			\varsigma_{\min}\left[(\bra{0}^{\otimes b}\!\otimes\! I)V[U_x](\ket{0}^{\!\otimes b}\!\otimes\! I)\right]&\geq f(x)-\eps.\label{eq:VxupperBound}		
		\end{align}
		Let use introduce the notation $\Pi_{\overline{\ketbra{0}{0}}}:=\left(\left(I_b-\ketbra{0}{0}^{\otimes b}\right)\otimes I\right)$, then by \eqref{eq:VyupperBound}-\eqref{eq:VxupperBound} we have that	
		\begin{align}
			\nrm{V[U_x]-V[U_y]}&\geq \nrm{(\ketbra{0}{0}^{\otimes b}\otimes I)V[U_x](\ketbra{0}{0}^{\otimes b}\otimes I)-(\ketbra{0}{0}^{\otimes b}\otimes I)V[U_y](\ketbra{0}{0}^{\otimes b}\otimes I)}\nonumber\\
			&\geq \varsigma_{\min}\left[(\bra{0}^{\otimes b}\!\otimes\! I)V[U_x](\ket{0}^{\!\otimes b}\!\otimes\! I)\right]-\varsigma_{\max}\left[(\bra{0}^{\otimes b}\!\otimes\! I)V[U_y](\ket{0}^{\!\otimes b}\!\otimes\! I)\right]\nonumber\\
			&\geq f(x)-f(y)-2\eps, \text{ and }\label{eq:VlowerBound1}\\
			\nrm{V[U_y]-V[U_x]}&\geq \nrm{\Pi_{\overline{\ketbra{0}{0}}}V[U_y](\ketbra{0}{0}^{\otimes b}\otimes I)-\Pi_{\overline{\ketbra{0}{0}}}V[U_x](\ketbra{0}{0}^{\otimes b}\otimes I)}\nonumber\\
			&\geq \varsigma_{\min}\left[\Pi_{\overline{\ketbra{0}{0}}}V[U_y](\ket{0}^{\!\otimes b}\!\otimes\! I)\right]-\varsigma_{\max}\left[\Pi_{\overline{\ketbra{0}{0}}}V[U_x](\ket{0}^{\!\otimes b}\!\otimes\! I)\right]\nonumber\\		
			&= \sqrt{1 \!-\!\varsigma^2_{\max}\left[(\bra{0}^{\otimes b}\!\otimes\! I)V[U_y](\ket{0}^{\!\otimes b}\!\otimes\! I)\right]}-\sqrt{1\!-\!\varsigma^2_{\min}\left[(\bra{0}^{\otimes b}\!\otimes\! I)V[U_x](\ket{0}^{\!\otimes b}\!\otimes\! I)\right]}\nonumber\\						
			&\geq \sqrt{1-(f(y)-\eps)^2}-\sqrt{1-(f(x)+\eps)^2}.\label{eq:VlowerBound2}	
		\end{align}
		By combining \eqref{eq:VupperBound} and \eqref{eq:VlowerBound1}-\eqref{eq:VlowerBound2} we get that
		\begin{equation*}
			T\geq \frac{\max\left[f(x)-f(y)-2\eps,\sqrt{1-(f(y)-\eps)^2}-\sqrt{1-(f(x)+\eps)^2}\right]}{\nrm{R(x)-R(y)}}.
		\end{equation*}
		Combining this inequality with \eqref{eq:normForm} proves \eqref{eq:preciseQueryBound}-\eqref{eq:lessPreciseQueryBound}.
		Finally note, that if $a>1$ essentially the same argument can be used to prove \eqref{eq:preciseQueryBound}-\eqref{eq:lessPreciseQueryBound}, just one needs to define $U_z$ with additional tensor products of identity matrices acting on the new ancilla qubits.
	\end{proof}
	
 	The above lower bound suggests that the spectrum of $H$ lying closea to $1$ is more flexible than the spectrum lying below say $\frac{1}{2}$ in absolute value. Indeed it turns out that the spectrum of $H$ lying close to $1$ is quadratically more useful than the spectrum $I\subseteq[-\frac12,\frac12]$, cf. Corollary~\ref{cor:exponentialApx} and Lemma~\ref{lem:PolyNormDiff}.
 	This lower bound also explains why is it so difficult to amplify the spectrum close to $1$, cf. Theorem~\ref{thm:singularValAmp}.
 	Finally note, that since eigenvalue transformation is a special case of singular value transformation it also gives a lower bound in singular value transformation. 
 	
 	\section{Conclusion}
 	Our main contribution in this paper is to provide a paradigm that unifies a host of quantum algorithms ranging from singular value estimation, linear equation solving, quantum simulation to quantum walks.  Prior to our contribution each of these fields would have to be understood independently, which makes mastering all of them a challenge.  By presenting them all within the framework of quantum singular value transformation, many of the most popular techniques in these fields follow as a direct consequence.  This greatly simplifies the learning process while also revealing algorithms that were hitherto unknown.
 	
 	The main result of this paper is an efficient method for implement singular value transformation, extending earlier qubitization techniques.
 	The paper describes several novel applications to this general result, including an algorithm for performing certain ``non-commutative'' measurements, an exponentially improved algorithm for simulating fractional queries to an unknown unitary oracle, and an improved algorithm for principal component regression.  
 	
 	We also give a novel view on quantum matrix arithmetics by summarizing known results about block-encoded matrices, showing that they enable performing matrix arithmetic on quantum computers in a simple and efficient manner.  The described method in principle can give exponential savings in terms of the dimension of the matrices, and perfectly fits into our framework.

 	An interesting question for future work involves the recent work by Catalin Dohotaru and Peter H\o yer which shows that a wide range of quantum walk algorithms can be unified within a single paradigm called controlled quantum amplification~\cite{dohotaru2017controlledQAmp}.  While the structure of the quantum circuits introduced by them bears a strong resemblance to those used in qubitization, it is difficult to place this work within the framework we present here.  The question of how to unify their approach with our techniques therefore remains open.

	\subsection*{Acknowledgments}
	A.G. thanks Ronald de Wolf, Robin Kothari, Joran van Apeldoorn, Shantanav Chakraborty, Stacey Jeffrey, Vedran Dunjko and Yimin Ge for inspiring discussions. Y.S. was supported in part by the Army Research Office (MURI award W911NF-16-1-0349); the U.S. Department of Energy, Office of Science, Office of Advanced Scientific Computing Research, Quantum Algorithms Teams program; and the National Science Foundation (grant 1526380). He thanks Andrew Childs, Guoming Wang, Cedric Lin, John Watrous, Ben Reichardt, Guojing Tian and Aaron Ostrander for helpful discussions.

    \bibliographystyle{alphaUrlePrint}
	\bibliography{Bibliography}
	
	\appendix

\providecommand\mywordcount{
	\ifcount
	$ \phantom{\sum}$ \\ \noindent\textbf{\large Wordcount} \\ $\phantom{\sum}$ \\
	\noindent\input|"if [ -f /home/gilyen/texcount.pl ]; then /home/gilyen/texcount.pl -sub=section \jobname.tex | grep -e Words -e Number -e Section -e top -e Part | awk 1 ORS='\string\\\string\\' | sed -e 's/\string\_/ /g'; else texcount -sub=section \jobname.tex | grep -e Words -e Number -e Section -e top -e Part | awk 1 ORS='\string\\\string\\' | sed -e 's/\string\_/ /g'; fi"
	text+headers+captions (\#headers/\#floats/\#inlines/\#displayed)\\
	\else
	\fi
}
	
	
\end{document}